\documentclass[%
 reprint,
 superscriptaddress,
 amsmath,amssymb,
 aps,
]{revtex4-2}

\usepackage{amsmath}
\usepackage{amsthm}
\usepackage{quantikz}
\usepackage{graphicx}
\usepackage[svgnames]{xcolor}
\usepackage{dcolumn}
\usepackage{bm}
\usepackage{braket}
\usepackage{chemformula}
\usepackage{afterpage}
\usepackage{placeins}

\let\ce\ch

\usepackage[
     colorlinks        = true,
     unicode           = true,
     pdfstartview      = FitV,
     linktocpage       = true,
     linkcolor         = RoyalBlue,
     citecolor         = RoyalBlue,
     urlcolor          = RoyalBlue,
     bookmarks         = true,
     bookmarksnumbered = true,
     breaklinks=true,
     pdftitle={mirror_subspace_diagonalization},
     pdfauthor={}
]{hyperref}
\newcommand{\qunasys}{QunaSys Inc., Aqua Hakusan Building 9F, 1-13-7 Hakusan, Bunkyo, Tokyo 113-0001, Japan}
\newcommand{\fujitsu}{Quantum Laboratory, Fujitsu Research, Fujitsu Limited, 4-1-1 Kamikodanaka, Nakahara, Kawasaki, Kanagawa 211-8588, Japan}
\newcommand{\abs}[1]{\left| #1 \right|}
\newcommand{\order}[1]{\mathcal{O}\left( #1 \right)}

\newcommand{\norm}[1]{\| #1 \|}

\newcommand{\1}{\mbox{1}\hspace{-0.25em}\mbox{l}}
\newcommand{\sgn}[1]{\mathrm{sgn}\left( #1 \right)}

\theoremstyle{plain}
\newtheorem{dfn}{Definition}

\newtheorem{lem}[dfn]{Lemma}
\newtheorem{thm}[dfn]{Theorem}

\begin{document}

\preprint{APS/123-QED}

\title{Mirror subspace diagonalization: A real-time quantum Krylov algorithm with near-optimal sampling cost}

\author{Shota Kanasugi}
 \email{kanasugi.shota@fujitsu.com}
 \affiliation{\fujitsu}

\author{Yuya O. Nakagawa}%
 \affiliation{\qunasys}

\author{Norifumi Matsumoto}
 \affiliation{\fujitsu}

\author{Yuichiro Hidaka}%
 \affiliation{\qunasys}

\author{Kazunori Maruyama}
 \affiliation{\fujitsu}

\author{Hirotaka Oshima}
 \affiliation{\fujitsu}

\date{\today}

\begin{abstract}
\textbf{Abstract}\par\medskip
Quantum Krylov algorithms are a promising approach for ground-state energy estimation on near-term quantum computers. However, their sampling cost is inherently high because each Hamiltonian term is measured individually. Despite various mitigation techniques, this overhead remains a bottleneck for large-scale electronic structure problems. We introduce mirror subspace diagonalization (MSD), which approaches the theoretical lower bound on the sampling cost for quantum Krylov algorithms built upon real-time evolution. MSD uses a finite-difference formula to express the Hamiltonian as a linear combination of time-evolution unitaries with symmetrically shifted timesteps, enabling efficient estimation of the Hamiltonian matrix within the Krylov subspace. By optimizing the timestep and shifting the energy spectrum, finite-difference and statistical errors are minimized simultaneously, allowing MSD to attain the sampling-cost lower bound up to a logarithmic factor. We further apply classical post-processing to infer Hamiltonian moments, which mitigate the ground-state energy error via the Lanczos scheme. Our analysis shows MSD is particularly effective when the Hamiltonian's spectral norm is much smaller than its 1-norm, as in high-accuracy molecular simulations with large basis sets and strong electronic correlations. 
Numerical results for various molecular models show sampling-cost reductions of roughly 10 to 10,000 times over the conventional quantum Krylov algorithm. 
\end{abstract}

\maketitle


\section{Introduction}
Accurately determining ground-state energies and properties of quantum many-body systems remains a central challenge in physics and chemistry, with far-reaching implications for materials science, drug discovery, and a fundamental understanding of nature. The advent of quantum computers has set expectations to revolutionize this field, offering the potential to solve problems intractable for even the most powerful classical supercomputers~\cite{cao2019quantum,mcardle2020quantum,bauer2020quantum}. Among the various quantum algorithms proposed for this task, quantum phase estimation (QPE) has long been considered the gold standard~\cite{kitaev1995quantum,kitaev2002classical}. QPE leverages the principles of quantum mechanics to directly extract the eigenvalues of a Hamiltonian operator, providing a rigorous guarantee on the accuracy of the estimated ground-state energy. This distinguishes QPE from classical heuristic methods such as density matrix renormalization group (DMRG)~\cite{chan2011density,baiardi2020density} and coupled cluster theory~\cite{bartlett2007coupled}, which often lack such guarantees and can struggle with strongly correlated systems. However, the implementation of QPE for practical chemistry problems demands deep quantum circuits with long coherence times and extensive quantum gate resources~\cite{lee2021femoco,goings2022p450}, making it a long-term goal contingent on the realization of fault-tolerant quantum computers (FTQC). The stringent requirements of QPE have spurred significant research efforts into alternative quantum algorithms that can be implemented on near-term quantum devices. 

In the near-term noisy intermediate-scale quantum (NISQ) era~\cite{preskill2018quantum}, the variational quantum eigensolver (VQE) and related hybrid quantum-classical algorithms have garnered significant attention~\cite{peruzzo2014variational}. VQE offers the potential for implementation on shallower circuits compared to QPE, making it more amenable to the limitations of current quantum hardware. In VQE, a parameterized quantum circuit, known as an ansatz, is used to prepare a trial wave function, and the energy expectation value is optimized iteratively using a classical optimization algorithm. While VQE has shown promise for tackling small-scale molecular systems, it suffers from several challenges that hinder its scalability to practical problem sizes. One major obstacle is the presence of barren plateaus in the optimization landscape~\cite{mcclean2018barren,cerezo2021cost,wang2021noise}, where the gradients of the energy expectation value vanish exponentially with the number of qubits, making it difficult to train the quantum circuit. Furthermore, the accuracy of VQE is highly dependent on the choice of ansatz and optimizer, and the lack of accuracy guarantees makes its advantage over classical heuristic methods unclear~\cite{lee2023evaluating}. 

Recently, quantum Krylov methods~\cite{parrish2019quantum,motta2020determining,stair2020multireference,cohn2021filter,seki2021power,klymko2022real,cortes2022quantum,epperly2022theory,shen2023real,stair2023stochastic,tkachenko2024quantum,zhang2024measurement,kirby2024analysis,lee2024sampling,dcunha2024fragment,yoshioka2024diagonalization,lee2025efficient,byrne2024quantum,oumarou2025molecular,szasz2025ground,oleary2025partitioned} have emerged as a compelling alternative, bridging the gap between QPE and VQE. Quantum Krylov algorithms leverage the power of Krylov subspace methods, which have been widely used in classical numerical linear algebra, to efficiently estimate the ground-state energy of a Hamiltonian. Quantum Krylov algorithms construct a Krylov subspace by repeatedly applying a function of the Hamiltonian operator $\hat{H}$, expressed as $f(\hat{H})$ for a function $f$, to an initial trial state, and then diagonalize the projected Hamiltonian within this subspace to obtain the ground-state energy estimate. Specifically, we refer to quantum Krylov methods utilizing the time-evolution operator, i.e., $f(\hat{H})=e^{-i\hat{H}t}$ for some time $t$, to construct the Krylov subspace as Krylov quantum diagonalization (KQD) in the following. Among the possible choices, the real-time evolution operator $e^{-i\hat{H}t}$ is particularly well suited to quantum hardware: it is unitary and can be implemented directly via standard Hamiltonian-simulation techniques with relatively shallow circuits~\cite{yoshioka2024diagonalization}, whereas non-unitary generators such as the Hamiltonian power, imaginary-time evolution, or inverse power cannot be applied directly and must instead be realized as a linear combination of unitaries, at an additional cost~\cite{zhang2024measurement}. This motivates our focus on optimizing the sampling cost within the real-time-evolution class.
KQD offers a convergence guarantee to the exact ground state~\cite{epperly2022theory,kirby2024analysis,shen2023real}, akin to QPE, while maintaining the relatively shallow circuit depths characteristic of quantum-classical hybrid algorithms like VQE. These features make KQD particularly well-suited for implementation on early fault-tolerant quantum computing (early-FTQC) devices~\cite{suzuki2022preftqc,katabarwa2024early,akahoshi2024star,toshio2025starv2,akahoshi2024compilation,kuroiwa2025averaging,Toshio2026-qm,Kanasugi2026-eo}, which are expected to have limited qubit counts and coherence times but still offer the potential for significant quantum speedups.

However, the KQD algorithms are susceptible to numerical instability due to various error sources, including quantum noise, Trotter errors, and statistical errors arising from finite sampling~\cite{epperly2022theory,kirby2024analysis,lee2024sampling,lee2025efficient}. While quantum noise and Trotter errors are expected to be mitigated in the early-FTQC era through techniques such as partial quantum error correction~\cite{akahoshi2024star,toshio2025starv2,akahoshi2024compilation,Toshio2026-qm,Kanasugi2026-eo}, statistical errors arising from finite sampling remain a persistent challenge, even with improved quantum device performance~\cite{lee2024sampling,lee2025efficient}. The sampling cost of KQD, which is determined by the number of quantum measurements required to estimate the matrix elements of the Hamiltonian projected onto the Krylov subspace, can be prohibitively high for classically intractable large-scale quantum systems. This is because the number of measurements required to achieve a desired level of accuracy scales polynomially with the system size. For example, if we naively decompose a molecular Hamiltonian as a linear combination of Pauli operators and measure each component equally, the sampling cost scales as $\order{N_{\rm orb}^4}$ with $N_{\rm orb}$ being the number of spatial orbitals. Therefore, reducing the sampling cost is crucial for the practical application of KQD and has been the subject of active research~\cite{zhang2024measurement,lee2024sampling,lee2025efficient}. Various techniques have been proposed to reduce the sampling cost of KQD, including the weighted sampling scheme~\cite {lee2024sampling} and the Hamiltonian modification techniques~\cite{lee2025efficient}. However, the fundamental limits on the achievable sampling cost in KQD remain poorly understood, and the development of algorithms that can approach these limits is a major challenge.

\begin{figure*}[htbp]
  \centering
  \includegraphics[width=0.98\textwidth]{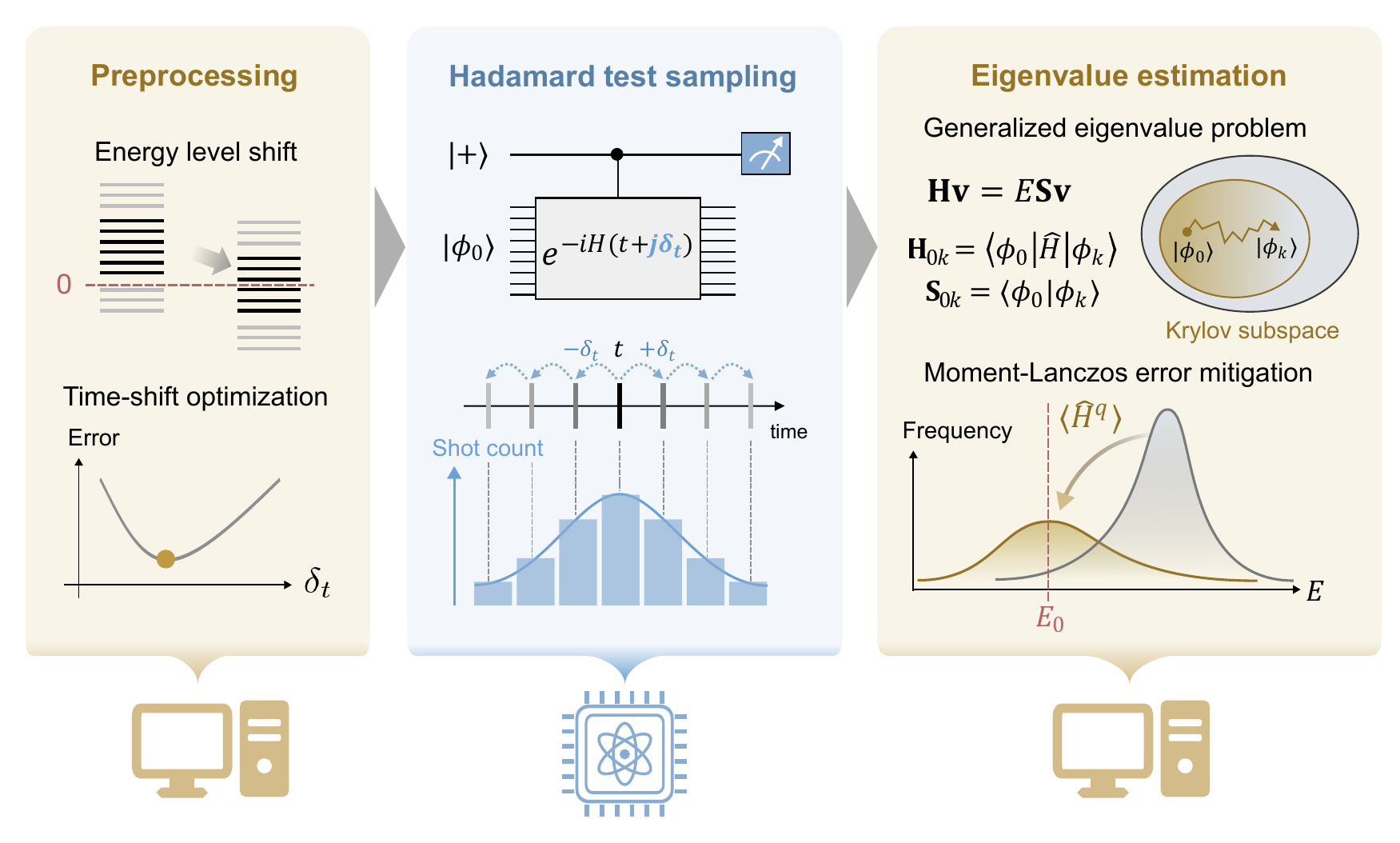}
  \caption{Schematic illustration of the MSD algorithm. It begins with two classical preprocessing steps: time-shift optimization and energy level shift, both designed to minimize the sampling cost. Next, Hadamard test sampling of shifted time-evolution operators is executed on a quantum processor. This sampling employs a shot allocation weighted by the central finite difference formula, which approximates the Hamiltonian operator. The matrix elements, estimated from the Hadamard test sampling data, are then used to classically solve a generalized eigenvalue problem for the Krylov subspace. Optionally, energy error mitigation can be applied using Hamiltonian moments derived from the Hadamard test sampling data.  
  \label{fig:msd}}
\end{figure*}
In this work, we propose a quantum Krylov algorithm, which we term mirror subspace diagonalization (MSD). Its name, ``mirror", derives from its core principle of utilizing symmetrically shifted timesteps, based on a central finite-difference formula. 
MSD leverages a central finite-difference formula to express the Hamiltonian operator as a linear combination of time-evolution unitaries with symmetrically shifted timesteps. This enables efficient estimation of the projected Hamiltonian matrix within the Krylov subspace without the need to measure each term of the Hamiltonian individually.
By carefully optimizing the shift timestep parameter and performing a shifting of the energy spectrum, MSD minimizes the sampling cost (see Fig.~\ref{fig:msd} for a schematic illustration), which becomes comparable to a theoretical lower bound of the sampling cost of KQD approaches, as we analyze in this work. We provide a detailed theoretical analysis of the sampling cost of MSD and demonstrate its effectiveness through numerical simulations and resource estimations for various molecular Hamiltonians. Our results show that MSD can achieve sampling cost reductions ranging from 10 to 10,000 times compared to the KQD approach, making it a promising candidate for ground-state energy estimation on near-term quantum computers such as NISQ and early-FTQC devices.

\section{Results}
\subsection{Sampling cost of Krylov quantum diagonalization}\label{subsec:preliminaries}
In this section, we review the KQD algorithm for the ground-state energy estimation as a preliminary to our work. 
KQD is a quantum-classical hybrid algorithm for obtaining the low-energy eigenvalues and eigenvectors of a quantum many-body Hamiltonian. To deal with the eigenvalue problem for an exponentially large Hilbert space of quantum many-body systems, KQD performs exact diagonalization within the Krylov subspace in analogy with the classical Lanczos method~\cite{lanczos1950iteration}. Specifically, KQD projects the given Hamiltonian $\hat{H}$ onto an $n$-dimensional Krylov subspace $\mathcal{K}_n$ defined as follows:
\begin{align}
    \mathcal{K}_n(\phi_0) &= \mathrm{Span}\left( \ket{\phi_0}, \ket{\phi_1}, \cdots, \ket{\phi_{n-1}} \right), \label{eq:krylov_space}
\end{align}
where $\ket{\phi_k}:= \hat{U}^k\ket{\phi_0}$ is a nonorthogonal Krylov basis generated by repeated application of the time-evolution unitary $\hat{U}:=e^{-i\hat{H}\tau}$ for some specific timestep $\tau$ to the initial reference state $\ket{\phi_0}$. The Krylov order $n$ is chosen as $n\ll\mathrm{dim}\hat{H}$, which enables us to solve the projected eigenvalue problem on a classical computer~\cite{epperly2022theory,kirby2024analysis}. 
Based on the Rayleigh-Ritz method for a nonorthogonal basis, eigenvalue estimates of $\hat{H}$ can be obtained from the Krylov subspace~\eqref{eq:krylov_space} by solving the following generalized eigenvalue problem (GEVP)
\begin{align}
    \mathbf{H}\mathbf{v}_{j} = E^{(n)}_{j}\mathbf{S}\mathbf{v}_{j}, \label{eq:GEVP}
\end{align}
where $E_j^{(n)}$ is the $j$-th lowest eigenvalue of the order-$n$ Krylov subspace (i.e., $E_0^{(n)}\leq E_1^{(n)} \leq \cdots \leq E_{n-1}^{(n)}$), and $\mathbf{v}_{j}$ is the corresponding eigenvector. The projected Hamiltonian matrix $\mathbf{H}$ and the overlap matrix $\mathbf{S}$ are $n\times n$ Hermitian matrices defined as 
\begin{align}
    \mathbf{H}_{k'k} &= \bra{\phi_{k'}}\hat{H}\ket{\phi_{k}} = \bra{\phi_0}e^{i\hat{H}k'\tau}\hat{H
    }e^{-i\hat{H}k\tau}\ket{\phi_0}, \label{eq:H_matrix} \\
    \mathbf{S}_{k'k} &= \braket{\phi_{k'}|\phi_{k}} = \bra{\phi_0}e^{-i\hat{H}(k-k')\tau}\ket{\phi_0}. \label{eq:S_matrix}
\end{align}
The overlap matrix $\mathbf{S}$ is a Toeplitz matrix satisfying $\mathbf{S}_{k',k} = \mathbf{S}_{0,k-k'}$ due to the unitarity of $e^{-i\hat{H}\tau}$. Furthermore, the projected Hamiltonian matrix $\mathbf{H}$ is also a Toeplitz matrix owing to the commutation relation $[\hat{H}, e^{-i\hat{H}\tau}]=0$ and is hence evaluated as 
\begin{align}
    \mathbf{H}_{k'k} &= \bra{\phi_0}\hat{H
    }e^{-i\hat{H}(k-k')\tau}\ket{\phi_0} = \mathbf{H}_{0,k-k'}. \label{eq:H_matrix_toeplitz}
\end{align}
As a consequence of the Toeplitz structure of matrices $\mathbf{S}$ and $\mathbf{H}$, we need to evaluate only $n$ independent matrix elements. 
The application of Eq.\eqref{eq:H_matrix_toeplitz} to the GEVP introduces an error primarily due to the non-commutativity between $\hat{H}$ and the Trotterized time-evolution unitary when approximating time evolution via Trotterization. This Trotter error can be substantial when only shallow-depth circuits are feasible on a quantum device. In such a case, we need to apply a non-Toeplitz construction based on Eq.~\eqref{eq:H_matrix} at the expense of the increased sampling cost. Nevertheless, the Toeplitz construction becomes superior to the non-Toeplitz construction in terms of the total error (encompassing both sampling and Trotter errors), even in the presence of such Trotter error, assuming that a specific circuit depth of $\mathcal{O}(n^2\tau^2)$ is attainable~\cite{lee2024sampling}. Therefore, assuming that sufficient circuit depth is achievable, our subsequent discussion will proceed with the Toeplitz construction unless otherwise mentioned.

In the KQD algorithm, the matrix pair $(\mathbf{H},\mathbf{S})$ is measured using a quantum computer. Therefore, the estimated matrices are perturbed by hardware noise of the quantum processor, algorithmic errors such as Trotter error, and statistical error due to finite sampling. We suppose that such matrix perturbations are described as follows: 
\begin{align}
    \tilde{\mathbf{S}} &:= \mathbf{S} + \mathbf{\Delta}_{\mathbf{S}}, \quad
    \tilde{\mathbf{H}} := \mathbf{H} + \mathbf{\Delta}_{\mathbf{H}}, \label{eq:noise_H}
\end{align}
where $(\tilde{\mathbf{H}}, \tilde{\mathbf{S}})$ is the perturbed matrix pair. Then, we need to solve the following perturbed GEVP: 
\begin{align}
    \tilde{\mathbf{H}}\tilde{\mathbf{v}}_j = \tilde{E}_j^{(n)}\tilde{\mathbf{S}}\tilde{\mathbf{v}}_j, \label{eq:GEVP_noisy}
\end{align}
where $\tilde{E}_j^{(n)}$ is the perturbed $j$-th lowest eigenvalue for the order-$n$ noisy Krylov subspace and $\tilde{\mathbf{v}}_j$ is the corresponding perturbed eigenvector. 

This noisy GEVP [Eq.~\eqref{eq:GEVP_noisy}] is generally ill-conditioned, because small eigenvalues of the perturbed overlap matrix $\tilde{\mathbf{S}}$ make its inversion numerically unstable. We regularize it by the standard thresholding procedure~\cite{epperly2022theory}, in which the eigenbasis of $\tilde{\mathbf{S}}$ with eigenvalues at or below a threshold level $\epsilon>0$ is discarded, yielding a reduced GEVP of dimension $n_\epsilon\leq n$. The details of this regularization are given in the Methods.

Now, we clarify the relationship between the matrix perturbation and the ground-state energy estimation error in the KQD algorithm.  As mentioned in the Introduction, KQD has a theoretical upper bound on the ground-state energy estimation error, which dictates a convergence guarantee similar to QPE. General error analyses for the perturbed GEVP of KQD were presented in Refs.~\cite{epperly2022theory,kirby2024analysis}. As reviewed in Supplementary Section I, the error of the ground-state energy estimate from this regularized GEVP is bounded by the sum of two contributions: a projection error that is suppressed exponentially with the Krylov dimension $n$ under a finite spectral gap, and a term that is sublinear in the noise strength $\eta:=\sqrt{\norm{\bm{\Delta}_{\mathbf{S}}}^2+\norm{\bm{\Delta}_{\mathbf{H}}}^2}$. Here, $\norm{\cdots}$ denotes the spectral norm. The key consequence for the present work is that the energy estimation error can be reduced by suppressing the matrix perturbations $\norm{\bm{\Delta}_{\mathbf{S}}}$ and $\norm{\bm{\Delta}_{\mathbf{H}}}$, which is the strategy MSD pursues.

Although hardware noise and algorithmic errors such as the Trotter error can be systematically alleviated on early-FTQC devices through quantum error correction and mitigation (e.g., Refs.~\cite{akahoshi2024star,toshio2025starv2,akahoshi2024compilation,Toshio2026-qm,Kanasugi2026-eo}), the statistical sampling error arising from the Hadamard-test measurement protocol (described in the Methods) remains inevitable even under perfect fault tolerance. The effect of such sampling error on KQD has been investigated in detail in previous works~\cite{lee2024sampling,lee2025efficient}.
To measure the projected Hamiltonian matrix elements in KQD, we typically decompose the Hamiltonian into a linear combination of unitaries (LCU) as 
\begin{align}
    \hat{H} &= c_0\hat{\1} +\sum_{l=1}^{L}c_l\hat{P}_l, \label{eq:LCU}
\end{align}
where $c_l\in\mathbb{R}$, $\hat{\1}$ is the identity operator, and $\hat{P}_l$ is a unitary operator (i.e., $\hat{P}_l\hat{P}_l^{\dag}=\hat{\1}$). Although $\{\hat{P}_l\}$ can be an arbitrary set of unitary operators in principle, we consider a Pauli LCU decomposition with $\hat{P}_l=\{\hat{\1}, \hat{X}, \hat{Y}, \hat{Z}\}^{\otimes N_q}$, where $\{\hat{X}, \hat{Y}, \hat{Z}\}$ are Pauli matrices and $N_q$ is the number of qubits. 
In this setting, the Hadamard-test sampling variance of each matrix element, combined with the optimal shot allocation over the Pauli-LCU components, is governed by the 1-norm of the LCU Hamiltonian~\eqref{eq:LCU} (see Supplementary Section III),
\begin{align}
    \lambda = \sum_{l=1}^{L}|c_l|. \label{eq:1_norm}
\end{align}
Based on random matrix theory and these optimized shot allocations, Ref.~\cite{lee2024sampling} elucidated the statistical behavior of the matrix perturbations $\norm{\bm{\Delta}_{\mathbf{S}}}$ and $\norm{\bm{\Delta}_{\mathbf{H}}}$, summarized in the following theorem (see Supplementary Section III for the derivation).

\begin{thm}[Sampling perturbation of KQD~\cite{lee2024sampling}]
    Suppose the matrix pair $\mathbf{Z}=\mathbf{S},\mathbf{H}$ of the KQD algorithm is obtained from the sequence of Hadamard test sampling, whose shot allocation $\{m_k\}_{k=0}^{n-1}$ is given by
    \begin{align}
        m_k &= 
        \begin{cases}
            \dfrac{1-\delta_{0k}}{n-1}M, & \mathbf{Z}=\mathbf{S} \\
            \dfrac{\delta_{0k}+\sqrt{2}(1-\delta_{0k})}{\sqrt{2}(n-1)+1}M, & \mathbf{Z}=\mathbf{H} \\ 
        \end{cases} \label{eq:kqd_shot_allocation}
    \end{align}
    where $M$ denotes the total number of shots. Then, the matrix perturbations are upper bounded as follows: 
    \begin{align}
        \norm{\mathbf{\Delta}_{\mathbf{S}}} &\lesssim \frac{2n\sqrt{2\log{(2n)}}}{\sqrt{M}}, \label{eq:sampling_perturbation_S} \\
        \norm{\mathbf{\Delta}_{\mathbf{H}}} &\lesssim \frac{2n\lambda\sqrt{2\log{(2n)}}}{\sqrt{M}}. \label{eq:sampling_perturbation_H}
    \end{align}
    Here, the projected Hamiltonian matrix $\mathbf{H}$ is assumed to be sampled based on the LCU decomposition in Eq.~\eqref{eq:LCU}, with the shots distributed over the Pauli-LCU components according to the optimal allocation derived in Supplementary Section III.
\end{thm}

As shown in Supplementary Section I, the ground-state energy error is bounded by a term that decreases sublinearly with the noise strength $\eta$, scaling as $\eta^{1/(1+\alpha)}$ with a constant $0\leq\alpha\leq 1/2$. Combining this bound with the matrix-perturbation scaling of Theorem~1 [Eqs.~\eqref{eq:sampling_perturbation_S} and~\eqref{eq:sampling_perturbation_H}], the total number of shots required to reach an energy accuracy $\epsilon_0$
(i.e., $|\tilde{E}_0^{(n_\epsilon)}-E_0|\leq\epsilon_0$) is estimated as
\begin{align}
    M = \order{\frac{n^2\log{(n)}\lambda^2}{\epsilon_{0}^{2(1+\alpha)}}}. \label{eq:kqd_sampling_scaling}
\end{align}
Equation~\eqref{eq:kqd_sampling_scaling} shows that the sampling cost of KQD is governed by the 1-norm $\lambda$ of a given Hamiltonian. However, this energy-accuracy-based estimate is of limited practical use. The underlying bound is loose and significantly overestimates the actual energy error~\cite{epperly2022theory}, and it contains system-dependent constants---such as a condition-number factor and the exponent $\alpha$---that leave the precise relation between $\epsilon_0$ and $\norm{\mathbf{\Delta}_\mathbf{H}}$ unclear. This ambiguity is reflected in the large discrepancy between $\epsilon_0$ and $\norm{\mathbf{\Delta}_\mathbf{H}}$ observed in our numerical results below. We therefore adopt the matrix perturbation $\norm{\mathbf{\Delta}_\mathbf{H}}$ itself, rather than $\epsilon_0$, as the target quantity for the sampling-cost analysis, noting that the two are related but not generally equivalent.

\subsection{Mirror subspace diagonalization}\label{sec:MSD}
In this section, we describe our methodology to reduce the sampling cost of the quantum Krylov method compared to the conventional LCU-based approach. The starting point of our approach is to use a differential representation of the Hamiltonian
\begin{align}
    \hat{H} &= \left.i\frac{d}{dt}e^{-i\hat{H}t}\right|_{t=0} , \label{eq:H_diff}
\end{align}
where $t$ denotes a parameter for the evolution time. 
Then, we approximate the differentiation based on the degree-$J$ central finite difference formula~\cite{gilyen2019optimizing} 
\begin{align}
    \left.\frac{d}{dt}e^{-i\hat{H}t}\right|_{t=0} 
    = \frac{1}{\delta_t}\sum_{j=-J}^{J}a_j^{(J;1)}e^{-i\hat{H}j\delta_t} + \order{\delta_t^{2J}}, \label{eq:finite_difference}
\end{align}
where $\delta_t>0$ is a time shift parameter and the coefficients $a_j^{(J;1)}$ are given by
\begin{align}
    a_j^{(J;1)}&= (1-\delta_{j0})\dfrac{(-1)^{j-1}}{j}\dfrac{(J!)^2}{(J-\abs{j})! (J+\abs{j})!}. \label{eq:finite_difference_coeff}
\end{align}
The finite difference error can be arbitrarily small by increasing the degree $J$ or decreasing $\delta_t$, scaling as $\order{\delta_t^{2J}}$. The reason we adopted the central finite difference scheme is to maintain the Hermiticity of the Hamiltonian operator $\hat{{H}}=\hat{{H}}^{\dag}$. 
Based on the above finite difference approximation of the Hamiltonian $\hat{H}$, the projected Hamiltonian matrix $\mathbf{H}$ of KQD is approximated as
\begin{align}
    \mathbf{H}_{k'k} &= \bra{\phi_0}\hat{H
    }e^{-i\hat{H}(k-k')\tau}\ket{\phi_0}, \nonumber\\
    & \simeq \frac{i}{\delta_t}\sum_{j=-J}^{J}a_j^{(J;1)}\bra{\phi_0}e^{-i\hat{H}[(k-k')\tau+j\delta_t]}\ket{\phi_0}. \label{eq:msd_H_mat}
\end{align}
The idea of applying the finite difference representation of the Hamiltonian to KQD was originally proposed in Ref.~\cite {bespalova2021hamiltonian}. 
However, Ref.~\cite {bespalova2021hamiltonian} did not investigate the critical relationship between the finite difference scheme and the theoretical sampling cost, nor did it comprehensively compare its performance against the conventional LCU-based KQD approach. Consequently, the true effectiveness and potential advantages of this finite difference approach remained ambiguous. In contrast, our work systematically develops a novel finite-difference-based KQD algorithm, MSD, that demonstrably approaches the theoretical lower bound for sampling cost, providing a detailed and favorable comparison to existing methods.

In the following, we introduce {\it time-shift optimization} and {\it energy level shift} techniques to maximize the sampling efficiency of the finite difference formulation in Eq.~\eqref{eq:msd_H_mat} (see Fig.~\ref{fig:msd} for a schematic illustration). We call such a finite difference-based quantum Krylov method, whose performance is optimized by the above techniques, MSD. Then, we demonstrate the superiority of MSD compared to the conventional LCU-based KQD. 

\subsection{Error analysis of MSD \label{subsec:msd_error_analysis}}
To optimize the performance of the finite difference scheme, we first analyze its error behavior. Specifically, we derive an upper bound on the matrix perturbation $\norm{\mathbf{\Delta}_{\mathbf{H}}}$ under the approximation in Eq.~\eqref{eq:msd_H_mat}. 
In the conventional KQD algorithms, the matrix perturbation arises from hardware noise of the quantum processor, algorithmic errors such as Trotter error, and statistical error due to finite sampling. In addition to these noise sources, the finite difference error contributes to the matrix perturbation when we sample the projected Hamiltonian matrix $\mathbf{H}$ according to Eq.~\eqref{eq:msd_H_mat}. Since hardware noise and algorithmic error can be suppressed with the growth of early-FTQC devices, we focus on the remaining two factors, i.e., the sampling error and finite difference error. Then, we describe the matrix perturbation as the following triangular inequality
\begin{align}
    \norm{\mathbf{\Delta}_{\mathbf{H}}} :=  \norm{\mathbf{\Delta}_{\mathbf{H}}^{(\rm s)}+\mathbf{\Delta}_{\mathbf{H}}^{(\rm d)}} \leq \norm{\mathbf{\Delta}_{\mathbf{H}}^{(\rm s)}} + \norm{\mathbf{\Delta}_{\mathbf{H}}^{(\rm d)}}, \label{eq:msd_matrix_perturbation_ineq}
\end{align}
where $\norm{\mathbf{\Delta}_{\mathbf{H}}^{(\rm s)}}$ and $\norm{\mathbf{\Delta}_{\mathbf{H}}^{(\rm d)}}$ denote the sampling and finite difference error contributions, respectively.  

We first estimate the sampling error $\norm{\mathbf{\Delta}_{\mathbf{H}}^{(\rm s)}}$. 
Suppose that each propagator matrix element in Eq.~\eqref{eq:msd_H_mat}
\begin{align}
    \mathbf{U}_{k'k}^{(j)}:=\bra{\phi_0}e^{-i\hat{H}[(k-k')\tau+j\delta_t]}\ket{\phi_0}, \label{eq:msd_U_mat}
\end{align}
is measured by the Hadamard test sampling described in the Methods. Note that this propagator matrix includes the overlap matrix as $\mathbf{U}_{k'k}^{(0)}=\mathbf{S}_{k'k}$. We also impose the Toeplitz-Hermitian structure as $\tilde{\mathbf{H}}_{k'k}=\tilde{\mathbf{H}}_{0,k-k'}$ and $\tilde{\mathbf{H}}_{k'k}=\tilde{\mathbf{H}}_{kk'}^*$ based on Eq.~\eqref{eq:msd_H_mat}, leading to the following requirement for the perturbed propagator matrix elements $\tilde{\mathbf{U}}_{k'k}^{(j)}$
\begin{align}
    \tilde{\mathbf{U}}_{k'k}^{(j)} &= \tilde{\mathbf{U}}_{0,k-k'}^{(j)}, \label{eq:msd_U_toeplitz} \\
    \tilde{\mathbf{U}}_{k'k}^{(j)} &= \tilde{\mathbf{U}}_{kk'}^{(-j)}. \label{eq:msd_U_hermite}
\end{align}
Under these requirements, $n((2J+1)-J)=n(J+1)$ independent matrix elements should be measured on a quantum processor. Here, we only need to perform the measurement for $j\geq0$ for the diagonal elements $\tilde{\mathbf{U}}_{00}^{(j)}=\tilde{\mathbf{U}}_{kk}^{(j)}$ owing to the Hermiticity condition~\eqref{eq:msd_U_hermite}. Then, based on the analysis presented in Supplementary Section III, the sampling variance is expressed as
\begin{align}
    \mathrm{Var}[\tilde{\mathbf{H}}_{0k},\vec{m}_k] &\simeq 
    \sum_{j\in \mathcal{J}_k}\frac{2(1+\delta_{0k})|a_j^{(J;1)}|^2}{\delta_t^2 m_{kj}}\left(2-\frac{1}{d}\right), \label{eq:variance_H_diff}
\end{align}
where 
\begin{align}
    \mathcal{J}_k:= 
    \begin{cases}
        \{1,\cdots,J\},  & k=0 \\
        \{-J,-J+1,\cdots,J-1,J\},  & k>0
    \end{cases}
\end{align}
and $\vec{m}_k=\{m_{kj}^{\rm(r)},m_{kj}^{\rm(i)}\}_{j=-J}^{J}=\{m_{kj}/2,m_{kj}/2\}_{j=-J}^{J}$ denotes the shot allocation for each component. Here, we have defined $m_{0j}:=0$ for $j<0$ considering the Hermiticity condition~\eqref{eq:msd_U_hermite}. Applying the Lagrange multiplier method, we obtain the optimized shot allocation
\begin{align}
    m_{kj}^{(\rm opt)}&=\frac{(1+\delta_{0k})|a_j^{(J;1)}|}{\norm{\vec{a}^{(J;1)}}_1}m_k \quad (j\in \mathcal{J}_k), \label{eq:shot_diff} 
\end{align}
and the resulting sampling variance
\begin{align}
    \mathrm{Var}[\tilde{\mathbf{H}}_{0k},\vec{m}_k^{(\rm opt)}] &\simeq 
    \begin{cases}
        \dfrac{\norm{\vec{a}^{(J;1)}}_1^2}{\delta_t^2 m_k}\left(2-\dfrac{1}{d}\right), & k=0 \\
        \dfrac{2\norm{\vec{a}^{(J;1)}}_1^2}{\delta_t^2 m_k}\left(2-\dfrac{1}{d}\right), & k >0
    \end{cases} \label{eq:variance_H_opt_diff}
\end{align}
where $\vec{a}^{(J;1)}:=\{a_{-J}^{(J;1)}, \cdots, a_{J}^{(J;1)}\}$ denotes the vector of the finite difference coefficients and $\norm{\vec{a}^{(J;1)}}_1:=\sum_{j=-J}^{J}|a_{j}^{(J;1)}|$ is the corresponding 1-norm. Then, adopting a random matrix theory similar to Ref.~\cite{lee2024sampling}, the sampling error contribution to the matrix perturbation $\norm{\mathbf{\Delta}_{\mathbf{H}}}$ is upper bounded as
\begin{align}
    \norm{\mathbf{\Delta}_{\mathbf{H}}^{(\rm s)}} &\lesssim \frac{2n\norm{\vec{a}^{(J;1)}}_1\sqrt{2\log{(2n)}}}{\delta_t \sqrt{M}},\label{eq:sampling_perturbation_H_diff}
\end{align}
where $M=\sum_{k=0}^{n-1} m_k$ is the total number of shots with the shot allocation $m_k$ given by Eq.~\eqref{eq:kqd_shot_allocation}. The detailed derivation is presented in Supplementary Section III. 

The finite difference error $\norm{\mathbf{\Delta}_{\mathbf{H}}^{(\rm d)}}$ is evaluated by considering the approximation error of the degree-$J$ central finite difference formula. Specifically, based on Lemma 22 of Ref.~\cite{gilyen2019optimizing} (see Supplementary Section IV for details), we obtain 
\begin{align}
    \abs{(\mathbf{\Delta}_{\mathbf{H}}^{(\rm d)})_{k'k}} &\leq \sum_{j=-J}^{J}\abs{a_j^{(J;1)}j^{2J+1}}\frac{\norm{\hat{H}}_{\mathcal{K}}^{2J+1}}{(2J+1)!} \delta_t^{2J} , \label{eq:H_mat_fd_error}
\end{align}
where 
\begin{align}
    \norm{\hat{H}}_{\mathcal{K}} := \max_{\psi\in\mathcal{K}_{\infty}}{|\bra{\psi}\hat{H}\ket{\psi}|}, \label{eq:krylov_norm}
\end{align}
denotes the maximum absolute eigenvalue (i.e., spectral norm) within the ideal order-$\infty$ Krylov subspace $\mathcal{K}_{\infty}$. For electronic structure problems, where the Hamiltonian preserves symmetries and possesses conserved quantities (e.g., particle number, total spin), $\norm{\hat{H}}_{\mathcal{K}}$ represents the spectral norm within the symmetry sector to which the initial state $\ket{\phi_0}$ belongs. 
Adopting the relation of the spectral norm $\norm{\cdots}$ and the Frobenius norm $\norm{\cdots}_{\rm F}$, the upper bound of the matrix perturbation $\norm{\mathbf{\Delta}_{\mathbf{H}}^{(\rm d)}}$ is derived as 
\begin{align}
    \norm{\mathbf{\Delta}_{\mathbf{H}}^{(\rm d)}} &\leq \norm{\mathbf{\Delta}_{\mathbf{H}}^{(\rm d)}}_{\rm F}=\sqrt{\sum_{k'k}\abs{(\mathbf{\Delta}_{\mathbf{H}}^{(\rm d)})_{k'k}}^2} \nonumber \\
    &\leq n\sum_{j=-J}^{J}\abs{a_j^{(J;1)}j^{2J+1}}\frac{\norm{\hat{H}}_{\mathcal{K}}^{2J+1}}{(2J+1)!} \delta_t^{2J}. \label{eq:H_norm_fd_error}
\end{align}
Based on the above inequalities and the triangular inequality~\eqref{eq:msd_matrix_perturbation_ineq}, the total matrix perturbation is upper bounded as
\begin{align}
    \norm{\mathbf{\Delta}_{\mathbf{H}}} &\lesssim \frac{2n\norm{\vec{a}^{(J;1)}}_1\sqrt{2\log{(2n)}}}{\delta_t \sqrt{M}} \nonumber\\
    &+ n\sum_{j=-J}^{J}\abs{a_j^{(J;1)}j^{2J+1}}\frac{\norm{\hat{H}}_{\mathcal{K}}^{2J+1}}{(2J+1)!} \delta_t^{2J}. \label{eq:msd_matrix_perturbation}
\end{align}
In the following, we introduce techniques to minimize this matrix perturbation. 

\subsection{Time-shift optimization }
Here, we introduce an optimization technique for the time-shift parameter $\delta_t$. As seen in Eq.~\eqref{eq:msd_matrix_perturbation}, there is a trade-off between the sampling and finite difference errors for the parameter $\delta_t$. If we increase $\delta_t$, the sampling variance and the resulting sampling error $\norm{\mathbf{\Delta}_{\mathbf{H}}^{(\rm s)}}$ is reduced. On the other hand, $\delta_t$ needs to be a small value to suppress the finite difference error $\norm{\mathbf{\Delta}_{\mathbf{H}}^{(\rm d)}}$. Therefore, an optimal $\delta_t$ must be determined to minimize the combined effect of the sampling and finite difference perturbations. As a consequence of this optimization of $\delta_t$, which can be solved analytically via a simple convex optimization, we arrive at the following theorem. 
\begin{thm}[Optimized sampling cost of MSD]\label{thm:msd_optimization}
    Suppose that we define the time shift parameter $\delta_t$ as
    \begin{align}
        \delta_t^{(\rm opt)} &:= \left(\frac{\alpha_{n,J}}{2J\beta_{n,J}\norm{\hat{H}}_{\mathcal{K}}^{2J+1}\sqrt{M}}\right)^{\frac{1}{2J+1}}, \label{eq:msd_optimal_delta}
    \end{align}
    where 
    \begin{align}
         \alpha_{n,J} &:= 2n\sqrt{2\log{(2n)}}\norm{\vec{a}^{(J;1)}}_1   , \\
        \beta_{n,J} &:= \frac{n}{(2J+1)!}\sum_{j=-J}^{J}\abs{a_j^{(J;1)}j^{2J+1}} . 
    \end{align}
    Then, the total number of shots $M$ to achieve $\norm{\mathbf{\Delta}_{\mathbf{H}}}\leq \eta_{\mathbf{H}}$ for a given target accuracy $\eta_{\mathbf{H}}>0$ is estimated as 
    \begin{align}
        M \gtrsim \frac{(2J+1)^{2+\frac{1}{J}}\alpha_{n,J}^2\beta_{n,J}^{\frac{1}{J}}\norm{\hat{H}}_{\mathcal{K}}^{2+\frac{1}{J}}}{(2J)^2\eta_{\mathbf{H}}^{2+\frac{1}{J}}}. \label{eq:msd_sampling_cost_formal}
    \end{align}
\end{thm}

\begin{proof}
    We aim to minimize the right-hand side of Eq.~\eqref{eq:msd_matrix_perturbation}, which is expressed as
    \begin{align}
        \frac{\alpha_{n,J}}{\delta_t\sqrt{M}} + \beta_{n,J}\norm{\hat{H}}_{\mathcal{K}}^{2J+1}\delta_t^{2J} =: f(\delta_t).
    \end{align}
    To do this, we take a derivative of $f(\delta_t)$ with respect to $\delta_t$ as
    \begin{align}
        f'(\delta_t) = -\frac{\alpha_{n,J}}{\delta_t^2\sqrt{M}} + 2J\beta_{n,J}\norm{\hat{H}}_{\mathcal{K}}^{2J+1}\delta_t^{2J-1}.
    \end{align}
    The minimum of $f(\delta_t)$ is obtained at an optimal $\delta_t$ satisfying $f'(\delta_t)=0$, which is written as Eq~\eqref{eq:msd_optimal_delta}. Inserting this optimal value $\delta_t^{(\rm opt)}$ into the expression of $f(\delta_t)$, we obtain Eq.~\eqref{eq:msd_sampling_cost_formal} from the requirement $\norm{\mathbf{\Delta}_{\mathbf{H}}}\lesssim f(\delta_t^{(\rm opt)}) \leq \eta_{\mathbf{H}}$. 
\end{proof}

Here, we note that the coefficients $\{a_{-J}^{J;1}, \cdots, a_{J}^{J;1}\}$ satisfy the following inequalities~\cite{gilyen2019optimizing}:
\begin{align}
    \norm{\vec{a}^{(J;1)}}_1 = \sum_{j=-J}^{J}|a_j^{(J;1)}| \leq  \sum_{j=1}^{J}\frac{2}{j} &\leq 2(\log{(J)}+1), \label{eq:a_norm_bound} \\
    \frac{1}{(2J+1)!}\sum_{j=-J}^{J}\abs{a_j^{(J;1)}j^{2J+1}} &\leq e^{-\frac{J}{2}}. 
\end{align}
Using these results, asymptotic behavior of the sampling cost~\eqref{eq:msd_sampling_cost_formal} for a sufficiently large $J$ ($J\gg1$) can be described as 
\begin{align}
    M \gtrsim \frac{32e^{-\frac{1}{2}}n^{2+o(1)}\log{(n)}\log^2{(J)}\norm{\hat{H}}_{\mathcal{K}}^{2+o(1)}}{\eta_{\mathbf{H}}^{2+o(1)}}. \label{eq:msd_sampling_cost_asymp}
\end{align}

\subsection{Energy level shift \label{subsec:energy_shift}}
Equation~\eqref{eq:msd_sampling_cost_asymp} shows that the system size dependence of the sampling cost is governed by the restricted spectral norm $\norm{\hat{H}}_{\mathcal{K}}$. In this section, we introduce a technique to reduce this sampling cost by reducing $\norm{\hat{H}}_{\mathcal{K}}$. 

Specifically, we consider a case where the many-body Hamiltonian $\hat{H}$ preserves some symmetries and possesses corresponding conserved quantities. For instance, the electronic structure Hamiltonian is generally U(1) symmetric and preserves the total number of electrons. In such a scenario, the Hamiltonian possesses a block-diagonal structure as
\begin{align}
    \hat{H} = \bigoplus_{\bm{Q}}\hat{H}_{\bm{Q}}, \label{eq:H_block_decomp}
\end{align}
where $\hat{H}_{\bm{Q}}$ denotes a block of a symmetry sector characterized by the corresponding quantum number $\bm{Q}$. A typical electronic structure Hamiltonian is generally characterized by three symmetries: $(\hat{N}_e, \hat{S}^2, \hat{S}_z)$, where $\hat{N}_e$ is the total particle number operator, $\hat{S}^2$ is the total spin operator, and $\hat{S}_z$ is the $z$-component of the spin projection. Then, the corresponding quantum number is $\bm{Q}=(N_e,S,S_z)$, where $N_e$ is the total number of electrons, while $S$ and $S_z$ denote the total spin and its projection. 

\begin{figure}[tbp]
  \centering
  \includegraphics[width=0.49\textwidth]{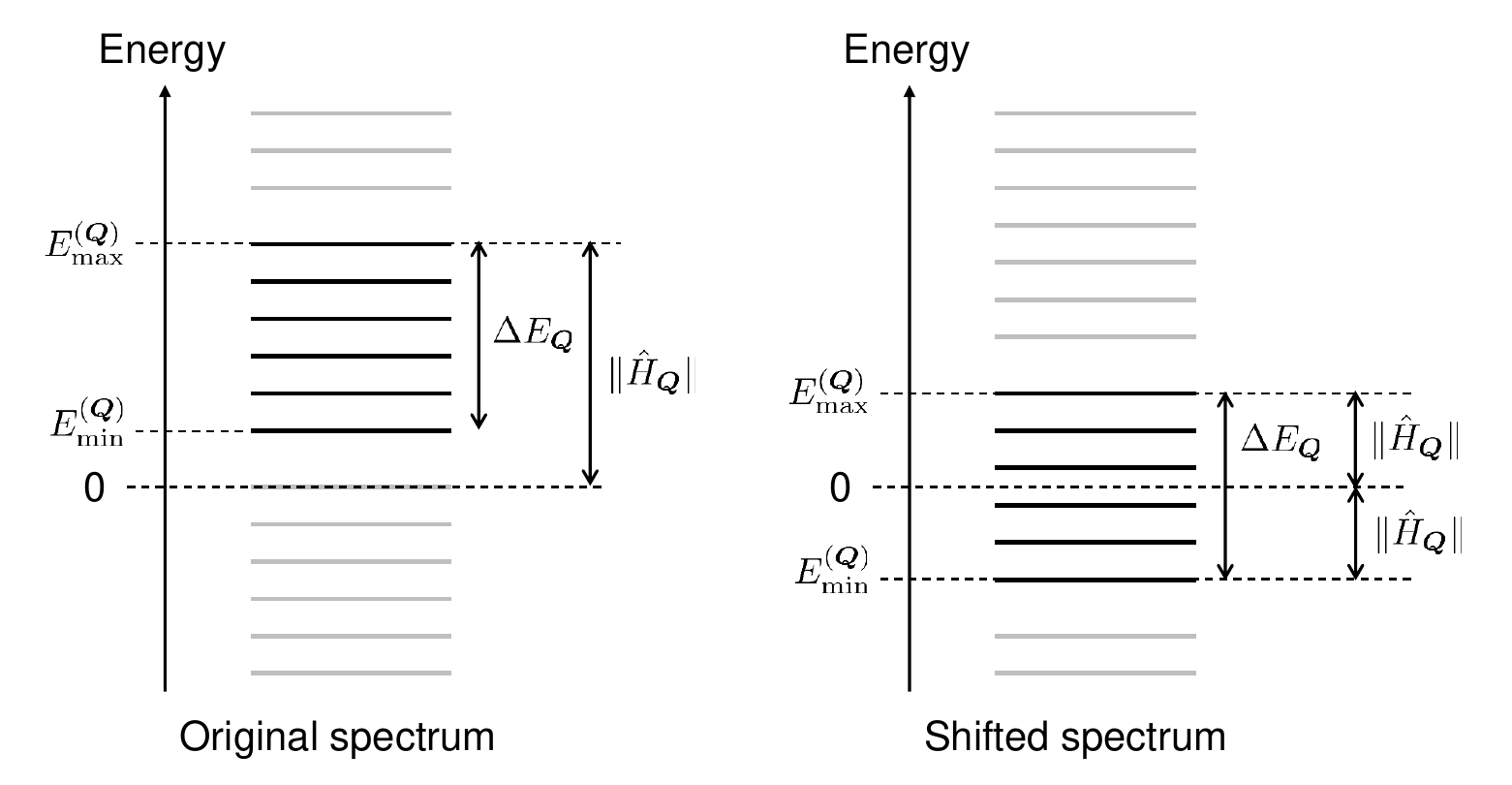}
  \caption{Illustration of the energy level shift technique to reduce the spectral norm $\norm{\hat{H}_{\bm{Q}}}$ for a $\bm{Q}$-symmetry sector. Energy levels belonging to the $\bm{Q}$-symmetry sector are depicted as black lines, while energy levels for other symmetry sectors are depicted as gray lines.
  By shifting the energy levels so that the center of the $\hat{H}_{\bm{Q}}$ spectrum is at the origin, the spectral norm $\norm{\hat{H}_{\bm{Q}}}$ is minimized to $\Delta{E}_{\bm{Q}}/2$. }
  \label{fig:spectral_shift}
\end{figure}

If the initial state $\ket{\phi_0}$ belongs to a specific $\bm{Q}$-symmetry sector, $\norm{\hat{H}}_{\mathcal{K}}$ represents the maximum absolute eigenvalue (spectral norm) of $\hat{H}_{\bm{Q}}$ according to Eq.~\eqref{eq:H_block_decomp}, i.e., $\norm{\hat{H}}_{\mathcal{K}}=\norm{\hat{H}_{\bm{Q}}}$. This is schematically illustrated in Fig.~\ref{fig:spectral_shift}, indicating that the value of $\norm{\hat{H}_{\bm{Q}}}$ depends on the choice of the origin of the energy spectrum. Specifically, the restricted spectral norm $\norm{\hat{H}_{\bm{Q}}}$ is minimized by performing the following energy shift
\begin{align}
    \hat{H} &\mapsto \hat{H} - \frac{E_{\rm min}^{(\bm{Q})} + E_{\rm max}^{(\bm{Q})}}{2} , \label{eq:energy_shift}
\end{align}
where $E_{\rm max}^{(\bm{Q})}:=\max_{j\in\bm{Q}}{E_{j}}$ and $E_{\rm min}^{(\bm{Q})}:=\min_{j\in\bm{Q}}{E_{j}}$ are the maximum and minimum energies in the $\bm{Q}$-symmetry sector, and $\Delta{E}_{\bm{Q}}:=E_{\rm max}^{(\bm{Q})}-E_{\rm min}^{(\bm{Q})}$ is the corresponding spectral range. This energy shift leads to the minimum spectral norm $\norm{\hat{H}_{\bm{Q}}}=\Delta{E}_{\bm{Q}}/2$ as illustrated in Fig.~\ref{fig:spectral_shift}. Then, the optimal sampling cost~\eqref{eq:msd_sampling_cost_asymp} is rewritten as 
\begin{align}
    M \gtrsim \frac{5n^{2+o(1)}\log{(n)}\log^2{(J)}\Delta{E}_{\bm{Q}}^{2+o(1)}}{\eta_{\mathbf{H}}^{2+o(1)}}, \label{eq:msd_sampling_cost_optimized}
\end{align}
where we have used $e^{-\frac{1}{2}}\simeq 0.61$. 

Performing the energy shifting in Eq.~\eqref{eq:energy_shift} requires values of $E_{\rm min}^{(\bm{Q})}$ and $E_{\rm max}^{(\bm{Q})}$. 
While exact eigenvalues are generally unfeasible to obtain, this shift can be efficiently implemented using approximate values readily derived from Hartree-Fock (HF) level calculations. Specifically, these approximate $E_{\rm min}^{(\bm{Q})}$ and $E_{\rm max}^{(\bm{Q})}$ values can be acquired via an HF-level orbital optimization method, as proposed in Ref.\cite{cortes2024spectral} (further details are available in Supplementary Section VII).  
In this approach, the HF-level approximate minimum eigenvalue $E_{\rm min}^{(\mathrm{HF},\bm{Q})}$ is larger than the exact value $E_{\rm min}^{(\bm{Q})}$ since the correlation energy is not fully included. This means that we can denote the approximate minimum eigenvalue as $E_{\rm min}^{(\mathrm{HF},\bm{Q})}=E_{\rm min}^{(\bm{Q})}+\varepsilon_{\rm corr}^{(\bm{Q})}$, where $\varepsilon_{\rm corr}^{(\bm{Q})}>0$ describes the correction due to the electron correlation. 
The approximate maximum eigenvalue $E_{\rm max}^{(\mathrm{HF},\bm{Q})}$ is obtained through the same optimization by taking the negative of the one- and two-electron integrals. Thus, the resulting $E_{\rm max}^{(\mathrm{HF},\bm{Q})}$ is smaller than the exact value $E_{\rm max}^{(\bm{Q})}$, leading to $E_{\rm max}^{(\mathrm{HF},\bm{Q})}=E_{\rm max}^{(\bm{Q})}-\tilde{\varepsilon}_{\rm corr}^{(\bm{Q})}$ with $\tilde{\varepsilon}_{\rm corr}^{(\bm{Q})}>0$ being the corresponding correlation energy. 
Consequently, the error of the energy shift can be quantified by the following metric
\begin{align}
    \delta_{\rm shift} &:= \frac{\abs{(E_{\rm min}^{(\mathrm{HF},\bm{Q})}+E_{\rm max}^{(\mathrm{HF},\bm{Q})}) - (E_{\rm min}^{(\bm{Q})}+E_{\rm max}^{(\bm{Q})})}}{\abs{E_{\rm min}^{(\mathrm{HF},\bm{Q})}+E_{\rm max}^{(\mathrm{HF},\bm{Q})}}} \nonumber\\
    &= \frac{\abs{\varepsilon_{\rm corr}^{(\bm{Q})}-\tilde{\varepsilon}_{\rm corr}^{(\bm{Q})}}}{\abs{E_{\rm min}^{(\mathrm{HF},\bm{Q})}+E_{\rm max}^{(\mathrm{HF},\bm{Q})}}}. \label{eq:energy_shift_error}
\end{align}
Since the HF energy typically accounts for over 99\% of the total energy of molecules~\cite{jensen2017introduction}, the energy shift error $\delta_{\rm shift}$ in Eq.~\eqref{eq:energy_shift_error} is generally expected to be within a few percent. Indeed, Supplementary Section VIII indicates that $\delta_{\rm shift}\lesssim 3\%$ for some benchmark molecular models. Therefore, the energy shift in Eq.~\eqref{eq:energy_shift} can be accurately conducted using the easily computable HF-level approximate eigenvalues with at most a few percent errors.

\subsection{Sampling cost reduction compared to KQD}\label{subsec:KQD_vs_MSD}
Here, we compare the sampling cost of KQD and MSD. Based on Eq.~\eqref{eq:sampling_perturbation_H}, the total number of shots required to achieve $\norm{\mathbf{\Delta}_{\mathbf{H}}}\leq \eta_{\mathbf{H}}$ in the conventional LCU-based KQD is given by 
\begin{align}
    M \gtrsim \frac{8n^2\log{(n)}\lambda^2}{\eta_{\mathbf{H}}^2}. \label{eq:kqd_sampling_cost}
\end{align}
This indicates that we can systematically reduce the sampling cost of KQD by reducing the 1-norm $\lambda$. Such 1-norm reduction is possible by adopting techniques like orbital optimization~\cite{koridon2021orbital} and block-invariant symmetry shift (BLISS)~\cite{loaiza2023bliss,patel2024blisslp}, which optimize the LCU decomposition of the Hamiltonian to reduce the 1-norm $\lambda$ while maintaining the spectrum structure (see Supplementary Section VII for details). 
The 1-norm for a general LCU-decomposed Hamiltonian has a lower bound given by $\lambda \ge \Delta E / 2$~\cite{loaiza2023lcu}, where $\Delta E$ represents the spectral range across the entire Hilbert space. For symmetry-preserving Hamiltonians, such as those encountered in electronic structure problems, the 1-norm can be significantly reduced using the symmetry-shift or BLISS~\cite{loaiza2023bliss,patel2024blisslp}. In such cases, the true lower bound is tightened to $\lambda \ge \Delta E_{\bm{Q}} / 2$~\cite{cortes2024spectral,loaiza2023bliss}.
Therefore, the lowest sampling cost, which is in principle achievable by the 1-norm reduction techniques, is described as
\begin{align}
    M_{\rm lowest} = \frac{2n^2\log{(2n)}\Delta{E}_{\bm{Q}}^2}{\eta_{\mathbf{H}}^2}. \label{eq:sampling_lower_bound}
\end{align}
However, achieving this sampling lower bound through the 1-norm reduction in the LCU-based KQD is in practice difficult. Indeed, various studies have shown that the 1-norm reduction methods typically realize one order of magnitude reduction of $\lambda$ at best~\cite{koridon2021orbital,patel2024blisslp}, while the spectral range $\Delta{E}_{\bm{Q}}$ can be more than one order of magnitude smaller than the 1-norm~\cite{cortes2024spectral}. Such a tendency can also be seen from our benchmark data in Supplementary Section VIII 

On the other hand, Eq.~\eqref{eq:msd_sampling_cost_optimized} indicates that MSD achieves a sampling cost whose scaling is almost equal to the sampling lower bound~\eqref{eq:sampling_lower_bound} up to a constant factor. 
The discrepancy from the sampling lower bound is characterized by the factor $\log{(J)}$. To avoid a significant increase in the simulation cost, analyzed later in this work, we here set $J=n$.
To circumvent the classical diagonalization of an exponentially large Hamiltonian matrix, the Krylov order $n$ in KQD is generally chosen as $n=\order{\mathrm{poly}(N_q)}$, where $N_q$ is the number of qubits. Consequently, the sampling cost of MSD is at most on the order of $\order{\mathrm{polylog}(N_q)}$ times that of the sampling lower bound $M_{\rm lowest}$ in Eq.~\eqref{eq:sampling_lower_bound}. Therefore, MSD can be considered as a real-time quantum Krylov algorithm that achieves a {\it near-optimal} sampling cost up to a logarithmic factor. 

Based on the above discussion, MSD realizes a significant sampling cost reduction from KQD, especially when the spectral range $\Delta{E}_{\bm{Q}}$ is much smaller than the 1-norm $\lambda$ and hence $\log{(n)}\Delta{E}_{\bm{Q}} \ll \lambda$. In the molecular Hamiltonian with $N_{\rm orb}$ spatial orbitals, the 1-norm $\lambda$ typically scales as $\lambda=\order{N_{\rm orb}^3}$ for the Pauli LCU, while the spectral range scales as $\order{N^x}$ with $x\lesssim 1$~\cite{cortes2024spectral}. Indeed, our benchmark data in Supplementary Section VIII demonstrates that $\lambda=\order{N_{\rm orb}^{2.52}}$ and $\Delta{E}_{\bm{Q}} =\order{N_{\rm orb}^{0.93}}$ on average across various molecules. 
This indicates that MSD is more efficient than KQD for a molecular Hamiltonian described by a large number of molecular orbitals. This condition is, for example, met in high-accuracy simulations of molecules using large basis sets incorporating electronic correlations~\cite{jensen2017introduction}. 
\begin{figure}[tbp]
  \centering
  \includegraphics[width=0.48\textwidth]{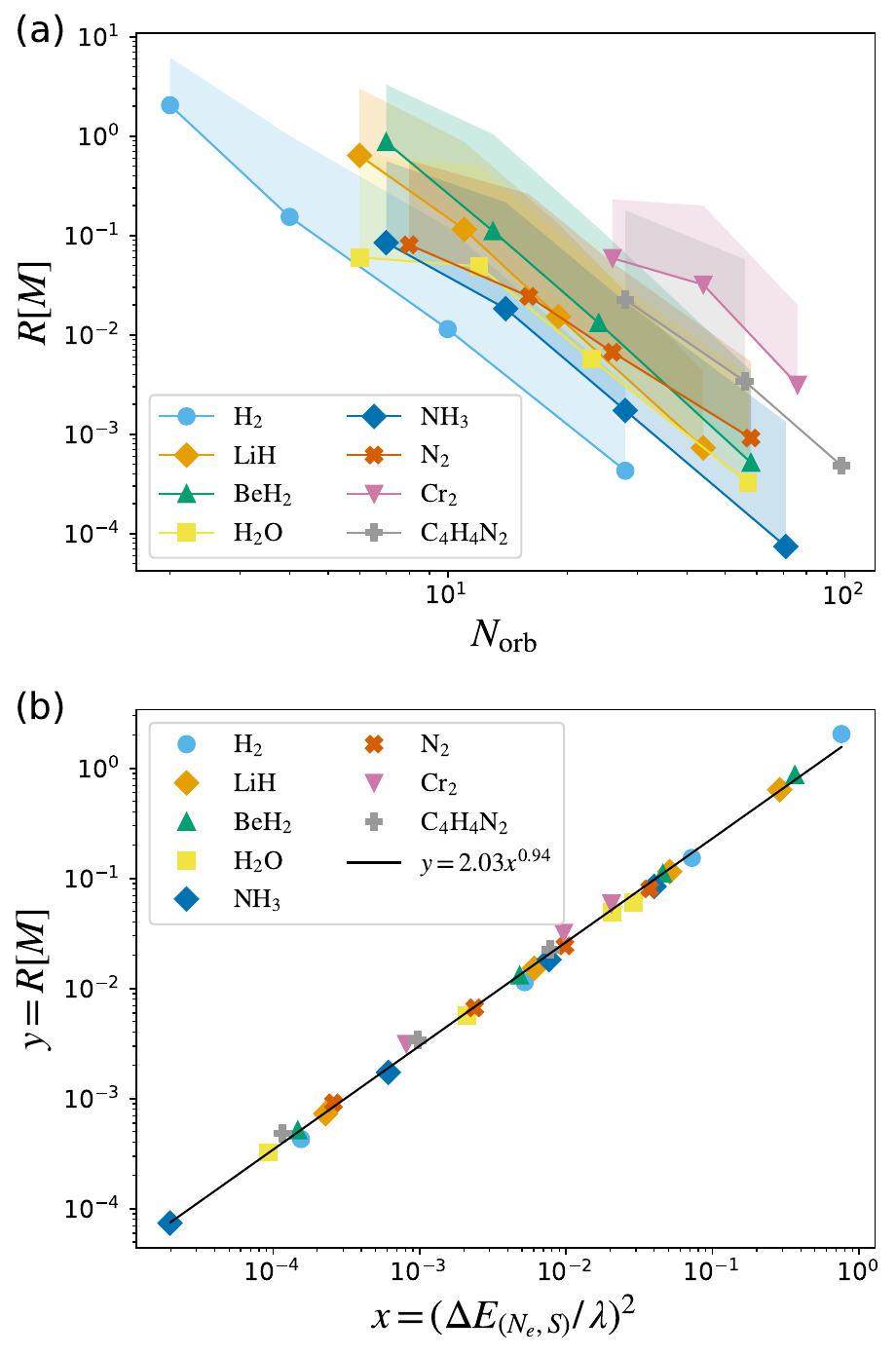}
  \caption{Numerical comparison of the sampling cost between KQD and MSD for various molecules. The number of spatial orbitals $N_{\rm orb}$ is increased for the same molecule by enlarging the basis set in the following order: STO-3G, 6-31G, cc-pVDZ, and cc-pVTZ. Dependence of the sampling cost ratio $R[M]:=M_{\rm msd}/M_{\rm kqd}$ on (a) the number of spatial orbitals $N_{\rm orb}$ and (b) the ratio $\Delta{E}_{(N_e,S)}/\lambda$. Here, $M_{\rm kqd}$ and $M_{\rm msd}$ are theoretical sampling costs for KQD and MSD, estimated based on Eqs.~\eqref{eq:kqd_sampling_cost} and~\eqref{eq:msd_sampling_cost_formal}, respectively, using the 1-norm $\lambda$ and spectral range $\Delta{E}_{\bm{Q}}$ for benchmark molecular models in Supplementary Section VIII. The shaded regions in (a) depict the ratio $R[M]$ achievable by adopting the 1-norm reduction to KQD. The upper bound of the shading is determined by the reduced 1-norm values, which are obtained by performing the Hamiltonian optimization in Supplementary Section VII. The parameters are chosen as $n=N_{\rm orb}$, $J=n=N_{\rm orb}$, and $\eta_{\mathbf{H}}=0.0016$.  }
  \label{fig:sampling_cost}
\end{figure}
To illustrate this behavior, Fig.~\ref{fig:sampling_cost} shows the sampling cost ratio $R[M]$ of MSD compared to KQD. Here, the ratio of the value of $A$ for MSD to that for KQD is denoted as $R[A]$. Concretely, $R[M]:=M_{\rm msd}/M_{\rm kqd}$, where $M_{\rm kqd}$ and $M_{\rm msd}$ are theoretical sampling costs for KQD and MSD estimated based on Eqs.~\eqref{eq:kqd_sampling_cost} and~\eqref{eq:msd_sampling_cost_formal}, respectively. The data points are obtained by using the 1-norm $\lambda$ and spectral range $\Delta{E}_{\bm{Q}}$ for benchmark molecular models detailed in Supplementary Section VIII. 
As shown in Fig.~\ref{fig:sampling_cost}(a), the sampling cost ratio $R[M]$ decreases with increasing $N_{\rm orb}$, meaning that MSD realizes a significant sampling cost reduction in the large basis set limit. Specifically, for the $\ce{NH3}$ molecule with the cc-pVDZ basis set, MSD achieves about a 13,500-fold (750-fold) decrease in sampling cost relative to KQD (1-norm optimized KQD). 
Figure~\ref{fig:sampling_cost}(b) depicts an almost linear correspondence between the sampling cost ratio $R[M]$ and $(\Delta{E}_{\bm{Q}}/\lambda)^2$, indicating that MSD is effective, especially when the spectral range $\Delta{E}_{\bm{Q}}$ is much smaller than the 1-norm $\lambda$. In other words, $R[M]$ scales as $\order{(\Delta{E}_{\bm{Q}}/\lambda)^2}=\order{N_{\rm orb}^{-3.18}}$, a finding consistent with the asymptotic scaling behavior in Eq.~\eqref{eq:msd_sampling_cost_optimized}. Consequently, MSD achieves an approximately cubic reduction in the sampling cost compared to KQD. 
The following numerical simulation will not only elaborate on this reduction in terms of ratios but also provide concrete numbers for the sampling iterations. 

The above sampling-cost advantage originates purely from the Hamiltonian property $\Delta{E}_{\bm{Q}}/\lambda$ and is independent of the quality of the reference state, since KQD and MSD share the identical Krylov subspace (Eq.~\eqref{eq:krylov_space}) and differ only in how the projected matrix elements are measured. Consequently, the multireference character of the target system affects both methods in the same way, and the ratio $R[M]$ is governed by $\Delta{E}_{\bm{Q}}/\lambda$ rather than by the reference-state overlap. We verify this expectation for genuinely strongly correlated systems—the dissociating \ce{N2} molecule and the \ce{Cr2} molecule at equilibrium—in Supplementary Section X.

We also comment on the relation to Ref.~\cite{zhang2024measurement}, which compares the measurement cost across different Krylov generating functions. In their analysis, the basis-dependent sampling overhead factorizes into an overlap-matrix conditioning factor and a factor set by the $1$-norm $C_{\mathbf{H}}$ of the LCU used to measure the projected Hamiltonian. Their Gaussian-power basis improves the former, whereas MSD improves the latter: within the real-time-evolution class, it provides an explicit protocol that lowers $C_{\mathbf{H}}$ from the naive value $\lambda$ to the spectral norm $\|\hat{H}\|$---the ideal value assumed, but not constructed, in Ref.~\cite{zhang2024measurement}---and further toward $\Delta E_{\bm{Q}}/2$ through the energy-level shift. Because the two methods act on different factors, they are complementary and could in principle be combined; the gain is most pronounced for molecular systems, where $\lambda \gg \|\hat{H}\|$. Finally, a complete cross-generator ranking would also require accounting for the circuit depth of each generator, which is orthogonal to the sampling-cost analysis here and is left for future work.

\subsection{Hamiltonian simulation cost \label{subsec:simulation_cost}}
We have seen that MSD realizes a significant sampling cost reduction from KQD and attains a near-optimal sampling cost. However, the gate count and circuit depth needed to perform the Hamiltonian simulation in the Hadamard test can be larger than KQD, since MSD generally requires longer time-evolution than KQD (see Eq.~\eqref{eq:msd_U_mat}). This subsection clarifies that MSD does not induce a significant increase in such Hamiltonian simulation costs under a proper choice of parameters. 

The Hamiltonian simulation cost of the repeated Hadamard test sampling, such as appearing in MSD and KQD, can be characterized by two metrics: {\it maximum evolution time} $T_{\rm max}$ and {\it total evolution time} $T_{\rm total}$~\cite{lin2022heisenberg}. The maximum evolution time $T_{\rm max}$ is the largest value of the evolution time $t$ in the propagator $e^{-i\hat{H}t}$ of the Hadamard test. Since the implementation cost of the propagator $e^{-i\hat{H}t}$ increases with $t$ in Hamiltonian simulation algorithms, $T_{\rm max}$ measures the maximum gate count or circuit depth required to perform the Hadamard test sampling. On the other hand, the total evolution time $T_{\rm total}$ denotes the sum of the evolution time of all of the Hadamard tests. This corresponds to the sum of the circuit depth over all executions, quantifying the total runtime. 

From the definition, the maximum evolution time $T_{\rm max}$ for KQD and MSD is given by
\begin{align}
    T_{\rm max} = 
    \begin{cases}
        (n-1)\tau, & \mathrm{KQD} \\
        (n-1)\tau + J\delta_t^{(\rm opt)}, & \mathrm{MSD} 
    \end{cases} \label{eq:t_max}
\end{align}
where $\tau=\pi/\Delta{E}_{\bm{Q}}$~\cite{epperly2022theory} and $\delta_t^{(\rm opt)}(>0)$ is determined by Eq.~\eqref{eq:msd_optimal_delta}. Therefore, the maximum evolution time of MSD increases by $J\delta_t^{(\rm opt)}$ from KQD. Since the parameter $J$ describes the degree of the central finite difference formula, it is desirable to take a sufficiently large $J$ to suppress the finite difference error. However, increasing $J$ comes at the cost of increased $T_{\rm max}$ by Eq.~\eqref{eq:t_max}. From Eq.~\eqref{eq:msd_optimal_delta}, the optimal shift timestep $\delta_t^{(\rm opt)}$ scales as $\delta_t^{(\rm opt)}=\order{\Delta{E}_{\bm{Q}}^{-1}\sqrt{M}^{-\frac{1}{2J+1}}}$, which is much smaller or comparable to the timestep $\tau=\pi/\Delta{E}_{\bm{Q}}$ depending on the total sampling $M$. Therefore, it is reasonable to take $J=n$ to limit the increase in $T_{\rm max}$ to at most a factor of two. 

The total evolution time $T_{\rm total}$ for KQD and MSD is described as 
\begin{align}
    T_{\rm total} = 
    \begin{dcases}
        \sum_{k=0}^{n-1}\sum_{l=1}^{L}m_{kl}^{(\rm opt)}k\tau, & \mathrm{KQD} \\
        \sum_{k=0}^{n-1}\sum_{j=-J}^{J}m_{kj}^{(\rm opt)}\abs{k\tau+j\delta_t^{(\rm opt)}}, & \mathrm{MSD} 
    \end{dcases} \label{eq:t_total}
\end{align}
where the optimal shot allocation for KQD is $m_{kl}^{(\rm opt)}=|c_l|m_k/\lambda$, as derived in Supplementary Section III, and the allocation for MSD, $m_{kj}^{(\rm opt)}$, is given by Eq.~\eqref{eq:shot_diff}. Substituting $m_{kl}^{(\rm opt)}$ together with Eq.~\eqref{eq:kqd_shot_allocation}, the $T_{\rm total}$ for KQD is obtained as
\begin{align}
    T_{\rm total} = \frac{n(n-1)}{2(n-1)+\sqrt{2}}M\tau \underset{n\gg1}{\simeq} \frac{n\tau}{2}M.  \label{eq:t_total_kqd}
\end{align}
Similarly, from Eqs.~\eqref{eq:shot_diff} and~\eqref{eq:kqd_shot_allocation}, the $T_{\rm total}$ for MSD is rewritten as
\begin{align}
    T_{\rm total} &= \frac{M}{\norm{\vec{a}^{(J;1)}}_1(\sqrt{2}(n-1)+1)} \left[ 2\delta_t^{(\rm opt)}\sum_{j=1}^{J}\abs{a_j^{(J;1)} j} \right. \nonumber\\
    & \left.+ \sqrt{2}\sum_{k=1}^{n-1}\sum_{j=-J}^{J}\abs{a_j^{(J;1)}}\abs{k\tau+j\delta_t^{(\rm opt)}} \right] .
\end{align}
Here, the first summation is bounded above as
\begin{align}
    \sum_{j=1}^{J}\abs{a_j^{(J;1)} j} \leq \sum_{j=1}^{J}1 = J,  
\end{align}
since $|a_j^{(J;1)}|\leq 1/|j|$. 
The second summation is evaluated as 
\begin{align}
    &\sum_{k=1}^{n-1}\sum_{j=-J}^{J}\abs{a_j^{(J;1)}}\abs{k\tau+j\delta_t^{(\rm opt)}} \nonumber\\
    &\leq 2\sum_{k=1}^{n-1}\sum_{j=1}^{J}\abs{a_j^{(J;1)}}(k\tau+j\delta_t^{(\rm opt)}) \nonumber\\
    &\leq n(n-1)\norm{\vec{a}^{(J;1)}}_1\tau + 2(n-1)J\delta_t^{(\rm opt)}.
\end{align}
Thus, the total evolution time $T_{\rm total}$ for MSD is estimated as 
\begin{align}
    T_{\rm total} &\leq \frac{n(n-1)M\tau}{\sqrt{2}(n-1)+1} + \frac{2JM\delta_t^{(\rm opt)}}{\norm{\vec{a}^{(J;1)}}_1}, \nonumber\\
    &\lesssim \left(\frac{n\tau}{\sqrt{2}} + \frac{J\delta_t^{(\rm opt)}}{\log{(J)}}\right)M, \label{eq:t_total_msd}
\end{align}
where we have used Eq.~\eqref{eq:a_norm_bound}. Comparing Eq.~\eqref{eq:t_total_msd} with Eq.~\eqref{eq:t_total_kqd}, it is expected that the value of $T_{\rm total}/M$ for MSD is nearly the same for KQD and MSD when $J=n$. Therefore, a reduction in the sampling cost $M$ is expected to directly translate into a reduction in the total runtime. 

\begin{figure}[tbp]
  \centering
  \includegraphics[width=0.48\textwidth]{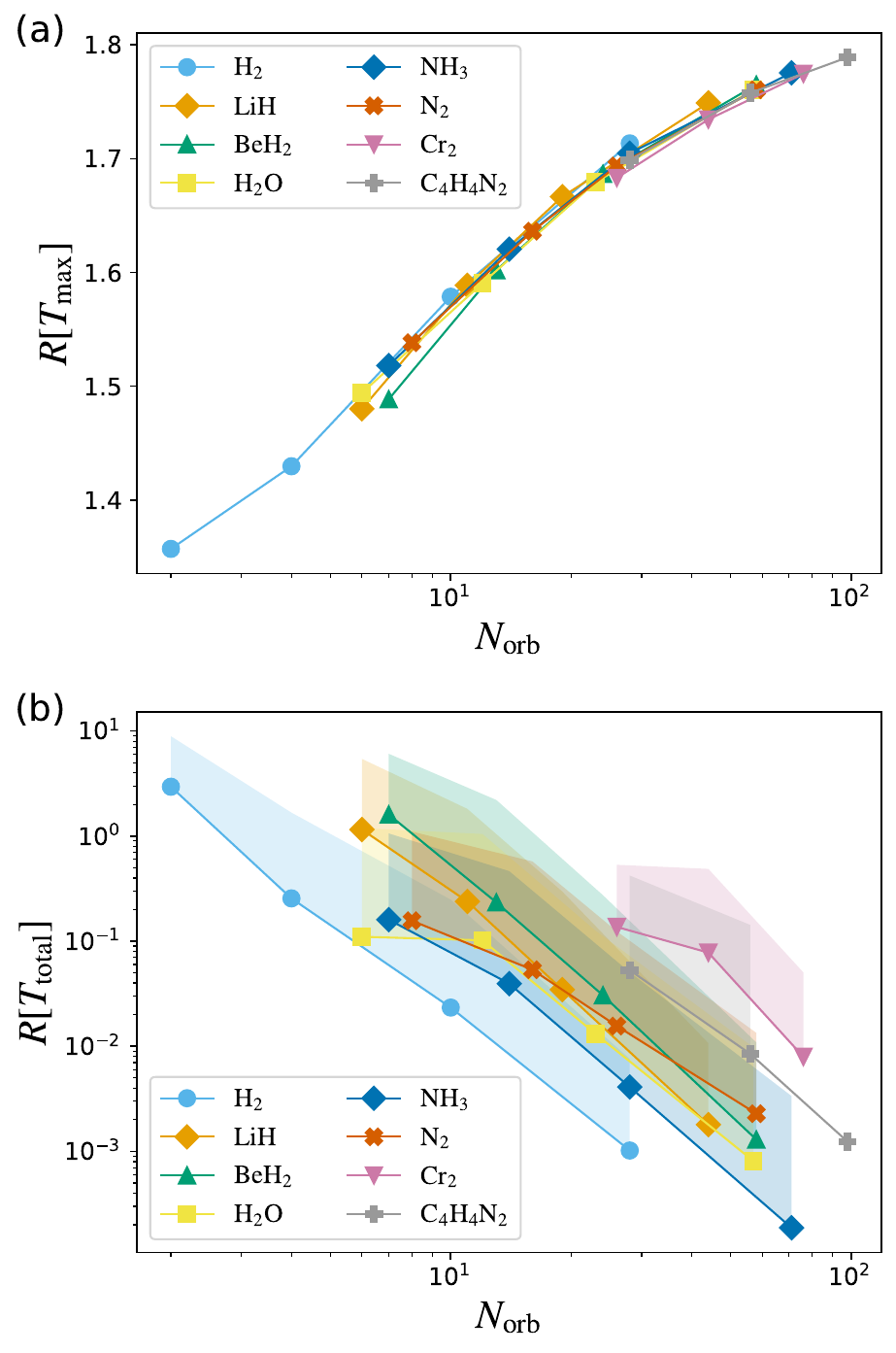}
  \caption{Numerical comparison of the simulation cost between KQD and MSD for various molecules. Dependence of (a) the maximal evolution time ratio $R[T_{\rm max}]$ and (b) the total evolution time ratio $R[T_{\rm total}]$ on the number of spatial orbitals $N_{\rm orb}$. Here, the total evolution time $T_{\rm total}$ is estimated based on theoretical sampling costs for KQD and MSD, determined by Eqs.~\eqref{eq:kqd_sampling_cost} and~\eqref{eq:msd_sampling_cost_optimized}, respectively, using the 1-norm $\lambda$ and spectral range $\Delta{E}_{\bm{Q}}$ for benchmark molecular models shown in Supplementary Section VIII. The shaded regions in (b) depict the ratio $R[T_{\rm total}]$ achievable by adopting the 1-norm reduction to KQD. The upper bound of the shading is determined by the reduced 1-norm values, which are obtained by performing the Hamiltonian optimization described in Supplementary Section VII. The parameters are chosen as $n=N_{\rm orb}$, $J=n=N_{\rm orb}$, and $\eta_{\mathbf{H}}=0.0016$. }
  \label{fig:simulation_cost}
\end{figure}
Figure~\ref{fig:simulation_cost} illustrates the ratio of $T_{\rm max}$ and $T_{\rm total}$ for MSD ($J=n$) to that for KQD, denoted as $R[T_{\rm max}]$ and $R[T_{\rm total}]$, respectively, following the $R[\cdot]$ notation introduced earlier. The data is obtained for various benchmark molecular models described in the Methods. As expected, the increase of $T_{\rm max}$ is limited to at most a factor of two as shown in Fig.~\ref{fig:simulation_cost}(a). Then, MSD does not cause a significant increase in the maximum gate count and circuit depth. The reduction ratio for the total evolution time $R[T_{\rm total}]$ in Fig.~\ref{fig:simulation_cost}(b) exhibits almost the same behavior as that for the total sampling $R[M]$ in Fig.~\ref{fig:sampling_cost}(a). This is consistent with the theoretical expectation described above. 

\subsection{Energy error mitigation via Hamiltonian moments \label{sec:moment}}
In MSD, we obtain a dataset of propagator matrix elements as in Eq.~\eqref{eq:msd_U_mat}. In this section, we introduce a technique to compute Hamiltonian moment~\cite{aulicino2022state,vallury2020quantum} from these propagation matrix elements and perform energy error mitigation using the Hamiltonian moments. 

First, we consider a differential representation of the $q$-th power of the Hamiltonian $\hat{H}$
\begin{align}
    \hat{H}^q &= \left.i^q\frac{d^q}{dt^q}e^{-i\hat{H}t}\right|_{t=0}, \label{eq:hamiltonian_power_diff}
\end{align}
where $q=1,2,\cdots$ and $q=1$ corresponds to the differential representation of the Hamiltonian in Eq.~\eqref{eq:H_diff}. 
The $q$-th derivative of Eq.~\eqref{eq:hamiltonian_power_diff} can be approximated by using the degree-$J$ central finite difference formula as 
\begin{align}
    \left.\frac{d^q}{dt^q}e^{-i\hat{H}t}\right|_{t=0}
    &\simeq \frac{1}{\delta_t^q}\sum_{j=-J}^{J}a_j^{(J;q)}e^{-i\hat{H}j\delta_t} + \order{\delta_t^{s-q+1}}, \label{eq:power_finite_difference}
\end{align}
where $\delta_t>0$ and $s:=q-1+2(J+1-\lfloor\frac{q+1}{2}\rfloor)$. $\{a_j^{(J;q)}\}_{j=-J}^{J}$ denotes the central finite difference coefficients for the $q$-th derivative, which are defined as 
\begin{align}
    a_j^{(J;q)} := \frac{d^q}{dz^q}\left. \left(\prod_{\substack{r=-J \\ r\neq j}}^{J} \frac{z-r}{j-r}\right) \right|_{z=0} . \label{eq:cfd_coeff_general}
\end{align}
Specifically, the solution of Eq.~\eqref{eq:cfd_coeff_general} for $q=1$ is given by Eq.~\eqref{eq:finite_difference_coeff}. 

Here, we introduce the projected Hamiltonian power matrix as 
\begin{align}
    \mathbf{M}^{(q)}_{k'k} &= \bra{\phi_0}e^{i\hat{H}k'\tau} \hat{H}^q e^{-i\hat{H}k\tau}\ket{\phi_0}, \nonumber\\
    &= \bra{\phi_0}\hat{H}^q e^{-i\hat{H}(k-k')\tau}\ket{\phi_0}, \label{eq:hamiltonian_power_matrix}
\end{align}
where we have adopted the commutation relation $[\hat{H}^q, e^{-i\hat{H}t}]=0$ in the second line. Note that $\mathbf{M}^{(1)}=\mathbf{H}$ by definition. 
Applying the central finite difference formula in Eq.~\eqref{eq:power_finite_difference}, the projected Hamiltonian power matrix $\mathbf{M}^{(q)}$ is approximated as 
\begin{align}
    \mathbf{M}^{(q)}_{k'k} &\simeq \frac{i^q}{\delta_t^q}\sum_{j=-J}^{J}a_j^{(J;q)}\bra{\phi_0}e^{-i\hat{H}[(k-k')\tau+j\delta_t]}\ket{\phi_0}. \label{eq:msd_power_matrix}
\end{align}
Based on Eq.~\eqref{eq:msd_power_matrix}, we can estimate the projected Hamiltonian power matrix $\mathbf{M}^{(q)}$ by classical postprocessing of the propagator matrices $\{\mathbf{U}^{(j)}\}_{j=-J}^{J}$ (Eq.~\eqref{eq:msd_U_mat}) obtained in MSD, for $q=1,2,\cdots,2J$. As detailed in Supplementary Section V, the corresponding matrix perturbation $\mathbf{\Delta}_{\mathbf{M}}^{(q)}:=\tilde{\mathbf{M}}^{(q)}-\mathbf{M}^{(q)}$ can be bounded in the same way as $\mathbf{\Delta}_{\mathbf{H}}$ by combining the finite-sampling and finite-difference contributions. With the optimal shift parameter $\delta_t^{(\rm opt)}$ [Eq.~\eqref{eq:msd_optimal_delta}], this bound scales as $\norm{\mathbf{\Delta}_{\mathbf{M}}^{(q)}} \lesssim \order{n\norm{\hat{H}}_{\mathcal{K}}^{q}/\sqrt{M}}$ up to $q$-dependent constant factors, so that the perturbation grows with the moment order $q$ mainly through the factor $\norm{\hat{H}}_{\mathcal{K}}^{q}$.

From the projected Hamiltonian power matrix $\mathbf{M}^{(q)}$, the $q$-th order Hamiltonian moment $\mu_q$ is obtained as 
\begin{align}
    \mu_q &:= \bra{\psi_0}\hat{H}^q\ket{\psi_0} 
    \simeq \frac{\mathbf{v}_0^{\dag}\mathbf{M}^{(q)}\mathbf{v}_0}{\mathbf{v}_0^{\dag}\mathbf{v}_0}, \label{eq:hamiltonian_moment}
\end{align}
where $\ket{\psi_0}$ is the ground state of the Hamiltonian $\hat{H}$ and $\mathbf{v}_0$ is the approximate ground state vector obtained by solving the GEVP in Eq.~\eqref{eq:GEVP}. Note that we obtain the Hamiltonian moments $\mu_q$ with $q=1,2,\cdots,2J$ for the degree-$J$ central finite difference formulation. 
Using the Hamiltonian moments, we perform error mitigation on the ground-state energy estimate based on the Lanczos scheme~\cite{witte1994plaquette,haxton2005piecewise,lanczos1950iteration}.
In this scheme, a tridiagonal (Lanczos) matrix $\mathbf{T}_j$ is constructed whose diagonal and off-diagonal elements $\alpha_j$ and $\beta_j$ are obtained recursively from the moments $\{\mu_q\}_{q=1}^{2j}$ as ratios of Hankel-type determinants; the explicit construction is reviewed in Supplementary Section VI. Diagonalizing $\mathbf{T}_j$ yields the lowest eigenvalue for $j=1,2,\cdots,J$.

In an ideal noiseless scenario, the lowest eigenvalue gradually decreases and approaches the exact ground state energy as the iteration step $j$ increases~\cite{witte1994plaquette,haxton2005piecewise}. However, in the presence of noise in MSD, the Hamiltonian moments are perturbed, leading to unphysical behavior. For instance, $\beta_j^2<0$ can occur, especially for large $j$, because higher-order Hamiltonian moments are significantly affected by noise, as described by the perturbation bound in Supplementary Section V. Furthermore, the lowest eigenvalue might increase with increasing iteration step $j$, which contradicts the principle of the Lanczos algorithm. To mitigate these undesirable behaviors, in practice, we terminate the Lanczos iteration when $\beta_j^2<0$ or an increase in the energy estimate is detected. 
Because these properties hold only in the noiseless limit and are broken by sampling noise, the moment-Lanczos procedure should be regarded as a heuristic error-mitigation scheme rather than one carrying a rigorous convergence guarantee.
Note that a similar moment-Lanczos approach has been previously suggested in Ref.~\cite{seki2021power}. However, the approach presented in~\cite{seki2021power} estimates the Hamiltonian moment using a degree-1 central finite difference formulation, in contrast to our higher-degree central finite difference scheme. 

\subsection{Numerical demonstration}\label{sec:numerical}
Here, we numerically demonstrate the performance of MSD. The details of the numerical simulations are described in the Methods.

\begin{figure}[tbp]
  \centering
  \includegraphics[width=0.48\textwidth]{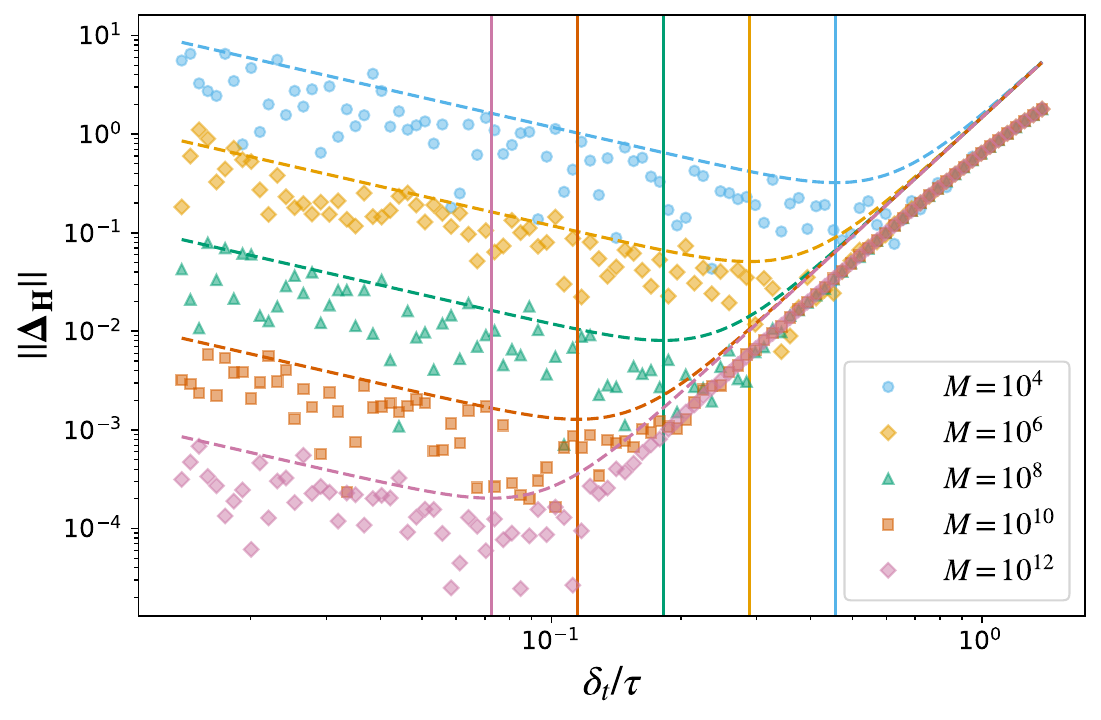}
  \caption{Dependence of matrix perturbation $\norm{\mathbf{\Delta}_{\mathbf{H}}}$ on shift parameter $\delta_t$ for the \ce{H2} (cc-pVDZ) model. Markers represent data obtained from finite sampling simulations of MSD with varying total numbers of shots $M$. MSD is conducted with the parameters $n=2$ and $J=n=2$. Dotted lines indicate the theoretical upper bound (Eq.\eqref{eq:msd_matrix_perturbation}), and vertical lines denote the optimal parameter $\delta_t^{(\rm opt)}$ (Eq.\eqref{eq:msd_optimal_delta}).}
  \label{fig:parameter_optimization}
\end{figure}

First, we demonstrate the time-shift parameter optimization introduced earlier. Figure~\ref{fig:parameter_optimization} illustrates the dependence of the matrix perturbation $\norm{\mathbf{\Delta}_{\mathbf{H}}}$ on $\delta_t$ for varying total numbers of shots, $M$. 
The matrix perturbation $\norm{\mathbf{\Delta}_{\mathbf{H}}}$ exhibits a convex downward behavior with respect to the shift parameter $\delta_t$. This behavior arises from the trade-off between the sampling error and the finite difference error, as discussed earlier in the time-shift optimization. Moreover, we can see that the matrix perturbation $\norm{\mathbf{\Delta}_{\mathbf{H}}}$ is minimized at the theoretical optimal value $\delta_t^{(\rm opt)}$ (denoted by vertical lines), which aligns with the assertion of Theorem~\ref{thm:msd_optimization}. Subsequent numerical simulations are conducted using this optimal shift parameter, $\delta_t^{(\rm opt)}$. 

\begin{figure}[tbp]
  \centering
  \includegraphics[width=0.48\textwidth]{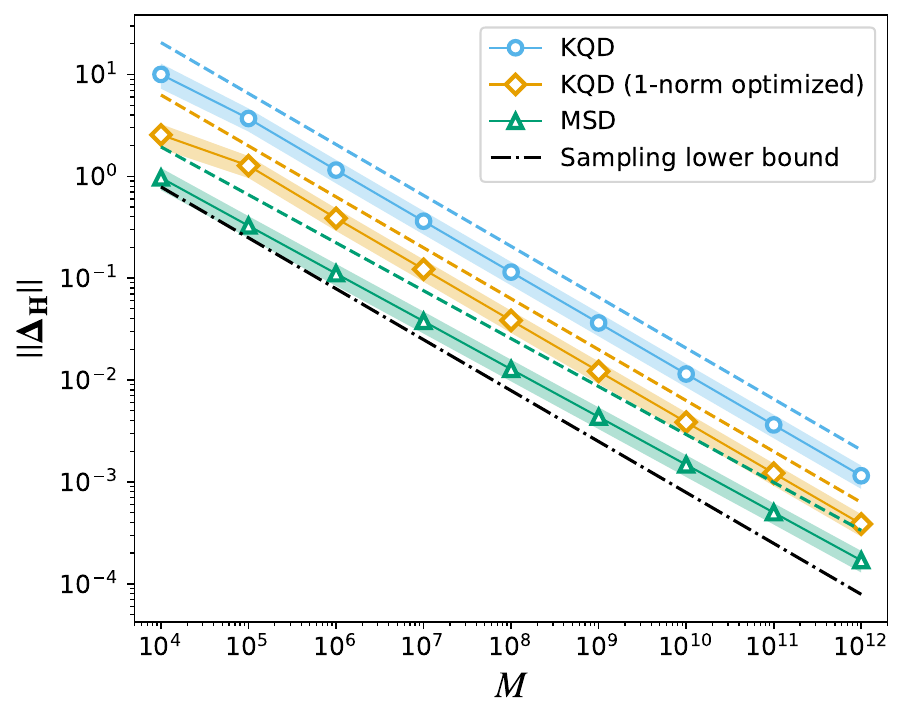}
  \caption{Dependence of matrix perturbation $\norm{\mathbf{\Delta}_{\mathbf{H}}}$ on total shot count $M$ for KQD (conventional and 1-norm optimized) and MSD of the \ce{H2} (cc-pVDZ) model. The markers and shaded regions represent the ensemble average and standard deviation, respectively, calculated from 10,000 random trials. The dashed and dash-dot lines indicate the theoretical upper bounds (Eq.\eqref{eq:msd_matrix_perturbation}) and sampling lower bound (Eq.\eqref{eq:sampling_lower_bound}), respectively. }
  \label{fig:matrix_error}
\end{figure}
Next, we investigate the sampling error of MSD in detail. Figure~\ref{fig:matrix_error} shows the dependence of the matrix perturbation $\norm{\mathbf{\Delta}_{\mathbf{H}}}$ on the total shot count $M$ for both KQD and MSD methods. KQD is applied to the Pauli-LCU obtained from the conventional electronic Hamiltonian (i.e., canonical orbital basis) and the 1-norm optimized electronic Hamiltonian (see Supplementary Section VII for details). MSD is implemented with $J=n=N_{\rm orb} = 8$, following the discussion above. The sampling simulation is conducted using 10,000 different random seeds. 
MSD exhibits a smaller perturbation at each $M$ compared to both the conventional KQD and the 1-norm optimized KQD. Furthermore, the behavior of the matrix perturbations is tightly bounded above by the theoretical prediction (dashed lines). Consequently, MSD achieves a reduction in sampling cost compared to KQD, which aligns with the theoretical argument presented earlier. Moreover, MSD attains a sampling cost comparable to the sampling lower bound (dash-dot line).

\begin{figure}[tbp]
  \centering
  \includegraphics[width=0.48\textwidth]{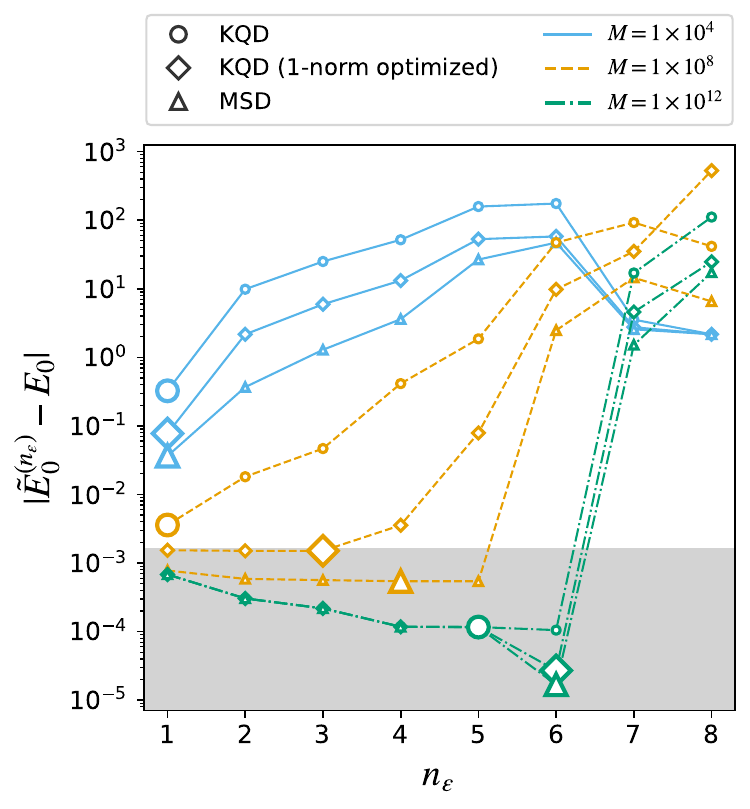}
  \caption{Absolute difference between the exact ground state energy, $E_0$, and perturbed Krylov solutions, $\tilde{E}_0^{(n_\epsilon)}$, for the \ce{H2} (cc-pVDZ) model. Energy values are expressed in Hartree (Ha), and the region corresponding to chemical accuracy (i.e., an absolute energy error below $1.6\times10^{-3}$ Ha) is highlighted with gray shading. Data points are obtained from finite-sampling simulations of KQD (conventional and 1-norm optimized) and MSD methods, varying the thresholding level $\epsilon>0$. The horizontal axis denotes the number of remaining basis vectors after the thresholding, $n_\epsilon\leq n=8$. Different colors represent variations in the total number of shots, $M$. The optimal threshold points, determined as $\epsilon = \max{(\norm{\mathbf{\Delta}_{\mathbf{S}}}, \norm{\mathbf{\Delta}_{\mathbf{H}}}/\norm{\hat{{H}}})}$~\cite{kirby2024analysis}, are emphasized with large markers. }
  \label{fig:thresholding}
\end{figure}
Now, we investigate the performance of MSD for ground-state energy estimation. Subsequent results are obtained from finite-sampling simulations of $\mathbf{S}$ and $\mathbf{H}$, using an equal number of shots, $M$, for each. 
Figure~\ref{fig:thresholding} shows the thresholding behavior (see the Methods for thresholding) of the absolute difference between the exact ground state energy, $E_0$, and perturbed Krylov solutions, $\tilde{E}_0^{(n_\epsilon)}$. The absolute error, $|\tilde{E}_0^{(n_\epsilon)}-E_0|$, varies non-monotonically with changes in the thresholding level, $\epsilon$, and the resulting effective subspace dimension, $n_\epsilon$. Specifically, the perturbation $|\tilde{E}_0^{(n_\epsilon)}-E_0|$ for both KQD and MSD begins to increase at a certain thresholding level due to the ill-conditioning at large subspace dimensions $n_\epsilon$, consistent with previously reported results~\cite{lee2024sampling,epperly2022theory}. This divergent behavior appears at larger $n_\epsilon$ as the total shot count, $M$, increases. We confirmed that this divergence point is well captured by the theoretically optimal thresholding level, $\epsilon = \max{(\norm{\mathbf{\Delta}_{\mathbf{S}}}, \norm{\mathbf{\Delta}_{\mathbf{H}}}/\norm{\hat{{H}}})}$, proposed in Ref.~\cite{kirby2024analysis}, as highlighted by the large markers in Fig.~\ref{fig:thresholding}. 

\begin{figure}[tbp]
  \centering
  \includegraphics[width=0.48\textwidth]{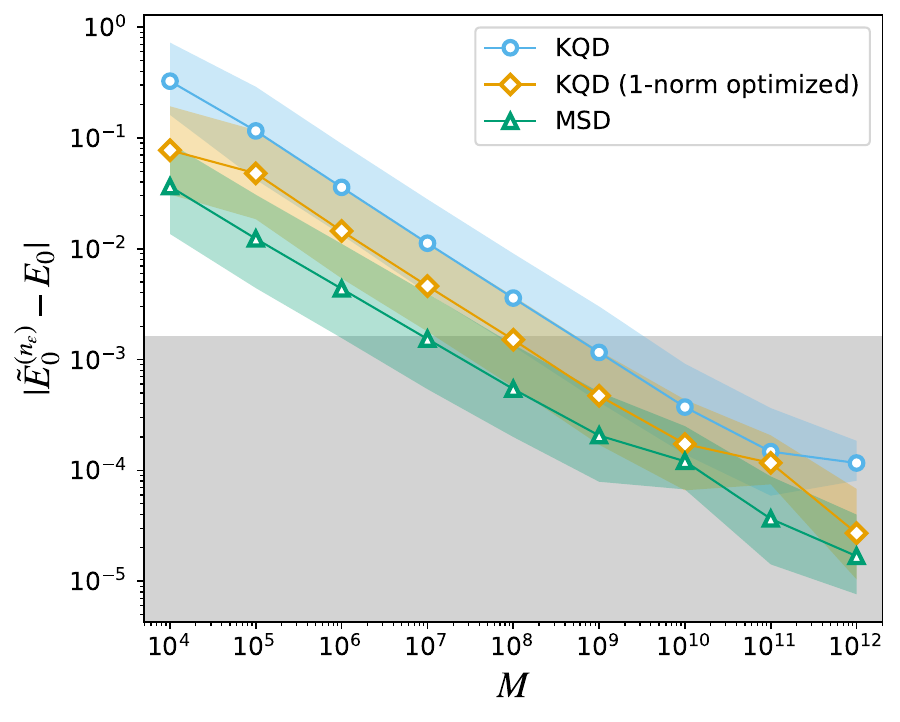}
  \caption{Dependence of energy error $|\tilde{E}_0^{(n_\epsilon)}-E_0|$ on total shot count $M$ for KQD (conventional and 1-norm optimized) and MSD applied to the \ce{H2} (cc-pVDZ) model. The markers and shaded regions represent the ensemble average and full width at half maximum (FWHM), respectively, calculated from 10,000 random trials. Energy values are expressed in Hartree (Ha), and the region corresponding to chemical accuracy (i.e., an absolute energy error below $1.6\times10^{-3}$ Ha) is highlighted with gray shading. The thresholding level $\epsilon$ is determined using the theoretically optimal value $\epsilon = \max{(\norm{\mathbf{\Delta}_{\mathbf{S}}}, \norm{\mathbf{\Delta}_{\mathbf{H}}}/\norm{\hat{{H}}})}$. }
  \label{fig:energy_error}
\end{figure}
Figure~\ref{fig:energy_error} illustrates the dependence of the absolute error, $|\tilde{E}_0^{(n_\epsilon)}-E_0|$, on the total number of shots, $M$, where the thresholding is applied at the optimal level, $\epsilon = \max{(\norm{\mathbf{\Delta}_{\mathbf{S}}}, \norm{\mathbf{\Delta}_{\mathbf{H}}}/\norm{\hat{{H}}})}$. Similar to the matrix perturbation observed in Fig.~\ref{fig:matrix_error}, the energy error $|\tilde{E}_0^{(n_\epsilon)}-E_0|$ also decreases monotonically with increasing total sampling $M$. Notably, MSD achieves a significant reduction in sampling cost compared to KQD. Specifically, MSD attains chemical accuracy with a total shot count approximately 100 (10) times smaller than the conventional KQD (1-norm optimized KQD). 

Combining these results provides an end-to-end comparison under a fixed target accuracy. Since the maximum evolution time of MSD grows by at most a factor of two (Fig.~\ref{fig:simulation_cost}(a)) while $T_{\rm total}/M$ is nearly method-independent for $J = n$, the reduction in shot count directly carries over to the total evolution time. Thus, reaching chemical accuracy with about 100 (10) times fewer shots than conventional (1-norm-optimized) KQD (Fig.~\ref{fig:energy_error}) corresponds to a comparable reduction in the total runtime, at the cost of at most a two-fold increase in circuit depth.

Next, we demonstrate the moment-Lanczos energy error mitigation method introduced above.
When the Krylov order $n$ is insufficient due to the limitations on maximum circuit depth, the corresponding Krylov solution fails to converge. In such cases, further Lanczos iterations using Hamiltonian moments $\bra{\psi_0}\hat{H}^q\ket{\psi_0}$ are generally expected to improve the accuracy of the energy estimate. 
In light of this consideration, subsequent numerical simulations are performed using a small Krylov order of $n=2$. 
\begin{figure*}[tbp]
  \centering
  \includegraphics[width=0.98\textwidth]{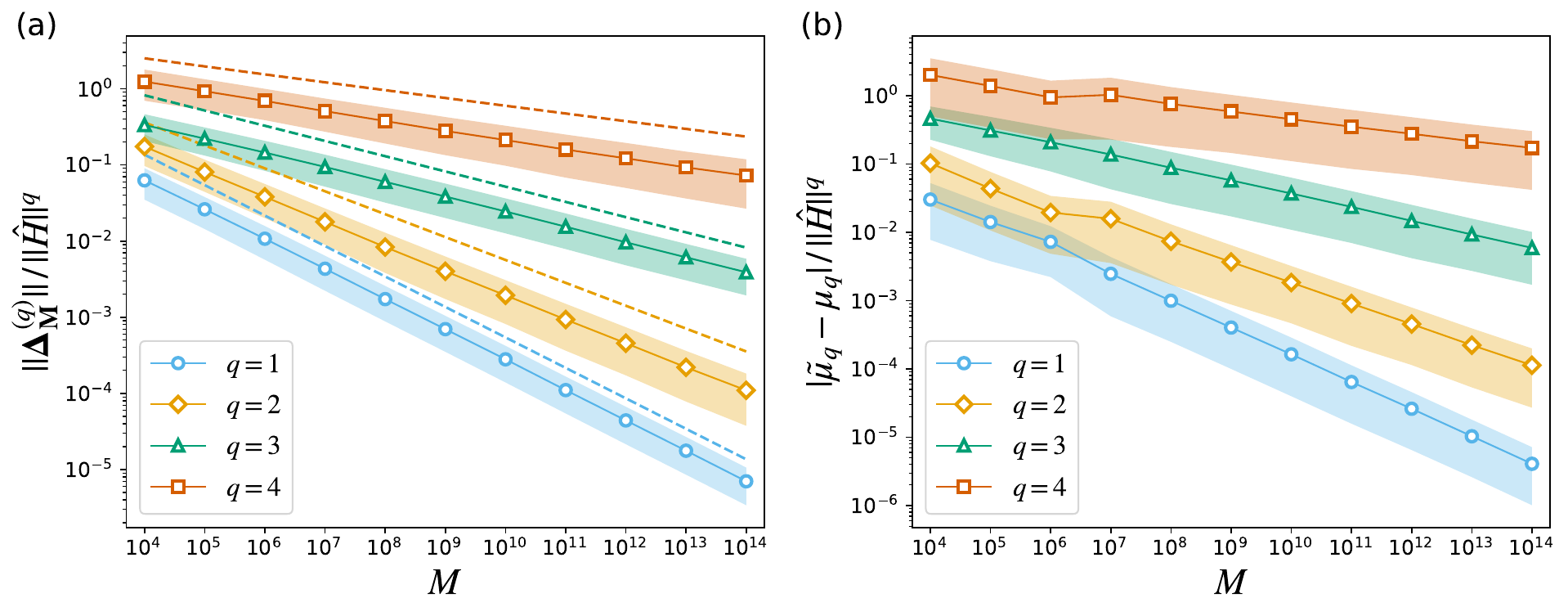}
  \caption{Dependence on total shot count, $M$, of (a) matrix perturbation $\norm{\mathbf{\Delta}^{(q)}_{\mathbf{M}}}/\norm{\hat{H}^q}$ and (b) error in the Hamiltonian moment $|\tilde{\mu}_q-\mu_q|/\norm{\hat{H}^q}$, for the \ce{H2} (6-31G) model. MSD is implemented with $n=2$ and $J=n=2$, resulting in Hamiltonian moments obtained up to order $q=2J=4$. Markers and shaded regions represent ensemble averages and standard deviations, respectively, calculated from 10,000 random trials. Dashed lines in (a) indicate the theoretical upper bound predicted by the bound in Supplementary Section V.  }
  \label{fig:moment_error}
\end{figure*}

Figure~\ref{fig:moment_error}(a) depicts the relationship between the matrix perturbation, $\norm{\mathbf{\Delta}^{(q)}_{\mathbf{M}}}$, and the total number of shots, $M$. With the degree of central finite difference, $J$, set to $J=n=2$, the projected Hamiltonian power matrix, $\mathbf{M}^{(q)}$, is computed up to $q=2J=4$. The results demonstrate a monotonic decrease in the perturbation $\norm{\mathbf{\Delta}^{(q)}_{\mathbf{M}}}$ as the total sampling $M$ increases, for all powers $q$. Furthermore, the perturbation is amplified at higher orders due to the increased finite-difference error, as discussed earlier. The observed numerical behavior remains bounded above by the theoretical prediction (dashed line) provided in Supplementary Section V. 
Figure~\ref{fig:moment_error}(b) illustrates the corresponding sampling perturbation affecting the Hamiltonian moments. Specifically, it quantifies the difference between the perturbed solution, $\tilde{\mu}_q=\tilde{\mathbf{v}}_0^{\dag}\tilde{\mathbf{M}}^{(q)}\tilde{\mathbf{v}}_0/\tilde{\mathbf{v}}_0^{\dag}\tilde{\mathbf{v}}_0$, and the unperturbed solution, $\mu_q=\mathbf{v}_0^{\dag}\mathbf{M}^{(q)}\mathbf{v}_0/\mathbf{v}_0^{\dag}\mathbf{v}_0$, where $\mathbf{v}_0$ represents the exact Krylov solution obtained from a noiseless infinite-sampling simulation of KQD (see Eq.~\eqref{eq:hamiltonian_moment} for refernce). The data reveal that the sampling perturbation on the Hamiltonian moments exhibits a behavior closely resembling that of the projected Hamiltonian power matrix in Fig.~\ref{fig:moment_error}(a).
\begin{figure*}[tbp]
  \centering
  \includegraphics[width=0.98\textwidth]{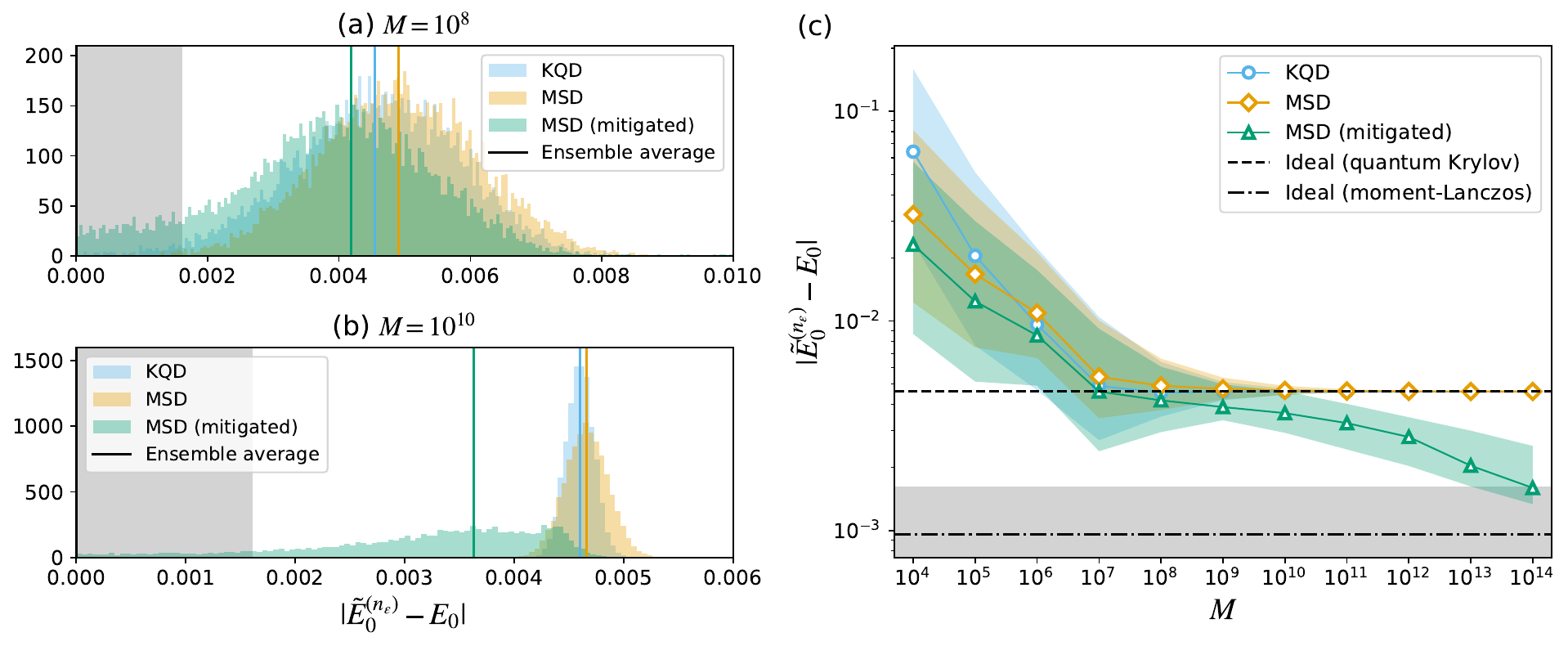}
  \caption{Numerical demonstration of moment-Lanczos error mitigation applied to the \ce{H2} (6-31G) model. The Krylov order is set to $n=2$, and MSD is implemented with $J=n=2$. (a), (b) Histograms depicting the error in ground state energy estimation obtained via KQD, MSD, and moment-Lanczos mitigated MSD. These histograms are generated from 10,000 independent random trials. Vertical lines indicate the ensemble averages derived from these trials. (c) Dependence of energy error on total shot count, $M$. Markers and shaded regions represent the ensemble averages and full width at half maximum (FWHM), respectively, calculated from the 10,000 random trials. Dashed and dash-dot horizontal lines denote the energy error obtained from infinite-sampling simulations of KQD, both with and without moment-Lanczos mitigation, using ideal Hamiltonian moments defined as $\mu_q=\mathbf{v}_0^{\dag}\mathbf{M}^{(q)}\mathbf{v}_0/\mathbf{v}_0^{\dag}\mathbf{v}_0$. Energy values are expressed in Hartree (Ha), and the region corresponding to chemical accuracy (i.e., an absolute energy error below $1.6\times10^{-3}$ Ha) is highlighted with gray shading. }
  \label{fig:lanczos_error}
\end{figure*}

Figure~\ref{fig:lanczos_error} demonstrates energy error mitigation achieved through the moment-Lanczos method. As illustrated in Figs.~\ref{fig:lanczos_error}(a) and (b), the systematic energy error, $|\tilde{E}^{(n_\epsilon)}_0-E_0|$, of MSD is reduced following the Lanczos iteration. This reduction is accompanied by an increase in variance, mirroring the bias-variance trade-off observed in quantum error mitigation methods~\cite{cai2023mitigation}. This increased variance stems from the amplified sampling perturbation affecting higher-order Hamiltonian moments, as detailed in Fig.~\ref{fig:moment_error}.
Figure~\ref{fig:lanczos_error}(c) depicts the relationship between the energy error, $|\tilde{E}^{(n_\epsilon)}_0-E_0|$, and the total shot count, $M$. Given that the Krylov order, $n=2$, is insufficient in this instance, the Krylov solutions (KQD and MSD) fail to converge and plateau around $M=10^9$, remaining above the threshold for chemical accuracy. This contrasts with the result obtained for $n=8$ in Fig.~\ref{fig:energy_error}. Conversely, the moment-Lanczos error mitigation effectively reduces the energy error of MSD. The ensemble-averaged value of the mitigated results exhibits a gradual decrease with increasing total number of shots, $M$, ultimately approaching chemical accuracy.
This demonstration highlights the significant advantage of our moment-Lanczos mitigation method in the NISQ era, where limitations on quantum hardware often restrict the practical application of high Krylov orders. By effectively improving accuracy without requiring deeper circuits, our approach offers a crucial pathway towards obtaining chemically accurate ground-state energies on near-term quantum devices.

\section{Discussion \label{sec:conclusion}}
In this paper, we developed the MSD algorithm for efficient ground-state energy estimation using near-term quantum computers. MSD leverages a central finite-difference approximation of the Hamiltonian operator, expressing it as a linear combination of time-evolution unitaries with symmetrically shifted timesteps. This approach, combined with optimized timestep parameters and an energy spectrum shifting technique, significantly reduces the sampling cost inherent in conventional quantum Krylov methods.

Our theoretical analysis demonstrated that MSD approaches the fundamental lower bound of the sampling cost for quantum Krylov algorithms built upon real-time evolution — a bound set by the 1-norm reduction limit of LCU-based KQD — achieving near-optimal performance up to a logarithmic factor.
This superiority is particularly pronounced when the 1-norm of the Hamiltonian is significantly larger than its spectral norm, a common scenario in molecular electronic structure problems. The sampling cost reduction was corroborated through numerical simulations on the H$_2$ molecule, and further supported by sampling cost estimations for various molecular models. These results showed sampling cost reductions of up to four orders of magnitude compared to the conventional LCU-based KQD, and achieved chemical accuracy with significantly fewer shots. Furthermore, we showed that MSD does not incur a substantial increase in Hamiltonian simulation costs, ensuring that the reduction in sampling effort directly translates into a more efficient total runtime.
Another key innovation of this work is the development of a method to infer higher-order Hamiltonian moments from the propagator matrix elements obtained in MSD. We demonstrated that these moments can be effectively utilized for energy error mitigation based on the Lanczos scheme, particularly valuable when the Krylov order is insufficient due to circuit depth limitations. 

Looking ahead, the MSD algorithm opens several promising avenues for future research.
First, the framework of approximating operators using central finite differences of time-evolution unitaries is highly versatile and could be extended beyond ground state energy estimation. This includes the potential for applying MSD to the calculation of excited states~\cite{cortes2022quantum}, ground state properties~\cite{oumarou2025molecular,zhang2022computing}, real-time dynamics~\cite{cortes2022fast}, and Green's functions~\cite{jamet2021krylov,jamet2022quantum,jamet2025anderson,umeano2025quantum}, by adapting the target observable or the post-processing scheme. Furthermore, the utility of Hamiltonian moments extends beyond ground state energy mitigation, offering potential for their direct application in Green's function calculations~\cite{greene2023quantum} and simulating quantum dynamics~\cite{aulicino2022state}. 
Second, while this work primarily focused on sampling errors, a comprehensive analysis of other error sources, such as Hamiltonian simulation errors (e.g., Trotter error) and hardware noise, is crucial for practical implementations. Investigating the robustness of MSD against these errors and developing tailored mitigation strategies would be important. The inherent compatibility of Hamiltonian moments with various error mitigation techniques based on Hamiltonian moments~\cite{vallury2020quantum,suchsland2021algorithmic} also warrants further investigation, exploring how these moments can be integrated into broader error mitigation frameworks to enhance overall performance.
It would also be valuable to combine MSD with more elaborate initial-state preparations—such as multireference or chemically motivated reference states~\cite{stair2020multireference,dcunha2024fragment}—which could reduce the Krylov dimension and hence the total cost required for strongly correlated systems, complementing the sampling-cost reduction achieved here.
Finally, to fully assess the practical utility of MSD, detailed resource estimations for actual quantum devices are desirable. This includes quantifying the gate counts, circuit depths, and actual runtime required to simulate utility-scale molecular and materials systems. Such rigorous resource analysis would enable a direct and fair comparison of MSD with other state-of-the-art algorithms that utilize time-evolution operators for ground state computation, such as sample-based quantum Krylov methods~\cite{sugisaki2024hamiltonian,mikkelsen2024quantum,yu2025quantum} and early-FTQC oriented phase estimation algorithms~\cite{lin2022heisenberg,wan2022randomized,wang2023quantum,ding2023qcels,ni2023low,kshirsagar2024proving,liang2024modeling,wang2025efficient}, providing a clearer roadmap for its deployment on future quantum hardware.

\section{Methods}
\subsection{Measurement of matrix elements}
In KQD, the matrices $\mathbf{S}$ and $\mathbf{H}$ are measured using quantum computers. Based on the LCU decomposition of the Hamiltonian [Eq.~\eqref{eq:LCU}], each projected Hamiltonian matrix element is expressed as
\begin{align}
    \mathbf{H}_{k',k} &= \sum_{l=1}^{L}c_l\bra{\phi_0}\hat{P
    }_le^{-i\hat{H}(k-k')\tau}\ket{\phi_0}, \label{eq:H_matrix_LCU}
\end{align}
where we have assumed the Toeplitz structure of the matrix. Note that the first term in Eq.~\eqref{eq:LCU} is ignored in Eq.~\eqref{eq:H_matrix_LCU} since it is just an energy constant, which is finally added to the eigenvalue estimate obtained by solving the GEVP. 

Each Pauli matrix element in Eq.~\eqref{eq:H_matrix_LCU} can be measured by performing the Hadamard test using the quantum circuit in Fig.~\ref{fig:measurement_circuit}. 
More specifically, the quantum state right before the measurement in the circuit in Fig.~\ref{fig:measurement_circuit} is described as follows: 
\begin{align}
    \ket{\Phi_t} = \frac{1}{\sqrt{2}}\left( \ket{0}_a\ket{\phi_0} + \ket{1}_a e^{-i\hat{H}t}\ket{\phi_0}\right), 
\end{align}
where $\ket{0}_a$ and $\ket{1}_a$ denote the computational basis states in the ancilla register. Then, the Pauli basis measurement at the end of the circuit gives the following expectation values: 
\begin{align}
    \bra{\Phi_t}\hat{X}\otimes\hat{P}\ket{\Phi_t}&= \mathrm{Re}\left[ \bra{\phi_0} \hat{P}e^{-i\hat{H}t}\ket{\phi_0}\right], \\
    \bra{\Phi_t}\hat{Y}\otimes\hat{P}\ket{\Phi_t}&= \mathrm{Im}\left[ \bra{\phi_0} \hat{P}e^{-i\hat{H}t}\ket{\phi_0}\right]. 
\end{align}
Therefore, we can obtain the projected Hamiltonian matrix~\eqref{eq:H_matrix_LCU} by performing the above Hadamard test for $\hat{P}=\hat{P_l}$ ($l\in\{1,\cdots,L\}$) and $t=k\tau$ ($k\in\{0,\cdots,n-1\}$). If we set $\hat{P}=\hat{\1}$, the measurement outcome of the ancilla register gives the overlap matrix $\mathbf{S}$ in Eq.~\eqref{eq:S_matrix}. 
Here, we remark that the Pauli LCU decomposition simplifies the gate operations in the Hadamard test. Specifically, the multi-qubit Pauli operation $\hat{P}$ can be implemented as a corresponding Pauli measurement, as shown in Fig.~\ref{fig:measurement_circuit}. If we apply an LCU decomposition other than the Pauli LCU, additional controlled gate operations are required after the controlled time-evolution. 

\begin{figure}[tbp]
    \begin{quantikz}[thin lines]
      \lstick{{$\ket{0}_a$}} & & \gate{H} & \ctrl{1} & \meter{\hat{X},\hat{Y}} \\
      \lstick{{$\ket{0}^{\otimes N_q}$}} & \qwbundle{N_q} & \gate{\hat{V}_{\phi_0}} & \gate{e^{-i\hat{H}t}} & \meter{\hat{P}} \\
    \end{quantikz}
    \vspace{0.2cm}
    \caption{Circuit diagrams for estimating the matrix element $\bra{\phi_0}\hat{P}e^{-i\hat{H}t}\ket{\phi_0}$. 
    Here, $\hat{V}_{\phi_0}$ denotes a quantum circuit to prepare the initial state $\ket{\phi_0}$ as $\hat{V}_{\phi_0}\ket{0}^{\otimes N_q} = \ket{\phi_0}$. 
    \label{fig:measurement_circuit}} 
\end{figure}

\subsection{Regularization of the perturbed GEVP}\label{subsec:GEVP_perturbation}
As introduced in the Results [the perturbed GEVP, Eq.~\eqref{eq:GEVP_noisy}], the noisy matrix pair $(\tilde{\mathbf{H}},\tilde{\mathbf{S}})$ yields an ill-conditioned problem. To regularize it, we perform an eigenvalue decomposition of $\tilde{\mathbf{S}}$ and discard the eigenvectors whose eigenvalues are smaller than or equal to a threshold level $\epsilon>0$. Denoting by $\tilde{\mathbf{V}}_{>\epsilon}$ the matrix of the retained eigenvectors, we solve the reduced $n_\epsilon\times n_\epsilon$ problem
\begin{align}
    \tilde{\mathbf{V}}_{>\epsilon}^{\dag}\tilde{\mathbf{H}}\tilde{\mathbf{V}}_{>\epsilon}\tilde{\mathbf{v}}_j
    = \tilde{E}_j^{(n_{\epsilon})}\tilde{\mathbf{V}}_{>\epsilon}^{\dag}\tilde{\mathbf{S}}\tilde{\mathbf{V}}_{>\epsilon}\tilde{\mathbf{v}}_j,
    \label{eq:GEVP_threshold}
\end{align}
instead of the original $n\times n$ problem, where $n_\epsilon\leq n$ is the number of retained basis vectors. In our numerical simulations, we set the threshold to the theoretically optimal level $\epsilon=\max(\norm{\mathbf{\Delta}_{\mathbf{S}}},\norm{\mathbf{\Delta}_{\mathbf{H}}}/\norm{\hat{H}})$~\cite{kirby2024analysis}.

\subsection{Numerical simulation details}
We constructed the second-quantized electronic structure Hamiltonians for \ce{H2}, \ce{LiH}, \ce{BeH2}, \ce{H2O}, \ce{NH3}, \ce{N2}, \ce{Cr2}, and \ce{C4H4N2} (pyrazine) at their equilibrium geometries for the various basis sets. The one- and two-electron integrals were obtained from self-consistent restricted Hartree-Fock calculations using the \texttt{PySCF} library~\cite{sun2018pyscf}, and the fermionic Hamiltonians were mapped to qubits by the Jordan-Wigner transformation~\cite{Jordan1928}. The frozen-core approximation was adopted for all molecules except \ce{H2}, with the 1s molecular orbitals frozen. The 1-norm $\lambda$ was reduced using orbital optimization~\cite{koridon2021orbital} and BLISS~\cite{loaiza2023bliss,patel2024blisslp}; the formulation, analytic gradients, and optimizer settings are given in Supplementary Section VII. The spectral range $\Delta{E}_{\bm{Q}}$ was obtained by exact diagonalization (full configuration interaction) for small models, and by the HF-level orbital-optimization estimate of Ref.~\cite{cortes2024spectral} for large models. Detailed data and system sizes for all molecular Hamiltonians are summarized in Supplementary Section VIII.

We benchmark MSD primarily on the \ce{H2} molecule using the cc-pVDZ or 6-31G basis set, and also present results for the other molecular models above in Supplementary Section IX to further illustrate its broad applicability. For the cc-pVDZ basis set, we adopt an active-space selection in which the two highest-energy orbitals are eliminated to reduce the numerical simulation cost, so that the resulting Hamiltonian contains 2 electrons and 8 (4) spatial orbitals, i.e., 16 (8) spin orbitals, for the cc-pVDZ (6-31G) basis set model. The \ce{H2} molecule is chosen because its small number of electrons allows for classical simulation even with large basis sets without prohibitive computational cost, making it an ideal system for demonstrating the significant sampling-cost reduction of MSD over KQD in the large-basis-set regime.

The state vector simulations were performed with the \texttt{ffsim} library~\cite{ffsim} within a subspace of fixed particle number $N_e$ and total spin $S$ (for the \ce{H2} model, $N_e=2$ and $S=0$). The Krylov subspace is constructed from the Hartree-Fock state $\ket{\phi_0}$ with the timestep $\tau=\pi/\Delta{E}_{\bm{Q}}$~\cite{epperly2022theory}, where $\bm{Q}=(N_e,S)$. The Krylov order $n$ is set equal to the number of orbitals $N_{\rm orb}$ (i.e., $n=N_{\rm orb}=8$ for the cc-pVDZ model) unless otherwise mentioned. To focus on the sampling error, we perform noiseless simulations using the exact time-evolution operator, and the central finite-difference coefficients are calculated using the \texttt{findiff} library~\cite{findiff}.

\section*{Data Availability}
The datasets generated and/or analysed during the current study are available in the mirror-subspace-diagonalization repository, \url{https://github.com/FujitsuResearch/mirror-subspace-diagonalization}. 

\section*{Code Availability}
The underlying code for this study is available in mirror-subspace-diagonalization repository and can be accessed via this link \url{https://github.com/FujitsuResearch/mirror-subspace-diagonalization}. 

\section*{Acknowledgements}
We are profoundly indebted to Shintaro Sato and Tennin Yan for their invaluable support in fostering a conducive research environment. 
S.K. acknowledges Riki Toshio for his helpful comments and fruitful discussions. 
S.K. thanks Masatoshi Ishii for technical support in performing numerical simulations. 
No funding was received for this research. 

\section*{Author Contributions}
S. K. was responsible for conceptualizing the detailed theoretical framework, developing the methodology, and performing numerical simulations for all figures and tables. S. K. also undertook the primary investigation and formal analysis of the data and prepared the original draft of the manuscript. Y. O. N. contributed to the conceptualization of the theory and provided technical supervision throughout the work. N. M. and Y. H. participated in technical discussions and reviewed the manuscript. K. M. and H. O. were instrumental in supervising the project, securing necessary resources, and overseeing project administration, in addition to reviewing the manuscript.

\section*{Competing interests}
The authors declare no competing financial or non-financial interests.

\FloatBarrier

\vspace{1em}
\noindent{\large\bfseries References}
\vspace{0.5em}
\bibliography{main}

\providecommand{\noopsort}[1]{}\providecommand{\singleletter}[1]{#1}%
\begin{thebibliography}{86}%
\makeatletter
\providecommand \@ifxundefined [1]{%
 \@ifx{#1\undefined}
}%
\providecommand \@ifnum [1]{%
 \ifnum #1\expandafter \@firstoftwo
 \else \expandafter \@secondoftwo
 \fi
}%
\providecommand \@ifx [1]{%
 \ifx #1\expandafter \@firstoftwo
 \else \expandafter \@secondoftwo
 \fi
}%
\providecommand \natexlab [1]{#1}%
\providecommand \enquote  [1]{``#1''}%
\providecommand \bibnamefont  [1]{#1}%
\providecommand \bibfnamefont [1]{#1}%
\providecommand \citenamefont [1]{#1}%
\providecommand \href@noop [0]{\@secondoftwo}%
\providecommand \href [0]{\begingroup \@sanitize@url \@href}%
\providecommand \@href[1]{\@@startlink{#1}\@@href}%
\providecommand \@@href[1]{\endgroup#1\@@endlink}%
\providecommand \@sanitize@url [0]{\catcode `\\12\catcode `\$12\catcode `\&12\catcode `\#12\catcode `\^12\catcode `\_12\catcode `\%12\relax}%
\providecommand \@@startlink[1]{}%
\providecommand \@@endlink[0]{}%
\providecommand \url  [0]{\begingroup\@sanitize@url \@url }%
\providecommand \@url [1]{\endgroup\@href {#1}{\urlprefix }}%
\providecommand \urlprefix  [0]{URL }%
\providecommand \Eprint [0]{\href }%
\providecommand \doibase [0]{https://doi.org/}%
\providecommand \selectlanguage [0]{\@gobble}%
\providecommand \bibinfo  [0]{\@secondoftwo}%
\providecommand \bibfield  [0]{\@secondoftwo}%
\providecommand \translation [1]{[#1]}%
\providecommand \BibitemOpen [0]{}%
\providecommand \bibitemStop [0]{}%
\providecommand \bibitemNoStop [0]{.\EOS\space}%
\providecommand \EOS [0]{\spacefactor3000\relax}%
\providecommand \BibitemShut  [1]{\csname bibitem#1\endcsname}%
\let\auto@bib@innerbib\@empty
\bibitem [{\citenamefont {Cao}\ \emph {et~al.}(2019)\citenamefont {Cao}, \citenamefont {Romero}, \citenamefont {Olson}, \citenamefont {Degroote}, \citenamefont {Johnson}, \citenamefont {Kieferov{\'a}}, \citenamefont {Kivlichan}, \citenamefont {Menke}, \citenamefont {Peropadre}, \citenamefont {Sawaya} \emph {et~al.}}]{cao2019quantum}%
  \BibitemOpen
  \bibfield  {author} {\bibinfo {author} {\bibfnamefont {Y.}~\bibnamefont {Cao}}, \bibinfo {author} {\bibfnamefont {J.}~\bibnamefont {Romero}}, \bibinfo {author} {\bibfnamefont {J.~P.}\ \bibnamefont {Olson}}, \bibinfo {author} {\bibfnamefont {M.}~\bibnamefont {Degroote}}, \bibinfo {author} {\bibfnamefont {P.~D.}\ \bibnamefont {Johnson}}, \bibinfo {author} {\bibfnamefont {M.}~\bibnamefont {Kieferov{\'a}}}, \bibinfo {author} {\bibfnamefont {I.~D.}\ \bibnamefont {Kivlichan}}, \bibinfo {author} {\bibfnamefont {T.}~\bibnamefont {Menke}}, \bibinfo {author} {\bibfnamefont {B.}~\bibnamefont {Peropadre}}, \bibinfo {author} {\bibfnamefont {N.~P.}\ \bibnamefont {Sawaya}}, \emph {et~al.},\ }\bibfield  {title} {\bibinfo {title} {Quantum chemistry in the age of quantum computing},\ }\href {https://pubs.acs.org/doi/10.1021/acs.chemrev.8b00803} {\bibfield  {journal} {\bibinfo  {journal} {Chemical reviews}\ }\textbf {\bibinfo {volume} {119}},\ \bibinfo {pages} {10856} (\bibinfo {year} {2019})}\BibitemShut {NoStop}%
\bibitem [{\citenamefont {McArdle}\ \emph {et~al.}(2020)\citenamefont {McArdle}, \citenamefont {Endo}, \citenamefont {Aspuru-Guzik}, \citenamefont {Benjamin},\ and\ \citenamefont {Yuan}}]{mcardle2020quantum}%
  \BibitemOpen
  \bibfield  {author} {\bibinfo {author} {\bibfnamefont {S.}~\bibnamefont {McArdle}}, \bibinfo {author} {\bibfnamefont {S.}~\bibnamefont {Endo}}, \bibinfo {author} {\bibfnamefont {A.}~\bibnamefont {Aspuru-Guzik}}, \bibinfo {author} {\bibfnamefont {S.~C.}\ \bibnamefont {Benjamin}},\ and\ \bibinfo {author} {\bibfnamefont {X.}~\bibnamefont {Yuan}},\ }\bibfield  {title} {\bibinfo {title} {Quantum computational chemistry},\ }\href {https://doi.org/10.1103/RevModPhys.92.015003} {\bibfield  {journal} {\bibinfo  {journal} {Rev. Mod. Phys.}\ }\textbf {\bibinfo {volume} {92}},\ \bibinfo {pages} {015003} (\bibinfo {year} {2020})}\BibitemShut {NoStop}%
\bibitem [{\citenamefont {Bauer}\ \emph {et~al.}(2020)\citenamefont {Bauer}, \citenamefont {Bravyi}, \citenamefont {Motta},\ and\ \citenamefont {Chan}}]{bauer2020quantum}%
  \BibitemOpen
  \bibfield  {author} {\bibinfo {author} {\bibfnamefont {B.}~\bibnamefont {Bauer}}, \bibinfo {author} {\bibfnamefont {S.}~\bibnamefont {Bravyi}}, \bibinfo {author} {\bibfnamefont {M.}~\bibnamefont {Motta}},\ and\ \bibinfo {author} {\bibfnamefont {G.~K.-L.}\ \bibnamefont {Chan}},\ }\bibfield  {title} {\bibinfo {title} {Quantum algorithms for quantum chemistry and quantum materials science},\ }\href {https://pubs.acs.org/doi/10.1021/acs.chemrev.9b00829} {\bibfield  {journal} {\bibinfo  {journal} {Chemical reviews}\ }\textbf {\bibinfo {volume} {120}},\ \bibinfo {pages} {12685} (\bibinfo {year} {2020})}\BibitemShut {NoStop}%
\bibitem [{\citenamefont {Kitaev}(1995)}]{kitaev1995quantum}%
  \BibitemOpen
  \bibfield  {author} {\bibinfo {author} {\bibfnamefont {A.~Y.}\ \bibnamefont {Kitaev}},\ }\bibfield  {title} {\bibinfo {title} {{Quantum measurements and the Abelian stabilizer problem}},\ }\href {https://arxiv.org/abs/quant-ph/9511026} {\bibfield  {journal} {\bibinfo  {journal} {arXiv:9511026}\ } (\bibinfo {year} {1995})}\BibitemShut {NoStop}%
\bibitem [{\citenamefont {Kitaev}\ \emph {et~al.}(2002)\citenamefont {Kitaev}, \citenamefont {Shen},\ and\ \citenamefont {Vyalyi}}]{kitaev2002classical}%
  \BibitemOpen
  \bibfield  {author} {\bibinfo {author} {\bibfnamefont {A.~Y.}\ \bibnamefont {Kitaev}}, \bibinfo {author} {\bibfnamefont {A.}~\bibnamefont {Shen}},\ and\ \bibinfo {author} {\bibfnamefont {M.~N.}\ \bibnamefont {Vyalyi}},\ }\href {https://www.ams.org/books/gsm/047/} {\emph {\bibinfo {title} {Classical and quantum computation}}},\ \bibinfo {number} {47}\ (\bibinfo  {publisher} {American Mathematical Soc.},\ \bibinfo {year} {2002})\BibitemShut {NoStop}%
\bibitem [{\citenamefont {Chan}\ and\ \citenamefont {Sharma}(2011)}]{chan2011density}%
  \BibitemOpen
  \bibfield  {author} {\bibinfo {author} {\bibfnamefont {G.~K.-L.}\ \bibnamefont {Chan}}\ and\ \bibinfo {author} {\bibfnamefont {S.}~\bibnamefont {Sharma}},\ }\bibfield  {title} {\bibinfo {title} {The density matrix renormalization group in quantum chemistry},\ }\href {https://www.annualreviews.org/content/journals/10.1146/annurev-physchem-032210-103338} {\bibfield  {journal} {\bibinfo  {journal} {Annual review of physical chemistry}\ }\textbf {\bibinfo {volume} {62}},\ \bibinfo {pages} {465} (\bibinfo {year} {2011})}\BibitemShut {NoStop}%
\bibitem [{\citenamefont {Baiardi}\ and\ \citenamefont {Reiher}(2020)}]{baiardi2020density}%
  \BibitemOpen
  \bibfield  {author} {\bibinfo {author} {\bibfnamefont {A.}~\bibnamefont {Baiardi}}\ and\ \bibinfo {author} {\bibfnamefont {M.}~\bibnamefont {Reiher}},\ }\bibfield  {title} {\bibinfo {title} {The density matrix renormalization group in chemistry and molecular physics: Recent developments and new challenges},\ }\href {https://pubs.aip.org/aip/jcp/article/152/4/040903/76348/The-density-matrix-renormalization-group-in} {\bibfield  {journal} {\bibinfo  {journal} {The Journal of Chemical Physics}\ }\textbf {\bibinfo {volume} {152}} (\bibinfo {year} {2020})}\BibitemShut {NoStop}%
\bibitem [{\citenamefont {Bartlett}\ and\ \citenamefont {Musia\l{}}(2007)}]{bartlett2007coupled}%
  \BibitemOpen
  \bibfield  {author} {\bibinfo {author} {\bibfnamefont {R.~J.}\ \bibnamefont {Bartlett}}\ and\ \bibinfo {author} {\bibfnamefont {M.}~\bibnamefont {Musia\l{}}},\ }\bibfield  {title} {\bibinfo {title} {Coupled-cluster theory in quantum chemistry},\ }\href {https://doi.org/10.1103/RevModPhys.79.291} {\bibfield  {journal} {\bibinfo  {journal} {Rev. Mod. Phys.}\ }\textbf {\bibinfo {volume} {79}},\ \bibinfo {pages} {291} (\bibinfo {year} {2007})}\BibitemShut {NoStop}%
\bibitem [{\citenamefont {Lee}\ \emph {et~al.}(2021)\citenamefont {Lee}, \citenamefont {Berry}, \citenamefont {Gidney}, \citenamefont {Huggins}, \citenamefont {McClean}, \citenamefont {Wiebe},\ and\ \citenamefont {Babbush}}]{lee2021femoco}%
  \BibitemOpen
  \bibfield  {author} {\bibinfo {author} {\bibfnamefont {J.}~\bibnamefont {Lee}}, \bibinfo {author} {\bibfnamefont {D.~W.}\ \bibnamefont {Berry}}, \bibinfo {author} {\bibfnamefont {C.}~\bibnamefont {Gidney}}, \bibinfo {author} {\bibfnamefont {W.~J.}\ \bibnamefont {Huggins}}, \bibinfo {author} {\bibfnamefont {J.~R.}\ \bibnamefont {McClean}}, \bibinfo {author} {\bibfnamefont {N.}~\bibnamefont {Wiebe}},\ and\ \bibinfo {author} {\bibfnamefont {R.}~\bibnamefont {Babbush}},\ }\bibfield  {title} {\bibinfo {title} {{Even More Efficient Quantum Computations of Chemistry Through Tensor Hypercontraction}},\ }\href {https://doi.org/10.1103/PRXQuantum.2.030305} {\bibfield  {journal} {\bibinfo  {journal} {PRX Quantum}\ }\textbf {\bibinfo {volume} {2}},\ \bibinfo {pages} {030305} (\bibinfo {year} {2021})}\BibitemShut {NoStop}%
\bibitem [{\citenamefont {Goings}\ \emph {et~al.}(2022)\citenamefont {Goings}, \citenamefont {White}, \citenamefont {Lee}, \citenamefont {Tautermann}, \citenamefont {Degroote}, \citenamefont {Gidney}, \citenamefont {Shiozaki}, \citenamefont {Babbush},\ and\ \citenamefont {Rubin}}]{goings2022p450}%
  \BibitemOpen
  \bibfield  {author} {\bibinfo {author} {\bibfnamefont {J.~J.}\ \bibnamefont {Goings}}, \bibinfo {author} {\bibfnamefont {A.}~\bibnamefont {White}}, \bibinfo {author} {\bibfnamefont {J.}~\bibnamefont {Lee}}, \bibinfo {author} {\bibfnamefont {C.~S.}\ \bibnamefont {Tautermann}}, \bibinfo {author} {\bibfnamefont {M.}~\bibnamefont {Degroote}}, \bibinfo {author} {\bibfnamefont {C.}~\bibnamefont {Gidney}}, \bibinfo {author} {\bibfnamefont {T.}~\bibnamefont {Shiozaki}}, \bibinfo {author} {\bibfnamefont {R.}~\bibnamefont {Babbush}},\ and\ \bibinfo {author} {\bibfnamefont {N.~C.}\ \bibnamefont {Rubin}},\ }\bibfield  {title} {\bibinfo {title} {Reliably assessing the electronic structure of cytochrome p450 on today’s classical computers and tomorrow’s quantum computers},\ }\href {https://www.pnas.org/doi/10.1073/pnas.2203533119} {\bibfield  {journal} {\bibinfo  {journal} {Proceedings of the National Academy of Sciences}\ }\textbf {\bibinfo {volume} {119}},\ \bibinfo {pages} {e2203533119} (\bibinfo {year}
  {2022})}\BibitemShut {NoStop}%
\bibitem [{\citenamefont {Preskill}(2018)}]{preskill2018quantum}%
  \BibitemOpen
  \bibfield  {author} {\bibinfo {author} {\bibfnamefont {J.}~\bibnamefont {Preskill}},\ }\bibfield  {title} {\bibinfo {title} {{Quantum computing in the {NISQ} era and beyond}},\ }\href {https://quantum-journal.org/papers/q-2018-08-06-79/} {\bibfield  {journal} {\bibinfo  {journal} {Quantum}\ }\textbf {\bibinfo {volume} {2}},\ \bibinfo {pages} {79} (\bibinfo {year} {2018})}\BibitemShut {NoStop}%
\bibitem [{\citenamefont {Peruzzo}\ \emph {et~al.}(2014)\citenamefont {Peruzzo}, \citenamefont {McClean}, \citenamefont {Shadbolt}, \citenamefont {Yung}, \citenamefont {Zhou}, \citenamefont {Love}, \citenamefont {Aspuru-Guzik},\ and\ \citenamefont {O'brien}}]{peruzzo2014variational}%
  \BibitemOpen
  \bibfield  {author} {\bibinfo {author} {\bibfnamefont {A.}~\bibnamefont {Peruzzo}}, \bibinfo {author} {\bibfnamefont {J.}~\bibnamefont {McClean}}, \bibinfo {author} {\bibfnamefont {P.}~\bibnamefont {Shadbolt}}, \bibinfo {author} {\bibfnamefont {M.-H.}\ \bibnamefont {Yung}}, \bibinfo {author} {\bibfnamefont {X.-Q.}\ \bibnamefont {Zhou}}, \bibinfo {author} {\bibfnamefont {P.~J.}\ \bibnamefont {Love}}, \bibinfo {author} {\bibfnamefont {A.}~\bibnamefont {Aspuru-Guzik}},\ and\ \bibinfo {author} {\bibfnamefont {J.~L.}\ \bibnamefont {O'brien}},\ }\bibfield  {title} {\bibinfo {title} {{A variational eigenvalue solver on a photonic quantum processor}},\ }\href {https://www.nature.com/articles/ncomms5213} {\bibfield  {journal} {\bibinfo  {journal} {Nat. Commun.}\ }\textbf {\bibinfo {volume} {5}},\ \bibinfo {pages} {4213} (\bibinfo {year} {2014})}\BibitemShut {NoStop}%
\bibitem [{\citenamefont {McClean}\ \emph {et~al.}(2018)\citenamefont {McClean}, \citenamefont {Boixo}, \citenamefont {Smelyanskiy}, \citenamefont {Babbush},\ and\ \citenamefont {Neven}}]{mcclean2018barren}%
  \BibitemOpen
  \bibfield  {author} {\bibinfo {author} {\bibfnamefont {J.~R.}\ \bibnamefont {McClean}}, \bibinfo {author} {\bibfnamefont {S.}~\bibnamefont {Boixo}}, \bibinfo {author} {\bibfnamefont {V.~N.}\ \bibnamefont {Smelyanskiy}}, \bibinfo {author} {\bibfnamefont {R.}~\bibnamefont {Babbush}},\ and\ \bibinfo {author} {\bibfnamefont {H.}~\bibnamefont {Neven}},\ }\bibfield  {title} {\bibinfo {title} {Barren plateaus in quantum neural network training landscapes},\ }\href {https://www.nature.com/articles/s41467-018-07090-4} {\bibfield  {journal} {\bibinfo  {journal} {Nature communications}\ }\textbf {\bibinfo {volume} {9}},\ \bibinfo {pages} {4812} (\bibinfo {year} {2018})}\BibitemShut {NoStop}%
\bibitem [{\citenamefont {Cerezo}\ \emph {et~al.}(2021)\citenamefont {Cerezo}, \citenamefont {Sone}, \citenamefont {Volkoff}, \citenamefont {Cincio},\ and\ \citenamefont {Coles}}]{cerezo2021cost}%
  \BibitemOpen
  \bibfield  {author} {\bibinfo {author} {\bibfnamefont {M.}~\bibnamefont {Cerezo}}, \bibinfo {author} {\bibfnamefont {A.}~\bibnamefont {Sone}}, \bibinfo {author} {\bibfnamefont {T.}~\bibnamefont {Volkoff}}, \bibinfo {author} {\bibfnamefont {L.}~\bibnamefont {Cincio}},\ and\ \bibinfo {author} {\bibfnamefont {P.~J.}\ \bibnamefont {Coles}},\ }\bibfield  {title} {\bibinfo {title} {Cost function dependent barren plateaus in shallow parametrized quantum circuits},\ }\href {https://www.nature.com/articles/s41467-021-21728-w} {\bibfield  {journal} {\bibinfo  {journal} {Nature communications}\ }\textbf {\bibinfo {volume} {12}},\ \bibinfo {pages} {1791} (\bibinfo {year} {2021})}\BibitemShut {NoStop}%
\bibitem [{\citenamefont {Wang}\ \emph {et~al.}(2021)\citenamefont {Wang}, \citenamefont {Fontana}, \citenamefont {Cerezo}, \citenamefont {Sharma}, \citenamefont {Sone}, \citenamefont {Cincio},\ and\ \citenamefont {Coles}}]{wang2021noise}%
  \BibitemOpen
  \bibfield  {author} {\bibinfo {author} {\bibfnamefont {S.}~\bibnamefont {Wang}}, \bibinfo {author} {\bibfnamefont {E.}~\bibnamefont {Fontana}}, \bibinfo {author} {\bibfnamefont {M.}~\bibnamefont {Cerezo}}, \bibinfo {author} {\bibfnamefont {K.}~\bibnamefont {Sharma}}, \bibinfo {author} {\bibfnamefont {A.}~\bibnamefont {Sone}}, \bibinfo {author} {\bibfnamefont {L.}~\bibnamefont {Cincio}},\ and\ \bibinfo {author} {\bibfnamefont {P.~J.}\ \bibnamefont {Coles}},\ }\bibfield  {title} {\bibinfo {title} {Noise-induced barren plateaus in variational quantum algorithms},\ }\href {https://www.nature.com/articles/s41467-021-27045-6} {\bibfield  {journal} {\bibinfo  {journal} {Nature communications}\ }\textbf {\bibinfo {volume} {12}},\ \bibinfo {pages} {6961} (\bibinfo {year} {2021})}\BibitemShut {NoStop}%
\bibitem [{\citenamefont {Lee}\ \emph {et~al.}(2023)\citenamefont {Lee}, \citenamefont {Lee}, \citenamefont {Zhai}, \citenamefont {Tong}, \citenamefont {Dalzell}, \citenamefont {Kumar}, \citenamefont {Helms}, \citenamefont {Gray}, \citenamefont {Cui}, \citenamefont {Liu} \emph {et~al.}}]{lee2023evaluating}%
  \BibitemOpen
  \bibfield  {author} {\bibinfo {author} {\bibfnamefont {S.}~\bibnamefont {Lee}}, \bibinfo {author} {\bibfnamefont {J.}~\bibnamefont {Lee}}, \bibinfo {author} {\bibfnamefont {H.}~\bibnamefont {Zhai}}, \bibinfo {author} {\bibfnamefont {Y.}~\bibnamefont {Tong}}, \bibinfo {author} {\bibfnamefont {A.~M.}\ \bibnamefont {Dalzell}}, \bibinfo {author} {\bibfnamefont {A.}~\bibnamefont {Kumar}}, \bibinfo {author} {\bibfnamefont {P.}~\bibnamefont {Helms}}, \bibinfo {author} {\bibfnamefont {J.}~\bibnamefont {Gray}}, \bibinfo {author} {\bibfnamefont {Z.-H.}\ \bibnamefont {Cui}}, \bibinfo {author} {\bibfnamefont {W.}~\bibnamefont {Liu}}, \emph {et~al.},\ }\bibfield  {title} {\bibinfo {title} {Evaluating the evidence for exponential quantum advantage in ground-state quantum chemistry},\ }\href {https://www.nature.com/articles/s41467-023-37587-6} {\bibfield  {journal} {\bibinfo  {journal} {Nature communications}\ }\textbf {\bibinfo {volume} {14}},\ \bibinfo {pages} {1952} (\bibinfo {year} {2023})}\BibitemShut {NoStop}%
\bibitem [{\citenamefont {Parrish}\ and\ \citenamefont {McMahon}(2019)}]{parrish2019quantum}%
  \BibitemOpen
  \bibfield  {author} {\bibinfo {author} {\bibfnamefont {R.~M.}\ \bibnamefont {Parrish}}\ and\ \bibinfo {author} {\bibfnamefont {P.~L.}\ \bibnamefont {McMahon}},\ }\bibfield  {title} {\bibinfo {title} {Quantum filter diagonalization: Quantum eigendecomposition without full quantum phase estimation},\ }\href {https://arxiv.org/abs/1909.08925} {\bibfield  {journal} {\bibinfo  {journal} {arXiv:1909.08925}\ } (\bibinfo {year} {2019})}\BibitemShut {NoStop}%
\bibitem [{\citenamefont {Motta}\ \emph {et~al.}(2020)\citenamefont {Motta}, \citenamefont {Sun}, \citenamefont {Tan}, \citenamefont {O’Rourke}, \citenamefont {Ye}, \citenamefont {Minnich}, \citenamefont {Brandao},\ and\ \citenamefont {Chan}}]{motta2020determining}%
  \BibitemOpen
  \bibfield  {author} {\bibinfo {author} {\bibfnamefont {M.}~\bibnamefont {Motta}}, \bibinfo {author} {\bibfnamefont {C.}~\bibnamefont {Sun}}, \bibinfo {author} {\bibfnamefont {A.~T.}\ \bibnamefont {Tan}}, \bibinfo {author} {\bibfnamefont {M.~J.}\ \bibnamefont {O’Rourke}}, \bibinfo {author} {\bibfnamefont {E.}~\bibnamefont {Ye}}, \bibinfo {author} {\bibfnamefont {A.~J.}\ \bibnamefont {Minnich}}, \bibinfo {author} {\bibfnamefont {F.~G.}\ \bibnamefont {Brandao}},\ and\ \bibinfo {author} {\bibfnamefont {G.~K.-L.}\ \bibnamefont {Chan}},\ }\bibfield  {title} {\bibinfo {title} {Determining eigenstates and thermal states on a quantum computer using quantum imaginary time evolution},\ }\href {https://www.nature.com/articles/s41567-019-0704-4} {\bibfield  {journal} {\bibinfo  {journal} {Nature Physics}\ }\textbf {\bibinfo {volume} {16}},\ \bibinfo {pages} {205} (\bibinfo {year} {2020})}\BibitemShut {NoStop}%
\bibitem [{\citenamefont {Stair}\ \emph {et~al.}(2020)\citenamefont {Stair}, \citenamefont {Huang},\ and\ \citenamefont {Evangelista}}]{stair2020multireference}%
  \BibitemOpen
  \bibfield  {author} {\bibinfo {author} {\bibfnamefont {N.~H.}\ \bibnamefont {Stair}}, \bibinfo {author} {\bibfnamefont {R.}~\bibnamefont {Huang}},\ and\ \bibinfo {author} {\bibfnamefont {F.~A.}\ \bibnamefont {Evangelista}},\ }\bibfield  {title} {\bibinfo {title} {A multireference quantum krylov algorithm for strongly correlated electrons},\ }\href {https://pubs.acs.org/doi/abs/10.1021/acs.jctc.9b01125} {\bibfield  {journal} {\bibinfo  {journal} {Journal of chemical theory and computation}\ }\textbf {\bibinfo {volume} {16}},\ \bibinfo {pages} {2236} (\bibinfo {year} {2020})}\BibitemShut {NoStop}%
\bibitem [{\citenamefont {Cohn}\ \emph {et~al.}(2021)\citenamefont {Cohn}, \citenamefont {Motta},\ and\ \citenamefont {Parrish}}]{cohn2021filter}%
  \BibitemOpen
  \bibfield  {author} {\bibinfo {author} {\bibfnamefont {J.}~\bibnamefont {Cohn}}, \bibinfo {author} {\bibfnamefont {M.}~\bibnamefont {Motta}},\ and\ \bibinfo {author} {\bibfnamefont {R.~M.}\ \bibnamefont {Parrish}},\ }\bibfield  {title} {\bibinfo {title} {{Quantum Filter Diagonalization with Compressed Double-Factorized Hamiltonians}},\ }\href {https://doi.org/10.1103/PRXQuantum.2.040352} {\bibfield  {journal} {\bibinfo  {journal} {PRX Quantum}\ }\textbf {\bibinfo {volume} {2}},\ \bibinfo {pages} {040352} (\bibinfo {year} {2021})}\BibitemShut {NoStop}%
\bibitem [{\citenamefont {Seki}\ and\ \citenamefont {Yunoki}(2021)}]{seki2021power}%
  \BibitemOpen
  \bibfield  {author} {\bibinfo {author} {\bibfnamefont {K.}~\bibnamefont {Seki}}\ and\ \bibinfo {author} {\bibfnamefont {S.}~\bibnamefont {Yunoki}},\ }\bibfield  {title} {\bibinfo {title} {{Quantum Power Method by a Superposition of Time-Evolved States}},\ }\href {https://doi.org/10.1103/PRXQuantum.2.010333} {\bibfield  {journal} {\bibinfo  {journal} {PRX Quantum}\ }\textbf {\bibinfo {volume} {2}},\ \bibinfo {pages} {010333} (\bibinfo {year} {2021})}\BibitemShut {NoStop}%
\bibitem [{\citenamefont {Klymko}\ \emph {et~al.}(2022)\citenamefont {Klymko}, \citenamefont {Mejuto-Zaera}, \citenamefont {Cotton}, \citenamefont {Wudarski}, \citenamefont {Urbanek}, \citenamefont {Hait}, \citenamefont {Head-Gordon}, \citenamefont {Whaley}, \citenamefont {Moussa}, \citenamefont {Wiebe}, \citenamefont {de~Jong},\ and\ \citenamefont {Tubman}}]{klymko2022real}%
  \BibitemOpen
  \bibfield  {author} {\bibinfo {author} {\bibfnamefont {K.}~\bibnamefont {Klymko}}, \bibinfo {author} {\bibfnamefont {C.}~\bibnamefont {Mejuto-Zaera}}, \bibinfo {author} {\bibfnamefont {S.~J.}\ \bibnamefont {Cotton}}, \bibinfo {author} {\bibfnamefont {F.}~\bibnamefont {Wudarski}}, \bibinfo {author} {\bibfnamefont {M.}~\bibnamefont {Urbanek}}, \bibinfo {author} {\bibfnamefont {D.}~\bibnamefont {Hait}}, \bibinfo {author} {\bibfnamefont {M.}~\bibnamefont {Head-Gordon}}, \bibinfo {author} {\bibfnamefont {K.~B.}\ \bibnamefont {Whaley}}, \bibinfo {author} {\bibfnamefont {J.}~\bibnamefont {Moussa}}, \bibinfo {author} {\bibfnamefont {N.}~\bibnamefont {Wiebe}}, \bibinfo {author} {\bibfnamefont {W.~A.}\ \bibnamefont {de~Jong}},\ and\ \bibinfo {author} {\bibfnamefont {N.~M.}\ \bibnamefont {Tubman}},\ }\bibfield  {title} {\bibinfo {title} {{Real-Time Evolution for Ultracompact Hamiltonian Eigenstates on Quantum Hardware}},\ }\href {https://doi.org/10.1103/PRXQuantum.3.020323} {\bibfield  {journal} {\bibinfo  {journal}
  {PRX Quantum}\ }\textbf {\bibinfo {volume} {3}},\ \bibinfo {pages} {020323} (\bibinfo {year} {2022})}\BibitemShut {NoStop}%
\bibitem [{\citenamefont {Cortes}\ and\ \citenamefont {Gray}(2022)}]{cortes2022quantum}%
  \BibitemOpen
  \bibfield  {author} {\bibinfo {author} {\bibfnamefont {C.~L.}\ \bibnamefont {Cortes}}\ and\ \bibinfo {author} {\bibfnamefont {S.~K.}\ \bibnamefont {Gray}},\ }\bibfield  {title} {\bibinfo {title} {{Quantum Krylov subspace algorithms for ground- and excited-state energy estimation}},\ }\href {https://doi.org/10.1103/PhysRevA.105.022417} {\bibfield  {journal} {\bibinfo  {journal} {Phys. Rev. A}\ }\textbf {\bibinfo {volume} {105}},\ \bibinfo {pages} {022417} (\bibinfo {year} {2022})}\BibitemShut {NoStop}%
\bibitem [{\citenamefont {Epperly}\ \emph {et~al.}(2022)\citenamefont {Epperly}, \citenamefont {Lin},\ and\ \citenamefont {Nakatsukasa}}]{epperly2022theory}%
  \BibitemOpen
  \bibfield  {author} {\bibinfo {author} {\bibfnamefont {E.~N.}\ \bibnamefont {Epperly}}, \bibinfo {author} {\bibfnamefont {L.}~\bibnamefont {Lin}},\ and\ \bibinfo {author} {\bibfnamefont {Y.}~\bibnamefont {Nakatsukasa}},\ }\bibfield  {title} {\bibinfo {title} {{A theory of quantum subspace diagonalization}},\ }\href {https://epubs.siam.org/doi/10.1137/21M145954X} {\bibfield  {journal} {\bibinfo  {journal} {SIAM Journal on Matrix Analysis and Applications}\ }\textbf {\bibinfo {volume} {43}},\ \bibinfo {pages} {1263} (\bibinfo {year} {2022})}\BibitemShut {NoStop}%
\bibitem [{\citenamefont {Shen}\ \emph {et~al.}(2023)\citenamefont {Shen}, \citenamefont {Klymko}, \citenamefont {Sud}, \citenamefont {Williams-Young}, \citenamefont {de~Jong},\ and\ \citenamefont {Tubman}}]{shen2023real}%
  \BibitemOpen
  \bibfield  {author} {\bibinfo {author} {\bibfnamefont {Y.}~\bibnamefont {Shen}}, \bibinfo {author} {\bibfnamefont {K.}~\bibnamefont {Klymko}}, \bibinfo {author} {\bibfnamefont {J.}~\bibnamefont {Sud}}, \bibinfo {author} {\bibfnamefont {D.~B.}\ \bibnamefont {Williams-Young}}, \bibinfo {author} {\bibfnamefont {W.~A.}\ \bibnamefont {de~Jong}},\ and\ \bibinfo {author} {\bibfnamefont {N.~M.}\ \bibnamefont {Tubman}},\ }\bibfield  {title} {\bibinfo {title} {{Real-time Krylov theory for quantum computing algorithms}},\ }\href {https://quantum-journal.org/papers/q-2023-07-25-1066/} {\bibfield  {journal} {\bibinfo  {journal} {Quantum}\ }\textbf {\bibinfo {volume} {7}},\ \bibinfo {pages} {1066} (\bibinfo {year} {2023})}\BibitemShut {NoStop}%
\bibitem [{\citenamefont {Stair}\ \emph {et~al.}(2023)\citenamefont {Stair}, \citenamefont {Cortes}, \citenamefont {Parrish}, \citenamefont {Cohn},\ and\ \citenamefont {Motta}}]{stair2023stochastic}%
  \BibitemOpen
  \bibfield  {author} {\bibinfo {author} {\bibfnamefont {N.~H.}\ \bibnamefont {Stair}}, \bibinfo {author} {\bibfnamefont {C.~L.}\ \bibnamefont {Cortes}}, \bibinfo {author} {\bibfnamefont {R.~M.}\ \bibnamefont {Parrish}}, \bibinfo {author} {\bibfnamefont {J.}~\bibnamefont {Cohn}},\ and\ \bibinfo {author} {\bibfnamefont {M.}~\bibnamefont {Motta}},\ }\bibfield  {title} {\bibinfo {title} {{Stochastic quantum Krylov protocol with double-factorized Hamiltonians}},\ }\href {https://doi.org/10.1103/PhysRevA.107.032414} {\bibfield  {journal} {\bibinfo  {journal} {Phys. Rev. A}\ }\textbf {\bibinfo {volume} {107}},\ \bibinfo {pages} {032414} (\bibinfo {year} {2023})}\BibitemShut {NoStop}%
\bibitem [{\citenamefont {Tkachenko}\ \emph {et~al.}(2024)\citenamefont {Tkachenko}, \citenamefont {Cincio}, \citenamefont {Boldyrev}, \citenamefont {Tretiak}, \citenamefont {Dub},\ and\ \citenamefont {Zhang}}]{tkachenko2024quantum}%
  \BibitemOpen
  \bibfield  {author} {\bibinfo {author} {\bibfnamefont {N.~V.}\ \bibnamefont {Tkachenko}}, \bibinfo {author} {\bibfnamefont {L.}~\bibnamefont {Cincio}}, \bibinfo {author} {\bibfnamefont {A.~I.}\ \bibnamefont {Boldyrev}}, \bibinfo {author} {\bibfnamefont {S.}~\bibnamefont {Tretiak}}, \bibinfo {author} {\bibfnamefont {P.~A.}\ \bibnamefont {Dub}},\ and\ \bibinfo {author} {\bibfnamefont {Y.}~\bibnamefont {Zhang}},\ }\bibfield  {title} {\bibinfo {title} {{Quantum Davidson algorithm for excited states}},\ }\href {https://iopscience.iop.org/article/10.1088/2058-9565/ad3a97} {\bibfield  {journal} {\bibinfo  {journal} {Quantum Science and Technology}\ }\textbf {\bibinfo {volume} {9}},\ \bibinfo {pages} {035012} (\bibinfo {year} {2024})}\BibitemShut {NoStop}%
\bibitem [{\citenamefont {Zhang}\ \emph {et~al.}(2024)\citenamefont {Zhang}, \citenamefont {Wang}, \citenamefont {Xu},\ and\ \citenamefont {Li}}]{zhang2024measurement}%
  \BibitemOpen
  \bibfield  {author} {\bibinfo {author} {\bibfnamefont {Z.}~\bibnamefont {Zhang}}, \bibinfo {author} {\bibfnamefont {A.}~\bibnamefont {Wang}}, \bibinfo {author} {\bibfnamefont {X.}~\bibnamefont {Xu}},\ and\ \bibinfo {author} {\bibfnamefont {Y.}~\bibnamefont {Li}},\ }\bibfield  {title} {\bibinfo {title} {{Measurement-efficient quantum krylov subspace diagonalisation}},\ }\href {https://quantum-journal.org/papers/q-2024-08-13-1438/} {\bibfield  {journal} {\bibinfo  {journal} {Quantum}\ }\textbf {\bibinfo {volume} {8}},\ \bibinfo {pages} {1438} (\bibinfo {year} {2024})}\BibitemShut {NoStop}%
\bibitem [{\citenamefont {Kirby}(2024)}]{kirby2024analysis}%
  \BibitemOpen
  \bibfield  {author} {\bibinfo {author} {\bibfnamefont {W.}~\bibnamefont {Kirby}},\ }\bibfield  {title} {\bibinfo {title} {{Analysis of quantum Krylov algorithms with errors}},\ }\href {https://quantum-journal.org/papers/q-2024-08-29-1457/} {\bibfield  {journal} {\bibinfo  {journal} {Quantum}\ }\textbf {\bibinfo {volume} {8}},\ \bibinfo {pages} {1457} (\bibinfo {year} {2024})}\BibitemShut {NoStop}%
\bibitem [{\citenamefont {Lee}\ \emph {et~al.}(2024)\citenamefont {Lee}, \citenamefont {Lee},\ and\ \citenamefont {Huh}}]{lee2024sampling}%
  \BibitemOpen
  \bibfield  {author} {\bibinfo {author} {\bibfnamefont {G.}~\bibnamefont {Lee}}, \bibinfo {author} {\bibfnamefont {D.}~\bibnamefont {Lee}},\ and\ \bibinfo {author} {\bibfnamefont {J.}~\bibnamefont {Huh}},\ }\bibfield  {title} {\bibinfo {title} {{Sampling error analysis in quantum Krylov subspace diagonalization}},\ }\href {https://quantum-journal.org/papers/q-2024-09-19-1477/} {\bibfield  {journal} {\bibinfo  {journal} {Quantum}\ }\textbf {\bibinfo {volume} {8}},\ \bibinfo {pages} {1477} (\bibinfo {year} {2024})}\BibitemShut {NoStop}%
\bibitem [{\citenamefont {D'Cunha}\ \emph {et~al.}(2024)\citenamefont {D'Cunha}, \citenamefont {Cortes}, \citenamefont {Gagliardi},\ and\ \citenamefont {Gray}}]{dcunha2024fragment}%
  \BibitemOpen
  \bibfield  {author} {\bibinfo {author} {\bibfnamefont {R.}~\bibnamefont {D'Cunha}}, \bibinfo {author} {\bibfnamefont {C.~L.}\ \bibnamefont {Cortes}}, \bibinfo {author} {\bibfnamefont {L.}~\bibnamefont {Gagliardi}},\ and\ \bibinfo {author} {\bibfnamefont {S.~K.}\ \bibnamefont {Gray}},\ }\bibfield  {title} {\bibinfo {title} {Fragment-based initialization for quantum subspace methods},\ }\href {https://doi.org/10.1103/PhysRevA.110.042613} {\bibfield  {journal} {\bibinfo  {journal} {Phys. Rev. A}\ }\textbf {\bibinfo {volume} {110}},\ \bibinfo {pages} {042613} (\bibinfo {year} {2024})}\BibitemShut {NoStop}%
\bibitem [{\citenamefont {Yoshioka}\ \emph {et~al.}(2025)\citenamefont {Yoshioka}, \citenamefont {Amico}, \citenamefont {Kirby}, \citenamefont {Jurcevic}, \citenamefont {Dutt}, \citenamefont {Fuller}, \citenamefont {Garion}, \citenamefont {Haas}, \citenamefont {Hamamura}, \citenamefont {Ivrii} \emph {et~al.}}]{yoshioka2024diagonalization}%
  \BibitemOpen
  \bibfield  {author} {\bibinfo {author} {\bibfnamefont {N.}~\bibnamefont {Yoshioka}}, \bibinfo {author} {\bibfnamefont {M.}~\bibnamefont {Amico}}, \bibinfo {author} {\bibfnamefont {W.}~\bibnamefont {Kirby}}, \bibinfo {author} {\bibfnamefont {P.}~\bibnamefont {Jurcevic}}, \bibinfo {author} {\bibfnamefont {A.}~\bibnamefont {Dutt}}, \bibinfo {author} {\bibfnamefont {B.}~\bibnamefont {Fuller}}, \bibinfo {author} {\bibfnamefont {S.}~\bibnamefont {Garion}}, \bibinfo {author} {\bibfnamefont {H.}~\bibnamefont {Haas}}, \bibinfo {author} {\bibfnamefont {I.}~\bibnamefont {Hamamura}}, \bibinfo {author} {\bibfnamefont {A.}~\bibnamefont {Ivrii}}, \emph {et~al.},\ }\bibfield  {title} {\bibinfo {title} {{Krylov diagonalization of large many-body Hamiltonians on a quantum processor}},\ }\href {https://www.nature.com/articles/s41467-025-59716-z} {\bibfield  {journal} {\bibinfo  {journal} {Nat. Commun.}\ }\textbf {\bibinfo {volume} {16}},\ \bibinfo {pages} {5014} (\bibinfo {year} {2025})}\BibitemShut {NoStop}%
\bibitem [{\citenamefont {Lee}\ \emph {et~al.}(2025)\citenamefont {Lee}, \citenamefont {Choi}, \citenamefont {Huh},\ and\ \citenamefont {Izmaylov}}]{lee2025efficient}%
  \BibitemOpen
  \bibfield  {author} {\bibinfo {author} {\bibfnamefont {G.}~\bibnamefont {Lee}}, \bibinfo {author} {\bibfnamefont {S.}~\bibnamefont {Choi}}, \bibinfo {author} {\bibfnamefont {J.}~\bibnamefont {Huh}},\ and\ \bibinfo {author} {\bibfnamefont {A.~F.}\ \bibnamefont {Izmaylov}},\ }\bibfield  {title} {\bibinfo {title} {{Efficient strategies for reducing sampling error in quantum Krylov subspace diagonalization}},\ }\href {https://pubs.rsc.org/en/content/articlelanding/2025/dd/d4dd00321g} {\bibfield  {journal} {\bibinfo  {journal} {Digital Discovery}\ }\textbf {\bibinfo {volume} {4}},\ \bibinfo {pages} {954} (\bibinfo {year} {2025})}\BibitemShut {NoStop}%
\bibitem [{\citenamefont {Byrne}\ \emph {et~al.}(2024)\citenamefont {Byrne}, \citenamefont {Kirby}, \citenamefont {Soodhalter},\ and\ \citenamefont {Zhuk}}]{byrne2024quantum}%
  \BibitemOpen
  \bibfield  {author} {\bibinfo {author} {\bibfnamefont {A.}~\bibnamefont {Byrne}}, \bibinfo {author} {\bibfnamefont {W.}~\bibnamefont {Kirby}}, \bibinfo {author} {\bibfnamefont {K.~M.}\ \bibnamefont {Soodhalter}},\ and\ \bibinfo {author} {\bibfnamefont {S.}~\bibnamefont {Zhuk}},\ }\bibfield  {title} {\bibinfo {title} {{A Quantum-Centric Super-Krylov Diagonalization Method}},\ }\href {https://arxiv.org/abs/2412.17289} {\bibfield  {journal} {\bibinfo  {journal} {arXiv:2412.17289}\ } (\bibinfo {year} {2024})}\BibitemShut {NoStop}%
\bibitem [{\citenamefont {Oumarou}\ \emph {et~al.}(2025)\citenamefont {Oumarou}, \citenamefont {Ollitrault}, \citenamefont {Cortes}, \citenamefont {Scheurer}, \citenamefont {Parrish},\ and\ \citenamefont {Gogolin}}]{oumarou2025molecular}%
  \BibitemOpen
  \bibfield  {author} {\bibinfo {author} {\bibfnamefont {O.}~\bibnamefont {Oumarou}}, \bibinfo {author} {\bibfnamefont {P.~J.}\ \bibnamefont {Ollitrault}}, \bibinfo {author} {\bibfnamefont {C.~L.}\ \bibnamefont {Cortes}}, \bibinfo {author} {\bibfnamefont {M.}~\bibnamefont {Scheurer}}, \bibinfo {author} {\bibfnamefont {R.~M.}\ \bibnamefont {Parrish}},\ and\ \bibinfo {author} {\bibfnamefont {C.}~\bibnamefont {Gogolin}},\ }\bibfield  {title} {\bibinfo {title} {{Molecular properties from quantum krylov subspace diagonalization}},\ }\href {https://arxiv.org/abs/2501.05286} {\bibfield  {journal} {\bibinfo  {journal} {Journal of Chemical Theory and Computation}\ }\textbf {\bibinfo {volume} {21}},\ \bibinfo {pages} {4543} (\bibinfo {year} {2025})}\BibitemShut {NoStop}%
\bibitem [{\citenamefont {Szasz}\ \emph {et~al.}(2025)\citenamefont {Szasz}, \citenamefont {Younis},\ and\ \citenamefont {de~Jong}}]{szasz2025ground}%
  \BibitemOpen
  \bibfield  {author} {\bibinfo {author} {\bibfnamefont {A.}~\bibnamefont {Szasz}}, \bibinfo {author} {\bibfnamefont {E.}~\bibnamefont {Younis}},\ and\ \bibinfo {author} {\bibfnamefont {W.~A.}\ \bibnamefont {de~Jong}},\ }\bibfield  {title} {\bibinfo {title} {Ground state energy and magnetization curve of a frustrated magnetic system from real-time evolution on a digital quantum processor},\ }\href {https://quantum-journal.org/papers/q-2025-04-09-1704/} {\bibfield  {journal} {\bibinfo  {journal} {Quantum}\ }\textbf {\bibinfo {volume} {9}},\ \bibinfo {pages} {1704} (\bibinfo {year} {2025})}\BibitemShut {NoStop}%
\bibitem [{\citenamefont {O'Leary}\ \emph {et~al.}(2025)\citenamefont {O'Leary}, \citenamefont {Anderson}, \citenamefont {Jaksch},\ and\ \citenamefont {Kiffner}}]{oleary2025partitioned}%
  \BibitemOpen
  \bibfield  {author} {\bibinfo {author} {\bibfnamefont {T.}~\bibnamefont {O'Leary}}, \bibinfo {author} {\bibfnamefont {L.~W.}\ \bibnamefont {Anderson}}, \bibinfo {author} {\bibfnamefont {D.}~\bibnamefont {Jaksch}},\ and\ \bibinfo {author} {\bibfnamefont {M.}~\bibnamefont {Kiffner}},\ }\bibfield  {title} {\bibinfo {title} {Partitioned quantum subspace expansion},\ }\href {https://quantum-journal.org/papers/q-2025-05-05-1726/} {\bibfield  {journal} {\bibinfo  {journal} {Quantum}\ }\textbf {\bibinfo {volume} {9}},\ \bibinfo {pages} {1726} (\bibinfo {year} {2025})}\BibitemShut {NoStop}%
\bibitem [{\citenamefont {Suzuki}\ \emph {et~al.}(2022)\citenamefont {Suzuki}, \citenamefont {Endo}, \citenamefont {Fujii},\ and\ \citenamefont {Tokunaga}}]{suzuki2022preftqc}%
  \BibitemOpen
  \bibfield  {author} {\bibinfo {author} {\bibfnamefont {Y.}~\bibnamefont {Suzuki}}, \bibinfo {author} {\bibfnamefont {S.}~\bibnamefont {Endo}}, \bibinfo {author} {\bibfnamefont {K.}~\bibnamefont {Fujii}},\ and\ \bibinfo {author} {\bibfnamefont {Y.}~\bibnamefont {Tokunaga}},\ }\bibfield  {title} {\bibinfo {title} {{Quantum Error Mitigation as a Universal Error Reduction Technique: Applications from the NISQ to the Fault-Tolerant Quantum Computing Eras}},\ }\href {https://doi.org/10.1103/PRXQuantum.3.010345} {\bibfield  {journal} {\bibinfo  {journal} {PRX Quantum}\ }\textbf {\bibinfo {volume} {3}},\ \bibinfo {pages} {010345} (\bibinfo {year} {2022})}\BibitemShut {NoStop}%
\bibitem [{\citenamefont {Katabarwa}\ \emph {et~al.}(2024)\citenamefont {Katabarwa}, \citenamefont {Gratsea}, \citenamefont {Caesura},\ and\ \citenamefont {Johnson}}]{katabarwa2024early}%
  \BibitemOpen
  \bibfield  {author} {\bibinfo {author} {\bibfnamefont {A.}~\bibnamefont {Katabarwa}}, \bibinfo {author} {\bibfnamefont {K.}~\bibnamefont {Gratsea}}, \bibinfo {author} {\bibfnamefont {A.}~\bibnamefont {Caesura}},\ and\ \bibinfo {author} {\bibfnamefont {P.~D.}\ \bibnamefont {Johnson}},\ }\bibfield  {title} {\bibinfo {title} {{Early Fault-Tolerant Quantum Computing}},\ }\href {https://doi.org/10.1103/PRXQuantum.5.020101} {\bibfield  {journal} {\bibinfo  {journal} {PRX Quantum}\ }\textbf {\bibinfo {volume} {5}},\ \bibinfo {pages} {020101} (\bibinfo {year} {2024})}\BibitemShut {NoStop}%
\bibitem [{\citenamefont {Akahoshi}\ \emph {et~al.}(2024)\citenamefont {Akahoshi}, \citenamefont {Maruyama}, \citenamefont {Oshima}, \citenamefont {Sato},\ and\ \citenamefont {Fujii}}]{akahoshi2024star}%
  \BibitemOpen
  \bibfield  {author} {\bibinfo {author} {\bibfnamefont {Y.}~\bibnamefont {Akahoshi}}, \bibinfo {author} {\bibfnamefont {K.}~\bibnamefont {Maruyama}}, \bibinfo {author} {\bibfnamefont {H.}~\bibnamefont {Oshima}}, \bibinfo {author} {\bibfnamefont {S.}~\bibnamefont {Sato}},\ and\ \bibinfo {author} {\bibfnamefont {K.}~\bibnamefont {Fujii}},\ }\bibfield  {title} {\bibinfo {title} {{Partially Fault-Tolerant Quantum Computing Architecture with Error-Corrected Clifford Gates and Space-Time Efficient Analog Rotations}},\ }\href {https://doi.org/10.1103/PRXQuantum.5.010337} {\bibfield  {journal} {\bibinfo  {journal} {PRX Quantum}\ }\textbf {\bibinfo {volume} {5}},\ \bibinfo {pages} {010337} (\bibinfo {year} {2024})}\BibitemShut {NoStop}%
\bibitem [{\citenamefont {Toshio}\ \emph {et~al.}(2025)\citenamefont {Toshio}, \citenamefont {Akahoshi}, \citenamefont {Fujisaki}, \citenamefont {Oshima}, \citenamefont {Sato},\ and\ \citenamefont {Fujii}}]{toshio2025starv2}%
  \BibitemOpen
  \bibfield  {author} {\bibinfo {author} {\bibfnamefont {R.}~\bibnamefont {Toshio}}, \bibinfo {author} {\bibfnamefont {Y.}~\bibnamefont {Akahoshi}}, \bibinfo {author} {\bibfnamefont {J.}~\bibnamefont {Fujisaki}}, \bibinfo {author} {\bibfnamefont {H.}~\bibnamefont {Oshima}}, \bibinfo {author} {\bibfnamefont {S.}~\bibnamefont {Sato}},\ and\ \bibinfo {author} {\bibfnamefont {K.}~\bibnamefont {Fujii}},\ }\bibfield  {title} {\bibinfo {title} {{Practical Quantum Advantage on Partially Fault-Tolerant Quantum Computer}},\ }\href {https://doi.org/10.1103/PhysRevX.15.021057} {\bibfield  {journal} {\bibinfo  {journal} {Phys. Rev. X}\ }\textbf {\bibinfo {volume} {15}},\ \bibinfo {pages} {021057} (\bibinfo {year} {2025})}\BibitemShut {NoStop}%
\bibitem [{\citenamefont {Akahoshi}\ \emph {et~al.}(2025)\citenamefont {Akahoshi}, \citenamefont {Toshio}, \citenamefont {Fujisaki}, \citenamefont {Oshima}, \citenamefont {Sato},\ and\ \citenamefont {Fujii}}]{akahoshi2024compilation}%
  \BibitemOpen
  \bibfield  {author} {\bibinfo {author} {\bibfnamefont {Y.}~\bibnamefont {Akahoshi}}, \bibinfo {author} {\bibfnamefont {R.}~\bibnamefont {Toshio}}, \bibinfo {author} {\bibfnamefont {J.}~\bibnamefont {Fujisaki}}, \bibinfo {author} {\bibfnamefont {H.}~\bibnamefont {Oshima}}, \bibinfo {author} {\bibfnamefont {S.}~\bibnamefont {Sato}},\ and\ \bibinfo {author} {\bibfnamefont {K.}~\bibnamefont {Fujii}},\ }\bibfield  {title} {\bibinfo {title} {{Compilation of Trotter-Based Time Evolution for Partially Fault-Tolerant Quantum Computing Architecture}},\ }\href {https://doi.org/10.1103/93zr-1ykb} {\bibfield  {journal} {\bibinfo  {journal} {PRX Quantum}\ }\textbf {\bibinfo {volume} {6}},\ \bibinfo {pages} {040319} (\bibinfo {year} {2025})}\BibitemShut {NoStop}%
\bibitem [{\citenamefont {Kuroiwa}\ and\ \citenamefont {Nakagawa}(2025)}]{kuroiwa2025averaging}%
  \BibitemOpen
  \bibfield  {author} {\bibinfo {author} {\bibfnamefont {K.}~\bibnamefont {Kuroiwa}}\ and\ \bibinfo {author} {\bibfnamefont {Y.~O.}\ \bibnamefont {Nakagawa}},\ }\bibfield  {title} {\bibinfo {title} {{Averaging gate approximation error and performance of Unitary Coupled Cluster ansatz in Pre-FTQC Era}},\ }\href {https://quantum-journal.org/papers/q-2025-07-21-1800/} {\bibfield  {journal} {\bibinfo  {journal} {Quantum}\ }\textbf {\bibinfo {volume} {9}},\ \bibinfo {pages} {1800} (\bibinfo {year} {2025})}\BibitemShut {NoStop}%
\bibitem [{\citenamefont {Toshio}\ \emph {et~al.}(2026)\citenamefont {Toshio}, \citenamefont {Kanasugi}, \citenamefont {Fujisaki}, \citenamefont {Oshima}, \citenamefont {Sato},\ and\ \citenamefont {Fujii}}]{Toshio2026-qm}%
  \BibitemOpen
  \bibfield  {author} {\bibinfo {author} {\bibfnamefont {R.}~\bibnamefont {Toshio}}, \bibinfo {author} {\bibfnamefont {S.}~\bibnamefont {Kanasugi}}, \bibinfo {author} {\bibfnamefont {J.}~\bibnamefont {Fujisaki}}, \bibinfo {author} {\bibfnamefont {H.}~\bibnamefont {Oshima}}, \bibinfo {author} {\bibfnamefont {S.}~\bibnamefont {Sato}},\ and\ \bibinfo {author} {\bibfnamefont {K.}~\bibnamefont {Fujii}},\ }\bibfield  {title} {\bibinfo {title} {{STAR}-magic mutation: Even more efficient analog rotation gates for early fault-tolerant quantum computer},\ }\href {http://arxiv.org/abs/2603.22891} {\bibfield  {journal} {\bibinfo  {journal} {arXiv:2603.22891}\ } (\bibinfo {year} {2026})}\BibitemShut {NoStop}%
\bibitem [{\citenamefont {Kanasugi}\ \emph {et~al.}(2026)\citenamefont {Kanasugi}, \citenamefont {Toshio}, \citenamefont {Maruyama},\ and\ \citenamefont {Oshima}}]{Kanasugi2026-eo}%
  \BibitemOpen
  \bibfield  {author} {\bibinfo {author} {\bibfnamefont {S.}~\bibnamefont {Kanasugi}}, \bibinfo {author} {\bibfnamefont {R.}~\bibnamefont {Toshio}}, \bibinfo {author} {\bibfnamefont {K.}~\bibnamefont {Maruyama}},\ and\ \bibinfo {author} {\bibfnamefont {H.}~\bibnamefont {Oshima}},\ }\bibfield  {title} {\bibinfo {title} {Enabling chemically accurate quantum phase estimation in the early fault-tolerant regime},\ }\href {http://arxiv.org/abs/2603.22778} {\bibfield  {journal} {\bibinfo  {journal} {arXiv:2603.22778}\ } (\bibinfo {year} {2026})}\BibitemShut {NoStop}%
\bibitem [{\citenamefont {Lanczos}(1950)}]{lanczos1950iteration}%
  \BibitemOpen
  \bibfield  {author} {\bibinfo {author} {\bibfnamefont {C.}~\bibnamefont {Lanczos}},\ }\bibfield  {title} {\bibinfo {title} {An iteration method for the solution of the eigenvalue problem of linear differential and integral operators},\ }\href@noop {} {\bibfield  {journal} {\bibinfo  {journal} {Journal of research of the National Bureau of Standards}\ }\textbf {\bibinfo {volume} {45}},\ \bibinfo {pages} {255} (\bibinfo {year} {1950})}\BibitemShut {NoStop}%
\bibitem [{\citenamefont {Gily{\'e}n}\ \emph {et~al.}(2019)\citenamefont {Gily{\'e}n}, \citenamefont {Arunachalam},\ and\ \citenamefont {Wiebe}}]{gilyen2019optimizing}%
  \BibitemOpen
  \bibfield  {author} {\bibinfo {author} {\bibfnamefont {A.}~\bibnamefont {Gily{\'e}n}}, \bibinfo {author} {\bibfnamefont {S.}~\bibnamefont {Arunachalam}},\ and\ \bibinfo {author} {\bibfnamefont {N.}~\bibnamefont {Wiebe}},\ }\bibfield  {title} {\bibinfo {title} {Optimizing quantum optimization algorithms via faster quantum gradient computation},\ }in\ \href {https://epubs.siam.org/doi/10.1137/1.9781611975482.87} {\emph {\bibinfo {booktitle} {Proceedings of the Thirtieth Annual ACM-SIAM Symposium on Discrete Algorithms}}}\ (\bibinfo {organization} {SIAM},\ \bibinfo {year} {2019})\ pp.\ \bibinfo {pages} {1425--1444}\BibitemShut {NoStop}%
\bibitem [{\citenamefont {Bespalova}\ and\ \citenamefont {Kyriienko}(2021)}]{bespalova2021hamiltonian}%
  \BibitemOpen
  \bibfield  {author} {\bibinfo {author} {\bibfnamefont {T.~A.}\ \bibnamefont {Bespalova}}\ and\ \bibinfo {author} {\bibfnamefont {O.}~\bibnamefont {Kyriienko}},\ }\bibfield  {title} {\bibinfo {title} {{Hamiltonian Operator Approximation for Energy Measurement and Ground-State Preparation}},\ }\href {https://doi.org/10.1103/PRXQuantum.2.030318} {\bibfield  {journal} {\bibinfo  {journal} {PRX Quantum}\ }\textbf {\bibinfo {volume} {2}},\ \bibinfo {pages} {030318} (\bibinfo {year} {2021})}\BibitemShut {NoStop}%
\bibitem [{\citenamefont {Cortes}\ \emph {et~al.}(2024)\citenamefont {Cortes}, \citenamefont {Rocca}, \citenamefont {Gonthier}, \citenamefont {Ollitrault}, \citenamefont {Parrish}, \citenamefont {Anselmetti}, \citenamefont {Degroote}, \citenamefont {Moll}, \citenamefont {Santagati},\ and\ \citenamefont {Streif}}]{cortes2024spectral}%
  \BibitemOpen
  \bibfield  {author} {\bibinfo {author} {\bibfnamefont {C.~L.}\ \bibnamefont {Cortes}}, \bibinfo {author} {\bibfnamefont {D.}~\bibnamefont {Rocca}}, \bibinfo {author} {\bibfnamefont {J.~F.}\ \bibnamefont {Gonthier}}, \bibinfo {author} {\bibfnamefont {P.~J.}\ \bibnamefont {Ollitrault}}, \bibinfo {author} {\bibfnamefont {R.~M.}\ \bibnamefont {Parrish}}, \bibinfo {author} {\bibfnamefont {G.-L.~R.}\ \bibnamefont {Anselmetti}}, \bibinfo {author} {\bibfnamefont {M.}~\bibnamefont {Degroote}}, \bibinfo {author} {\bibfnamefont {N.}~\bibnamefont {Moll}}, \bibinfo {author} {\bibfnamefont {R.}~\bibnamefont {Santagati}},\ and\ \bibinfo {author} {\bibfnamefont {M.}~\bibnamefont {Streif}},\ }\bibfield  {title} {\bibinfo {title} {Assessing the query complexity limits of quantum phase estimation using symmetry-aware spectral bounds},\ }\href {https://doi.org/10.1103/PhysRevA.110.022420} {\bibfield  {journal} {\bibinfo  {journal} {Phys. Rev. A}\ }\textbf {\bibinfo {volume} {110}},\ \bibinfo {pages} {022420} (\bibinfo {year}
  {2024})}\BibitemShut {NoStop}%
\bibitem [{\citenamefont {Jensen}(2017)}]{jensen2017introduction}%
  \BibitemOpen
  \bibfield  {author} {\bibinfo {author} {\bibfnamefont {F.}~\bibnamefont {Jensen}},\ }\href@noop {} {\emph {\bibinfo {title} {{Introduction to Computational Chemistry}}}}\ (\bibinfo  {publisher} {John wiley \& sons},\ \bibinfo {year} {2017})\BibitemShut {NoStop}%
\bibitem [{\citenamefont {Koridon}\ \emph {et~al.}(2021)\citenamefont {Koridon}, \citenamefont {Yalouz}, \citenamefont {Senjean}, \citenamefont {Buda}, \citenamefont {O'Brien},\ and\ \citenamefont {Visscher}}]{koridon2021orbital}%
  \BibitemOpen
  \bibfield  {author} {\bibinfo {author} {\bibfnamefont {E.}~\bibnamefont {Koridon}}, \bibinfo {author} {\bibfnamefont {S.}~\bibnamefont {Yalouz}}, \bibinfo {author} {\bibfnamefont {B.}~\bibnamefont {Senjean}}, \bibinfo {author} {\bibfnamefont {F.}~\bibnamefont {Buda}}, \bibinfo {author} {\bibfnamefont {T.~E.}\ \bibnamefont {O'Brien}},\ and\ \bibinfo {author} {\bibfnamefont {L.}~\bibnamefont {Visscher}},\ }\bibfield  {title} {\bibinfo {title} {Orbital transformations to reduce the 1-norm of the electronic structure hamiltonian for quantum computing applications},\ }\href {https://doi.org/10.1103/PhysRevResearch.3.033127} {\bibfield  {journal} {\bibinfo  {journal} {Phys. Rev. Res.}\ }\textbf {\bibinfo {volume} {3}},\ \bibinfo {pages} {033127} (\bibinfo {year} {2021})}\BibitemShut {NoStop}%
\bibitem [{\citenamefont {Loaiza}\ and\ \citenamefont {Izmaylov}(2023)}]{loaiza2023bliss}%
  \BibitemOpen
  \bibfield  {author} {\bibinfo {author} {\bibfnamefont {I.}~\bibnamefont {Loaiza}}\ and\ \bibinfo {author} {\bibfnamefont {A.~F.}\ \bibnamefont {Izmaylov}},\ }\bibfield  {title} {\bibinfo {title} {{Block-invariant symmetry shift: Preprocessing technique for second-quantized hamiltonians to improve their decompositions to linear combination of unitaries}},\ }\href {https://pubs.acs.org/doi/abs/10.1021/acs.jctc.3c00912} {\bibfield  {journal} {\bibinfo  {journal} {Journal of Chemical Theory and Computation}\ }\textbf {\bibinfo {volume} {19}},\ \bibinfo {pages} {8201} (\bibinfo {year} {2023})}\BibitemShut {NoStop}%
\bibitem [{\citenamefont {Patel}\ \emph {et~al.}(2024)\citenamefont {Patel}, \citenamefont {Brahmachari}, \citenamefont {Cantin}, \citenamefont {Wang},\ and\ \citenamefont {Izmaylov}}]{patel2024blisslp}%
  \BibitemOpen
  \bibfield  {author} {\bibinfo {author} {\bibfnamefont {S.}~\bibnamefont {Patel}}, \bibinfo {author} {\bibfnamefont {A.~S.}\ \bibnamefont {Brahmachari}}, \bibinfo {author} {\bibfnamefont {J.~T.}\ \bibnamefont {Cantin}}, \bibinfo {author} {\bibfnamefont {L.}~\bibnamefont {Wang}},\ and\ \bibinfo {author} {\bibfnamefont {A.~F.}\ \bibnamefont {Izmaylov}},\ }\bibfield  {title} {\bibinfo {title} {{Guaranteed Global Minimum of Electronic Hamiltonian 1-Norm via Linear Programming in the Block Invariant Symmetry Shift (BLISS) Method}},\ }\href {https://arxiv.org/abs/2409.18277} {\bibfield  {journal} {\bibinfo  {journal} {arXiv:2409.18277}\ } (\bibinfo {year} {2024})}\BibitemShut {NoStop}%
\bibitem [{\citenamefont {Loaiza}\ \emph {et~al.}(2023)\citenamefont {Loaiza}, \citenamefont {Khah}, \citenamefont {Wiebe},\ and\ \citenamefont {Izmaylov}}]{loaiza2023lcu}%
  \BibitemOpen
  \bibfield  {author} {\bibinfo {author} {\bibfnamefont {I.}~\bibnamefont {Loaiza}}, \bibinfo {author} {\bibfnamefont {A.~M.}\ \bibnamefont {Khah}}, \bibinfo {author} {\bibfnamefont {N.}~\bibnamefont {Wiebe}},\ and\ \bibinfo {author} {\bibfnamefont {A.~F.}\ \bibnamefont {Izmaylov}},\ }\bibfield  {title} {\bibinfo {title} {Reducing molecular electronic hamiltonian simulation cost for linear combination of unitaries approaches},\ }\href {https://iopscience.iop.org/article/10.1088/2058-9565/acd577} {\bibfield  {journal} {\bibinfo  {journal} {Quantum Science and Technology}\ }\textbf {\bibinfo {volume} {8}},\ \bibinfo {pages} {035019} (\bibinfo {year} {2023})}\BibitemShut {NoStop}%
\bibitem [{\citenamefont {Lin}\ and\ \citenamefont {Tong}(2022)}]{lin2022heisenberg}%
  \BibitemOpen
  \bibfield  {author} {\bibinfo {author} {\bibfnamefont {L.}~\bibnamefont {Lin}}\ and\ \bibinfo {author} {\bibfnamefont {Y.}~\bibnamefont {Tong}},\ }\bibfield  {title} {\bibinfo {title} {{Heisenberg-Limited Ground-State Energy Estimation for Early Fault-Tolerant Quantum Computers}},\ }\href {https://doi.org/10.1103/PRXQuantum.3.010318} {\bibfield  {journal} {\bibinfo  {journal} {PRX Quantum}\ }\textbf {\bibinfo {volume} {3}},\ \bibinfo {pages} {010318} (\bibinfo {year} {2022})}\BibitemShut {NoStop}%
\bibitem [{\citenamefont {Aulicino}\ \emph {et~al.}(2022)\citenamefont {Aulicino}, \citenamefont {Keen},\ and\ \citenamefont {Peng}}]{aulicino2022state}%
  \BibitemOpen
  \bibfield  {author} {\bibinfo {author} {\bibfnamefont {J.~C.}\ \bibnamefont {Aulicino}}, \bibinfo {author} {\bibfnamefont {T.}~\bibnamefont {Keen}},\ and\ \bibinfo {author} {\bibfnamefont {B.}~\bibnamefont {Peng}},\ }\bibfield  {title} {\bibinfo {title} {{State preparation and evolution in quantum computing: A perspective from Hamiltonian moments}},\ }\href {https://onlinelibrary.wiley.com/doi/10.1002/qua.26853} {\bibfield  {journal} {\bibinfo  {journal} {International Journal of Quantum Chemistry}\ }\textbf {\bibinfo {volume} {122}},\ \bibinfo {pages} {e26853} (\bibinfo {year} {2022})}\BibitemShut {NoStop}%
\bibitem [{\citenamefont {Vallury}\ \emph {et~al.}(2020)\citenamefont {Vallury}, \citenamefont {Jones}, \citenamefont {Hill},\ and\ \citenamefont {Hollenberg}}]{vallury2020quantum}%
  \BibitemOpen
  \bibfield  {author} {\bibinfo {author} {\bibfnamefont {H.~J.}\ \bibnamefont {Vallury}}, \bibinfo {author} {\bibfnamefont {M.~A.}\ \bibnamefont {Jones}}, \bibinfo {author} {\bibfnamefont {C.~D.}\ \bibnamefont {Hill}},\ and\ \bibinfo {author} {\bibfnamefont {L.~C.}\ \bibnamefont {Hollenberg}},\ }\bibfield  {title} {\bibinfo {title} {Quantum computed moments correction to variational estimates},\ }\href {https://quantum-journal.org/papers/q-2020-12-15-373/} {\bibfield  {journal} {\bibinfo  {journal} {Quantum}\ }\textbf {\bibinfo {volume} {4}},\ \bibinfo {pages} {373} (\bibinfo {year} {2020})}\BibitemShut {NoStop}%
\bibitem [{\citenamefont {Witte}\ and\ \citenamefont {Hollenberg}(1994)}]{witte1994plaquette}%
  \BibitemOpen
  \bibfield  {author} {\bibinfo {author} {\bibfnamefont {N.}~\bibnamefont {Witte}}\ and\ \bibinfo {author} {\bibfnamefont {L.~C.}\ \bibnamefont {Hollenberg}},\ }\bibfield  {title} {\bibinfo {title} {Plaquette expansion proof and interpretation},\ }\href {https://link.springer.com/article/10.1007/BF01313364} {\bibfield  {journal} {\bibinfo  {journal} {Zeitschrift f{\"u}r Physik B Condensed Matter}\ }\textbf {\bibinfo {volume} {95}},\ \bibinfo {pages} {531} (\bibinfo {year} {1994})}\BibitemShut {NoStop}%
\bibitem [{\citenamefont {Haxton}\ \emph {et~al.}(2005)\citenamefont {Haxton}, \citenamefont {Nollett},\ and\ \citenamefont {Zurek}}]{haxton2005piecewise}%
  \BibitemOpen
  \bibfield  {author} {\bibinfo {author} {\bibfnamefont {W.~C.}\ \bibnamefont {Haxton}}, \bibinfo {author} {\bibfnamefont {K.~M.}\ \bibnamefont {Nollett}},\ and\ \bibinfo {author} {\bibfnamefont {K.~M.}\ \bibnamefont {Zurek}},\ }\bibfield  {title} {\bibinfo {title} {Piecewise moments method: Generalized lanczos technique for nuclear response surfaces},\ }\href {https://doi.org/10.1103/PhysRevC.72.065501} {\bibfield  {journal} {\bibinfo  {journal} {Phys. Rev. C}\ }\textbf {\bibinfo {volume} {72}},\ \bibinfo {pages} {065501} (\bibinfo {year} {2005})}\BibitemShut {NoStop}%
\bibitem [{\citenamefont {Cai}\ \emph {et~al.}(2023)\citenamefont {Cai}, \citenamefont {Babbush}, \citenamefont {Benjamin}, \citenamefont {Endo}, \citenamefont {Huggins}, \citenamefont {Li}, \citenamefont {McClean},\ and\ \citenamefont {O'Brien}}]{cai2023mitigation}%
  \BibitemOpen
  \bibfield  {author} {\bibinfo {author} {\bibfnamefont {Z.}~\bibnamefont {Cai}}, \bibinfo {author} {\bibfnamefont {R.}~\bibnamefont {Babbush}}, \bibinfo {author} {\bibfnamefont {S.~C.}\ \bibnamefont {Benjamin}}, \bibinfo {author} {\bibfnamefont {S.}~\bibnamefont {Endo}}, \bibinfo {author} {\bibfnamefont {W.~J.}\ \bibnamefont {Huggins}}, \bibinfo {author} {\bibfnamefont {Y.}~\bibnamefont {Li}}, \bibinfo {author} {\bibfnamefont {J.~R.}\ \bibnamefont {McClean}},\ and\ \bibinfo {author} {\bibfnamefont {T.~E.}\ \bibnamefont {O'Brien}},\ }\bibfield  {title} {\bibinfo {title} {Quantum error mitigation},\ }\href {https://doi.org/10.1103/RevModPhys.95.045005} {\bibfield  {journal} {\bibinfo  {journal} {Rev. Mod. Phys.}\ }\textbf {\bibinfo {volume} {95}},\ \bibinfo {pages} {045005} (\bibinfo {year} {2023})}\BibitemShut {NoStop}%
\bibitem [{\citenamefont {Zhang}\ \emph {et~al.}(2022)\citenamefont {Zhang}, \citenamefont {Wang},\ and\ \citenamefont {Johnson}}]{zhang2022computing}%
  \BibitemOpen
  \bibfield  {author} {\bibinfo {author} {\bibfnamefont {R.}~\bibnamefont {Zhang}}, \bibinfo {author} {\bibfnamefont {G.}~\bibnamefont {Wang}},\ and\ \bibinfo {author} {\bibfnamefont {P.}~\bibnamefont {Johnson}},\ }\bibfield  {title} {\bibinfo {title} {Computing ground state properties with early fault-tolerant quantum computers},\ }\href {https://quantum-journal.org/papers/q-2022-07-11-761/} {\bibfield  {journal} {\bibinfo  {journal} {Quantum}\ }\textbf {\bibinfo {volume} {6}},\ \bibinfo {pages} {761} (\bibinfo {year} {2022})}\BibitemShut {NoStop}%
\bibitem [{\citenamefont {Cortes}\ \emph {et~al.}(2022)\citenamefont {Cortes}, \citenamefont {DePrince},\ and\ \citenamefont {Gray}}]{cortes2022fast}%
  \BibitemOpen
  \bibfield  {author} {\bibinfo {author} {\bibfnamefont {C.~L.}\ \bibnamefont {Cortes}}, \bibinfo {author} {\bibfnamefont {A.~E.}\ \bibnamefont {DePrince}},\ and\ \bibinfo {author} {\bibfnamefont {S.~K.}\ \bibnamefont {Gray}},\ }\bibfield  {title} {\bibinfo {title} {{Fast-forwarding quantum simulation with real-time quantum Krylov subspace algorithms}},\ }\href {https://doi.org/10.1103/PhysRevA.106.042409} {\bibfield  {journal} {\bibinfo  {journal} {Phys. Rev. A}\ }\textbf {\bibinfo {volume} {106}},\ \bibinfo {pages} {042409} (\bibinfo {year} {2022})}\BibitemShut {NoStop}%
\bibitem [{\citenamefont {Jamet}\ \emph {et~al.}(2021)\citenamefont {Jamet}, \citenamefont {Agarwal}, \citenamefont {Lupo}, \citenamefont {Browne}, \citenamefont {Weber},\ and\ \citenamefont {Rungger}}]{jamet2021krylov}%
  \BibitemOpen
  \bibfield  {author} {\bibinfo {author} {\bibfnamefont {F.}~\bibnamefont {Jamet}}, \bibinfo {author} {\bibfnamefont {A.}~\bibnamefont {Agarwal}}, \bibinfo {author} {\bibfnamefont {C.}~\bibnamefont {Lupo}}, \bibinfo {author} {\bibfnamefont {D.~E.}\ \bibnamefont {Browne}}, \bibinfo {author} {\bibfnamefont {C.}~\bibnamefont {Weber}},\ and\ \bibinfo {author} {\bibfnamefont {I.}~\bibnamefont {Rungger}},\ }\bibfield  {title} {\bibinfo {title} {Krylov variational quantum algorithm for first principles materials simulations},\ }\href {https://arxiv.org/abs/2105.13298} {\bibfield  {journal} {\bibinfo  {journal} {arXiv:2105.13298}\ } (\bibinfo {year} {2021})}\BibitemShut {NoStop}%
\bibitem [{\citenamefont {Jamet}\ \emph {et~al.}(2022)\citenamefont {Jamet}, \citenamefont {Agarwal},\ and\ \citenamefont {Rungger}}]{jamet2022quantum}%
  \BibitemOpen
  \bibfield  {author} {\bibinfo {author} {\bibfnamefont {F.}~\bibnamefont {Jamet}}, \bibinfo {author} {\bibfnamefont {A.}~\bibnamefont {Agarwal}},\ and\ \bibinfo {author} {\bibfnamefont {I.}~\bibnamefont {Rungger}},\ }\bibfield  {title} {\bibinfo {title} {{Quantum subspace expansion algorithm for Green's functions}},\ }\href {https://arxiv.org/abs/2205.00094} {\bibfield  {journal} {\bibinfo  {journal} {arXiv:2205.00094}\ } (\bibinfo {year} {2022})}\BibitemShut {NoStop}%
\bibitem [{\citenamefont {Jamet}\ \emph {et~al.}(2025)\citenamefont {Jamet}, \citenamefont {Lindoy}, \citenamefont {Rath}, \citenamefont {Lenihan}, \citenamefont {Agarwal}, \citenamefont {Fontana}, \citenamefont {Simkovic}, \citenamefont {Martin},\ and\ \citenamefont {Rungger}}]{jamet2025anderson}%
  \BibitemOpen
  \bibfield  {author} {\bibinfo {author} {\bibfnamefont {F.}~\bibnamefont {Jamet}}, \bibinfo {author} {\bibfnamefont {L.~P.}\ \bibnamefont {Lindoy}}, \bibinfo {author} {\bibfnamefont {Y.}~\bibnamefont {Rath}}, \bibinfo {author} {\bibfnamefont {C.}~\bibnamefont {Lenihan}}, \bibinfo {author} {\bibfnamefont {A.}~\bibnamefont {Agarwal}}, \bibinfo {author} {\bibfnamefont {E.}~\bibnamefont {Fontana}}, \bibinfo {author} {\bibfnamefont {F.}~\bibnamefont {Simkovic}}, \bibinfo {author} {\bibfnamefont {B.~A.}\ \bibnamefont {Martin}},\ and\ \bibinfo {author} {\bibfnamefont {I.}~\bibnamefont {Rungger}},\ }\bibfield  {title} {\bibinfo {title} {Anderson impurity solver integrating tensor network methods with quantum computing},\ }\href {https://pubs.aip.org/aip/apq/article/2/1/016121/3336642/Anderson-impurity-solver-integrating-tensor} {\bibfield  {journal} {\bibinfo  {journal} {APL Quantum}\ }\textbf {\bibinfo {volume} {2}} (\bibinfo {year} {2025})}\BibitemShut {NoStop}%
\bibitem [{\citenamefont {Umeano}\ \emph {et~al.}(2025)\citenamefont {Umeano}, \citenamefont {Jamet}, \citenamefont {Lindoy}, \citenamefont {Rungger},\ and\ \citenamefont {Kyriienko}}]{umeano2025quantum}%
  \BibitemOpen
  \bibfield  {author} {\bibinfo {author} {\bibfnamefont {C.}~\bibnamefont {Umeano}}, \bibinfo {author} {\bibfnamefont {F.~m.~c.}\ \bibnamefont {Jamet}}, \bibinfo {author} {\bibfnamefont {L.~P.}\ \bibnamefont {Lindoy}}, \bibinfo {author} {\bibfnamefont {I.}~\bibnamefont {Rungger}},\ and\ \bibinfo {author} {\bibfnamefont {O.}~\bibnamefont {Kyriienko}},\ }\bibfield  {title} {\bibinfo {title} {{Quantum subspace expansion approach for simulating dynamical response functions of Kitaev spin liquids}},\ }\href {https://doi.org/10.1103/PhysRevMaterials.9.034401} {\bibfield  {journal} {\bibinfo  {journal} {Phys. Rev. Mater.}\ }\textbf {\bibinfo {volume} {9}},\ \bibinfo {pages} {034401} (\bibinfo {year} {2025})}\BibitemShut {NoStop}%
\bibitem [{\citenamefont {Greene-Diniz}\ \emph {et~al.}(2024)\citenamefont {Greene-Diniz}, \citenamefont {Manrique}, \citenamefont {Yamamoto}, \citenamefont {Plekhanov}, \citenamefont {Fitzpatrick}, \citenamefont {Krompiec}, \citenamefont {Sakuma},\ and\ \citenamefont {Ramo}}]{greene2023quantum}%
  \BibitemOpen
  \bibfield  {author} {\bibinfo {author} {\bibfnamefont {G.}~\bibnamefont {Greene-Diniz}}, \bibinfo {author} {\bibfnamefont {D.~Z.}\ \bibnamefont {Manrique}}, \bibinfo {author} {\bibfnamefont {K.}~\bibnamefont {Yamamoto}}, \bibinfo {author} {\bibfnamefont {E.}~\bibnamefont {Plekhanov}}, \bibinfo {author} {\bibfnamefont {N.}~\bibnamefont {Fitzpatrick}}, \bibinfo {author} {\bibfnamefont {M.}~\bibnamefont {Krompiec}}, \bibinfo {author} {\bibfnamefont {R.}~\bibnamefont {Sakuma}},\ and\ \bibinfo {author} {\bibfnamefont {D.~M.}\ \bibnamefont {Ramo}},\ }\bibfield  {title} {\bibinfo {title} {{Quantum computed Green's functions using a cumulant expansion of the Lanczos method}},\ }\href {https://quantum-journal.org/papers/q-2024-06-20-1383/} {\bibfield  {journal} {\bibinfo  {journal} {Quantum}\ }\textbf {\bibinfo {volume} {8}},\ \bibinfo {pages} {1383} (\bibinfo {year} {2024})}\BibitemShut {NoStop}%
\bibitem [{\citenamefont {Suchsland}\ \emph {et~al.}(2021)\citenamefont {Suchsland}, \citenamefont {Tacchino}, \citenamefont {Fischer}, \citenamefont {Neupert}, \citenamefont {Barkoutsos},\ and\ \citenamefont {Tavernelli}}]{suchsland2021algorithmic}%
  \BibitemOpen
  \bibfield  {author} {\bibinfo {author} {\bibfnamefont {P.}~\bibnamefont {Suchsland}}, \bibinfo {author} {\bibfnamefont {F.}~\bibnamefont {Tacchino}}, \bibinfo {author} {\bibfnamefont {M.~H.}\ \bibnamefont {Fischer}}, \bibinfo {author} {\bibfnamefont {T.}~\bibnamefont {Neupert}}, \bibinfo {author} {\bibfnamefont {P.~K.}\ \bibnamefont {Barkoutsos}},\ and\ \bibinfo {author} {\bibfnamefont {I.}~\bibnamefont {Tavernelli}},\ }\bibfield  {title} {\bibinfo {title} {Algorithmic error mitigation scheme for current quantum processors},\ }\href {https://quantum-journal.org/papers/q-2021-07-01-492/} {\bibfield  {journal} {\bibinfo  {journal} {Quantum}\ }\textbf {\bibinfo {volume} {5}},\ \bibinfo {pages} {492} (\bibinfo {year} {2021})}\BibitemShut {NoStop}%
\bibitem [{\citenamefont {Sugisaki}\ \emph {et~al.}(2025)\citenamefont {Sugisaki}, \citenamefont {Kanno}, \citenamefont {Itoko}, \citenamefont {Sakuma},\ and\ \citenamefont {Yamamoto}}]{sugisaki2024hamiltonian}%
  \BibitemOpen
  \bibfield  {author} {\bibinfo {author} {\bibfnamefont {K.}~\bibnamefont {Sugisaki}}, \bibinfo {author} {\bibfnamefont {S.}~\bibnamefont {Kanno}}, \bibinfo {author} {\bibfnamefont {T.}~\bibnamefont {Itoko}}, \bibinfo {author} {\bibfnamefont {R.}~\bibnamefont {Sakuma}},\ and\ \bibinfo {author} {\bibfnamefont {N.}~\bibnamefont {Yamamoto}},\ }\bibfield  {title} {\bibinfo {title} {Hamiltonian simulation-based quantum-selected configuration interaction for large-scale electronic structure calculations with a quantum computer},\ }\href {https://doi.org/10.1039/d5cp02202a} {\bibfield  {journal} {\bibinfo  {journal} {Phys. Chem. Chem. Phys.}\ }\textbf {\bibinfo {volume} {27}},\ \bibinfo {pages} {20869} (\bibinfo {year} {2025})}\BibitemShut {NoStop}%
\bibitem [{\citenamefont {Mikkelsen}\ and\ \citenamefont {Nakagawa}(2025)}]{mikkelsen2024quantum}%
  \BibitemOpen
  \bibfield  {author} {\bibinfo {author} {\bibfnamefont {M.}~\bibnamefont {Mikkelsen}}\ and\ \bibinfo {author} {\bibfnamefont {Y.~O.}\ \bibnamefont {Nakagawa}},\ }\bibfield  {title} {\bibinfo {title} {Quantum-selected configuration interaction with time-evolved state},\ }\href {https://doi.org/10.1103/75pv-hbrx} {\bibfield  {journal} {\bibinfo  {journal} {Phys. Rev. Res.}\ }\textbf {\bibinfo {volume} {7}},\ \bibinfo {pages} {043043} (\bibinfo {year} {2025})}\BibitemShut {NoStop}%
\bibitem [{\citenamefont {Yu}\ \emph {et~al.}(2025)\citenamefont {Yu}, \citenamefont {Moreno}, \citenamefont {Iosue}, \citenamefont {Bertels}, \citenamefont {Claudino}, \citenamefont {Fuller}, \citenamefont {Groszkowski}, \citenamefont {Humble}, \citenamefont {Jurcevic}, \citenamefont {Kirby} \emph {et~al.}}]{yu2025quantum}%
  \BibitemOpen
  \bibfield  {author} {\bibinfo {author} {\bibfnamefont {J.}~\bibnamefont {Yu}}, \bibinfo {author} {\bibfnamefont {J.~R.}\ \bibnamefont {Moreno}}, \bibinfo {author} {\bibfnamefont {J.~T.}\ \bibnamefont {Iosue}}, \bibinfo {author} {\bibfnamefont {L.}~\bibnamefont {Bertels}}, \bibinfo {author} {\bibfnamefont {D.}~\bibnamefont {Claudino}}, \bibinfo {author} {\bibfnamefont {B.}~\bibnamefont {Fuller}}, \bibinfo {author} {\bibfnamefont {P.}~\bibnamefont {Groszkowski}}, \bibinfo {author} {\bibfnamefont {T.~S.}\ \bibnamefont {Humble}}, \bibinfo {author} {\bibfnamefont {P.}~\bibnamefont {Jurcevic}}, \bibinfo {author} {\bibfnamefont {W.}~\bibnamefont {Kirby}}, \emph {et~al.},\ }\bibfield  {title} {\bibinfo {title} {{Quantum-centric algorithm for sample-based krylov diagonalization}},\ }\href {https://arxiv.org/abs/2501.09702} {\bibfield  {journal} {\bibinfo  {journal} {arXiv:2501.09702}\ } (\bibinfo {year} {2025})}\BibitemShut {NoStop}%
\bibitem [{\citenamefont {Wan}\ \emph {et~al.}(2022)\citenamefont {Wan}, \citenamefont {Berta},\ and\ \citenamefont {Campbell}}]{wan2022randomized}%
  \BibitemOpen
  \bibfield  {author} {\bibinfo {author} {\bibfnamefont {K.}~\bibnamefont {Wan}}, \bibinfo {author} {\bibfnamefont {M.}~\bibnamefont {Berta}},\ and\ \bibinfo {author} {\bibfnamefont {E.~T.}\ \bibnamefont {Campbell}},\ }\bibfield  {title} {\bibinfo {title} {{Randomized Quantum Algorithm for Statistical Phase Estimation}},\ }\href {https://doi.org/10.1103/PhysRevLett.129.030503} {\bibfield  {journal} {\bibinfo  {journal} {Phys. Rev. Lett.}\ }\textbf {\bibinfo {volume} {129}},\ \bibinfo {pages} {030503} (\bibinfo {year} {2022})}\BibitemShut {NoStop}%
\bibitem [{\citenamefont {Wang}\ \emph {et~al.}(2023)\citenamefont {Wang}, \citenamefont {Fran{\c{c}}a}, \citenamefont {Zhang}, \citenamefont {Zhu},\ and\ \citenamefont {Johnson}}]{wang2023quantum}%
  \BibitemOpen
  \bibfield  {author} {\bibinfo {author} {\bibfnamefont {G.}~\bibnamefont {Wang}}, \bibinfo {author} {\bibfnamefont {D.~S.}\ \bibnamefont {Fran{\c{c}}a}}, \bibinfo {author} {\bibfnamefont {R.}~\bibnamefont {Zhang}}, \bibinfo {author} {\bibfnamefont {S.}~\bibnamefont {Zhu}},\ and\ \bibinfo {author} {\bibfnamefont {P.~D.}\ \bibnamefont {Johnson}},\ }\bibfield  {title} {\bibinfo {title} {Quantum algorithm for ground state energy estimation using circuit depth with exponentially improved dependence on precision},\ }\href {https://quantum-journal.org/papers/q-2023-11-06-1167/} {\bibfield  {journal} {\bibinfo  {journal} {Quantum}\ }\textbf {\bibinfo {volume} {7}},\ \bibinfo {pages} {1167} (\bibinfo {year} {2023})}\BibitemShut {NoStop}%
\bibitem [{\citenamefont {Ding}\ and\ \citenamefont {Lin}(2023)}]{ding2023qcels}%
  \BibitemOpen
  \bibfield  {author} {\bibinfo {author} {\bibfnamefont {Z.}~\bibnamefont {Ding}}\ and\ \bibinfo {author} {\bibfnamefont {L.}~\bibnamefont {Lin}},\ }\bibfield  {title} {\bibinfo {title} {Even shorter quantum circuit for phase estimation on early fault-tolerant quantum computers with applications to ground-state energy estimation},\ }\href {https://doi.org/10.1103/PRXQuantum.4.020331} {\bibfield  {journal} {\bibinfo  {journal} {PRX Quantum}\ }\textbf {\bibinfo {volume} {4}},\ \bibinfo {pages} {020331} (\bibinfo {year} {2023})}\BibitemShut {NoStop}%
\bibitem [{\citenamefont {Ni}\ \emph {et~al.}(2023)\citenamefont {Ni}, \citenamefont {Li},\ and\ \citenamefont {Ying}}]{ni2023low}%
  \BibitemOpen
  \bibfield  {author} {\bibinfo {author} {\bibfnamefont {H.}~\bibnamefont {Ni}}, \bibinfo {author} {\bibfnamefont {H.}~\bibnamefont {Li}},\ and\ \bibinfo {author} {\bibfnamefont {L.}~\bibnamefont {Ying}},\ }\bibfield  {title} {\bibinfo {title} {On low-depth algorithms for quantum phase estimation},\ }\href {https://quantum-journal.org/papers/q-2023-11-06-1165/} {\bibfield  {journal} {\bibinfo  {journal} {Quantum}\ }\textbf {\bibinfo {volume} {7}},\ \bibinfo {pages} {1165} (\bibinfo {year} {2023})}\BibitemShut {NoStop}%
\bibitem [{\citenamefont {Kshirsagar}\ \emph {et~al.}(2024)\citenamefont {Kshirsagar}, \citenamefont {Katabarwa},\ and\ \citenamefont {Johnson}}]{kshirsagar2024proving}%
  \BibitemOpen
  \bibfield  {author} {\bibinfo {author} {\bibfnamefont {R.}~\bibnamefont {Kshirsagar}}, \bibinfo {author} {\bibfnamefont {A.}~\bibnamefont {Katabarwa}},\ and\ \bibinfo {author} {\bibfnamefont {P.~D.}\ \bibnamefont {Johnson}},\ }\bibfield  {title} {\bibinfo {title} {On proving the robustness of algorithms for early fault-tolerant quantum computers},\ }\href {https://quantum-journal.org/papers/q-2024-11-20-1531/} {\bibfield  {journal} {\bibinfo  {journal} {Quantum}\ }\textbf {\bibinfo {volume} {8}},\ \bibinfo {pages} {1531} (\bibinfo {year} {2024})}\BibitemShut {NoStop}%
\bibitem [{\citenamefont {Liang}\ \emph {et~al.}(2024)\citenamefont {Liang}, \citenamefont {Zhou}, \citenamefont {Dalal},\ and\ \citenamefont {Johnson}}]{liang2024modeling}%
  \BibitemOpen
  \bibfield  {author} {\bibinfo {author} {\bibfnamefont {Q.}~\bibnamefont {Liang}}, \bibinfo {author} {\bibfnamefont {Y.}~\bibnamefont {Zhou}}, \bibinfo {author} {\bibfnamefont {A.}~\bibnamefont {Dalal}},\ and\ \bibinfo {author} {\bibfnamefont {P.}~\bibnamefont {Johnson}},\ }\bibfield  {title} {\bibinfo {title} {Modeling the performance of early fault-tolerant quantum algorithms},\ }\href {https://doi.org/10.1103/PhysRevResearch.6.023118} {\bibfield  {journal} {\bibinfo  {journal} {Phys. Rev. Res.}\ }\textbf {\bibinfo {volume} {6}},\ \bibinfo {pages} {023118} (\bibinfo {year} {2024})}\BibitemShut {NoStop}%
\bibitem [{\citenamefont {Wang}\ \emph {et~al.}(2025)\citenamefont {Wang}, \citenamefont {Fran\ifmmode~\mbox{\c{c}}\else \c{c}\fi{}a}, \citenamefont {Rendon},\ and\ \citenamefont {Johnson}}]{wang2025efficient}%
  \BibitemOpen
  \bibfield  {author} {\bibinfo {author} {\bibfnamefont {G.}~\bibnamefont {Wang}}, \bibinfo {author} {\bibfnamefont {D.~S.}\ \bibnamefont {Fran\ifmmode~\mbox{\c{c}}\else \c{c}\fi{}a}}, \bibinfo {author} {\bibfnamefont {G.}~\bibnamefont {Rendon}},\ and\ \bibinfo {author} {\bibfnamefont {P.~D.}\ \bibnamefont {Johnson}},\ }\bibfield  {title} {\bibinfo {title} {Efficient ground-state-energy estimation and certification on early fault-tolerant quantum computers},\ }\href {https://doi.org/10.1103/PhysRevA.111.012426} {\bibfield  {journal} {\bibinfo  {journal} {Phys. Rev. A}\ }\textbf {\bibinfo {volume} {111}},\ \bibinfo {pages} {012426} (\bibinfo {year} {2025})}\BibitemShut {NoStop}%
\bibitem [{\citenamefont {Sun}\ \emph {et~al.}(2018)\citenamefont {Sun}, \citenamefont {Berkelbach}, \citenamefont {Blunt}, \citenamefont {Booth}, \citenamefont {Guo}, \citenamefont {Li}, \citenamefont {Liu}, \citenamefont {McClain}, \citenamefont {Sayfutyarova}, \citenamefont {Sharma} \emph {et~al.}}]{sun2018pyscf}%
  \BibitemOpen
  \bibfield  {author} {\bibinfo {author} {\bibfnamefont {Q.}~\bibnamefont {Sun}}, \bibinfo {author} {\bibfnamefont {T.~C.}\ \bibnamefont {Berkelbach}}, \bibinfo {author} {\bibfnamefont {N.~S.}\ \bibnamefont {Blunt}}, \bibinfo {author} {\bibfnamefont {G.~H.}\ \bibnamefont {Booth}}, \bibinfo {author} {\bibfnamefont {S.}~\bibnamefont {Guo}}, \bibinfo {author} {\bibfnamefont {Z.}~\bibnamefont {Li}}, \bibinfo {author} {\bibfnamefont {J.}~\bibnamefont {Liu}}, \bibinfo {author} {\bibfnamefont {J.~D.}\ \bibnamefont {McClain}}, \bibinfo {author} {\bibfnamefont {E.~R.}\ \bibnamefont {Sayfutyarova}}, \bibinfo {author} {\bibfnamefont {S.}~\bibnamefont {Sharma}}, \emph {et~al.},\ }\bibfield  {title} {\bibinfo {title} {{PySCF: the Python-based simulations of chemistry framework}},\ }\href {https://wires.onlinelibrary.wiley.com/doi/abs/10.1002/wcms.1340} {\bibfield  {journal} {\bibinfo  {journal} {Wiley Interdisciplinary Reviews: Computational Molecular Science}\ }\textbf {\bibinfo {volume} {8}},\ \bibinfo {pages} {e1340}
  (\bibinfo {year} {2018})}\BibitemShut {NoStop}%
\bibitem [{\citenamefont {Jordan}\ and\ \citenamefont {Wigner}(1928)}]{Jordan1928}%
  \BibitemOpen
  \bibfield  {author} {\bibinfo {author} {\bibfnamefont {P.}~\bibnamefont {Jordan}}\ and\ \bibinfo {author} {\bibfnamefont {E.}~\bibnamefont {Wigner}},\ }\bibfield  {title} {\bibinfo {title} {{\"U}ber das paulische {\"a}quivalenzverbot},\ }\href {https://doi.org/10.1007/BF01331938} {\bibfield  {journal} {\bibinfo  {journal} {Zeitschrift f{\"u}r Physik}\ }\textbf {\bibinfo {volume} {47}},\ \bibinfo {pages} {631} (\bibinfo {year} {1928})}\BibitemShut {NoStop}%
\bibitem [{\citenamefont {Sung}\ \emph {et~al.}(2026)\citenamefont {Sung}, \citenamefont {Choi}, \citenamefont {Amico}, \citenamefont {Andrews}, \citenamefont {Ayantuna}, \citenamefont {Kawashima}, \citenamefont {Lin}, \citenamefont {Omanovic}, \citenamefont {Piccinelli}, \citenamefont {Moreno}, \citenamefont {Saki}, \citenamefont {Shee}, \citenamefont {Shin}, \citenamefont {Tran}, \citenamefont {Ueda}, \citenamefont {Zhang},\ and\ \citenamefont {Motta}}]{ffsim}%
  \BibitemOpen
  \bibfield  {author} {\bibinfo {author} {\bibfnamefont {K.~J.}\ \bibnamefont {Sung}}, \bibinfo {author} {\bibfnamefont {I.}~\bibnamefont {Choi}}, \bibinfo {author} {\bibfnamefont {M.}~\bibnamefont {Amico}}, \bibinfo {author} {\bibfnamefont {B.}~\bibnamefont {Andrews}}, \bibinfo {author} {\bibfnamefont {E.}~\bibnamefont {Ayantuna}}, \bibinfo {author} {\bibfnamefont {Y.}~\bibnamefont {Kawashima}}, \bibinfo {author} {\bibfnamefont {W.-H.}\ \bibnamefont {Lin}}, \bibinfo {author} {\bibfnamefont {D.}~\bibnamefont {Omanovic}}, \bibinfo {author} {\bibfnamefont {S.}~\bibnamefont {Piccinelli}}, \bibinfo {author} {\bibfnamefont {J.~R.}\ \bibnamefont {Moreno}}, \bibinfo {author} {\bibfnamefont {A.~A.}\ \bibnamefont {Saki}}, \bibinfo {author} {\bibfnamefont {J.}~\bibnamefont {Shee}}, \bibinfo {author} {\bibfnamefont {S.}~\bibnamefont {Shin}}, \bibinfo {author} {\bibfnamefont {M.~C.}\ \bibnamefont {Tran}}, \bibinfo {author} {\bibfnamefont {K.}~\bibnamefont {Ueda}}, \bibinfo {author} {\bibfnamefont {H.}~\bibnamefont
  {Zhang}},\ and\ \bibinfo {author} {\bibfnamefont {M.}~\bibnamefont {Motta}},\ }\href {https://arxiv.org/abs/2605.03123} {\bibinfo {title} {ffsim: Faster simulation of fermionic quantum circuits}} (\bibinfo {year} {2026}),\ \Eprint {https://arxiv.org/abs/2605.03123} {arXiv:2605.03123 [quant-ph]} \BibitemShut {NoStop}%
\bibitem [{\citenamefont {Baer}(2018)}]{findiff}%
  \BibitemOpen
  \bibfield  {author} {\bibinfo {author} {\bibfnamefont {M.}~\bibnamefont {Baer}},\ }\href {https://github.com/maroba/findiff} {\bibinfo {title} {{findiff} software package}} (\bibinfo {year} {2018}),\ \bibinfo {note} {\url{https://github.com/maroba/findiff}}\BibitemShut {NoStop}%
\bibitem [{\citenamefont {Tropp}\ \emph {et~al.}(2015)\citenamefont {Tropp} \emph {et~al.}}]{tropp2015introduction}%
  \BibitemOpen
  \bibfield  {author} {\bibinfo {author} {\bibfnamefont {J.~A.}\ \bibnamefont {Tropp}} \emph {et~al.},\ }\bibfield  {title} {\bibinfo {title} {An introduction to matrix concentration inequalities},\ }\href {https://www.nowpublishers.com/article/Details/MAL-048} {\bibfield  {journal} {\bibinfo  {journal} {Foundations and Trends{\textregistered} in Machine Learning}\ }\textbf {\bibinfo {volume} {8}},\ \bibinfo {pages} {1} (\bibinfo {year} {2015})}\BibitemShut {NoStop}%
\bibitem [{\citenamefont {Virtanen}\ \emph {et~al.}(2020)\citenamefont {Virtanen}, \citenamefont {Gommers}, \citenamefont {Oliphant}, \citenamefont {Haberland}, \citenamefont {Reddy}, \citenamefont {Cournapeau}, \citenamefont {Burovski}, \citenamefont {Peterson}, \citenamefont {Weckesser}, \citenamefont {Bright}, \citenamefont {{van der Walt}}, \citenamefont {Brett}, \citenamefont {Wilson}, \citenamefont {Millman}, \citenamefont {Mayorov}, \citenamefont {Nelson}, \citenamefont {Jones}, \citenamefont {Kern}, \citenamefont {Larson}, \citenamefont {Carey}, \citenamefont {Polat}, \citenamefont {Feng}, \citenamefont {Moore}, \citenamefont {{VanderPlas}}, \citenamefont {Laxalde}, \citenamefont {Perktold}, \citenamefont {Cimrman}, \citenamefont {Henriksen}, \citenamefont {Quintero}, \citenamefont {Harris}, \citenamefont {Archibald}, \citenamefont {Ribeiro}, \citenamefont {Pedregosa}, \citenamefont {{van Mulbregt}},\ and\ \citenamefont {{SciPy 1.0 Contributors}}}]{2020SciPy-NMeth}%
  \BibitemOpen
  \bibfield  {author} {\bibinfo {author} {\bibfnamefont {P.}~\bibnamefont {Virtanen}}, \bibinfo {author} {\bibfnamefont {R.}~\bibnamefont {Gommers}}, \bibinfo {author} {\bibfnamefont {T.~E.}\ \bibnamefont {Oliphant}}, \bibinfo {author} {\bibfnamefont {M.}~\bibnamefont {Haberland}}, \bibinfo {author} {\bibfnamefont {T.}~\bibnamefont {Reddy}}, \bibinfo {author} {\bibfnamefont {D.}~\bibnamefont {Cournapeau}}, \bibinfo {author} {\bibfnamefont {E.}~\bibnamefont {Burovski}}, \bibinfo {author} {\bibfnamefont {P.}~\bibnamefont {Peterson}}, \bibinfo {author} {\bibfnamefont {W.}~\bibnamefont {Weckesser}}, \bibinfo {author} {\bibfnamefont {J.}~\bibnamefont {Bright}}, \bibinfo {author} {\bibfnamefont {S.~J.}\ \bibnamefont {{van der Walt}}}, \bibinfo {author} {\bibfnamefont {M.}~\bibnamefont {Brett}}, \bibinfo {author} {\bibfnamefont {J.}~\bibnamefont {Wilson}}, \bibinfo {author} {\bibfnamefont {K.~J.}\ \bibnamefont {Millman}}, \bibinfo {author} {\bibfnamefont {N.}~\bibnamefont {Mayorov}}, \bibinfo {author} {\bibfnamefont
  {A.~R.~J.}\ \bibnamefont {Nelson}}, \bibinfo {author} {\bibfnamefont {E.}~\bibnamefont {Jones}}, \bibinfo {author} {\bibfnamefont {R.}~\bibnamefont {Kern}}, \bibinfo {author} {\bibfnamefont {E.}~\bibnamefont {Larson}}, \bibinfo {author} {\bibfnamefont {C.~J.}\ \bibnamefont {Carey}}, \bibinfo {author} {\bibfnamefont {{\.I}.}~\bibnamefont {Polat}}, \bibinfo {author} {\bibfnamefont {Y.}~\bibnamefont {Feng}}, \bibinfo {author} {\bibfnamefont {E.~W.}\ \bibnamefont {Moore}}, \bibinfo {author} {\bibfnamefont {J.}~\bibnamefont {{VanderPlas}}}, \bibinfo {author} {\bibfnamefont {D.}~\bibnamefont {Laxalde}}, \bibinfo {author} {\bibfnamefont {J.}~\bibnamefont {Perktold}}, \bibinfo {author} {\bibfnamefont {R.}~\bibnamefont {Cimrman}}, \bibinfo {author} {\bibfnamefont {I.}~\bibnamefont {Henriksen}}, \bibinfo {author} {\bibfnamefont {E.~A.}\ \bibnamefont {Quintero}}, \bibinfo {author} {\bibfnamefont {C.~R.}\ \bibnamefont {Harris}}, \bibinfo {author} {\bibfnamefont {A.~M.}\ \bibnamefont {Archibald}}, \bibinfo {author}
  {\bibfnamefont {A.~H.}\ \bibnamefont {Ribeiro}}, \bibinfo {author} {\bibfnamefont {F.}~\bibnamefont {Pedregosa}}, \bibinfo {author} {\bibfnamefont {P.}~\bibnamefont {{van Mulbregt}}},\ and\ \bibinfo {author} {\bibnamefont {{SciPy 1.0 Contributors}}},\ }\bibfield  {title} {\bibinfo {title} {{{SciPy} 1.0: Fundamental Algorithms for Scientific Computing in Python}},\ }\href {https://doi.org/10.1038/s41592-019-0686-2} {\bibfield  {journal} {\bibinfo  {journal} {Nature Methods}\ }\textbf {\bibinfo {volume} {17}},\ \bibinfo {pages} {261} (\bibinfo {year} {2020})}\BibitemShut {NoStop}%
\bibitem [{\citenamefont {Bradbury}\ \emph {et~al.}(2018)\citenamefont {Bradbury}, \citenamefont {Frostig}, \citenamefont {Hawkins}, \citenamefont {Johnson}, \citenamefont {Leary}, \citenamefont {Maclaurin}, \citenamefont {Necula}, \citenamefont {Paszke}, \citenamefont {Vander{P}las}, \citenamefont {Wanderman-{M}ilne},\ and\ \citenamefont {Zhang}}]{jax2018github}%
  \BibitemOpen
  \bibfield  {author} {\bibinfo {author} {\bibfnamefont {J.}~\bibnamefont {Bradbury}}, \bibinfo {author} {\bibfnamefont {R.}~\bibnamefont {Frostig}}, \bibinfo {author} {\bibfnamefont {P.}~\bibnamefont {Hawkins}}, \bibinfo {author} {\bibfnamefont {M.~J.}\ \bibnamefont {Johnson}}, \bibinfo {author} {\bibfnamefont {C.}~\bibnamefont {Leary}}, \bibinfo {author} {\bibfnamefont {D.}~\bibnamefont {Maclaurin}}, \bibinfo {author} {\bibfnamefont {G.}~\bibnamefont {Necula}}, \bibinfo {author} {\bibfnamefont {A.}~\bibnamefont {Paszke}}, \bibinfo {author} {\bibfnamefont {J.}~\bibnamefont {Vander{P}las}}, \bibinfo {author} {\bibfnamefont {S.}~\bibnamefont {Wanderman-{M}ilne}},\ and\ \bibinfo {author} {\bibfnamefont {Q.}~\bibnamefont {Zhang}},\ }\href {http://github.com/jax-ml/jax} {\bibinfo {title} {{JAX}: composable transformations of {P}ython+{N}um{P}y programs}} (\bibinfo {year} {2018})\BibitemShut {NoStop}%
\bibitem [{\citenamefont {Larsson}\ \emph {et~al.}(2022)\citenamefont {Larsson}, \citenamefont {Zhai}, \citenamefont {Umrigar},\ and\ \citenamefont {Chan}}]{larsson2022chromium}%
  \BibitemOpen
  \bibfield  {author} {\bibinfo {author} {\bibfnamefont {H.~R.}\ \bibnamefont {Larsson}}, \bibinfo {author} {\bibfnamefont {H.}~\bibnamefont {Zhai}}, \bibinfo {author} {\bibfnamefont {C.~J.}\ \bibnamefont {Umrigar}},\ and\ \bibinfo {author} {\bibfnamefont {G.~K.-L.}\ \bibnamefont {Chan}},\ }\bibfield  {title} {\bibinfo {title} {The chromium dimer: closing a chapter of quantum chemistry},\ }\href {https://pubs.acs.org/doi/10.1021/jacs.2c06357} {\bibfield  {journal} {\bibinfo  {journal} {Journal of the American Chemical Society}\ }\textbf {\bibinfo {volume} {144}},\ \bibinfo {pages} {15932} (\bibinfo {year} {2022})}\BibitemShut {NoStop}%
\end{thebibliography}%

\clearpage
\onecolumngrid   

\setcounter{secnumdepth}{2}   
\setcounter{section}{0}
\setcounter{equation}{0}
\setcounter{figure}{0}
\setcounter{table}{0}
\renewcommand{\theequation}{S\arabic{equation}}
\renewcommand{\thefigure}{S\arabic{figure}}
\renewcommand{\thetable}{S\arabic{table}}

\begin{center}
  \large\bfseries Supplementary Information
\end{center}
\vspace{1em}


\section{Error bound of the KQD algorithm}
\label{si:kqd_error}
Here, we review the general error analysis of KQD, following Ref.~\cite{epperly2022theory}.
We characterize the noise strength for the matrix pair $(\mathbf{H},\mathbf{S})$ as follows: 
\begin{align}
    \eta := \sqrt{\norm{\bm{\Delta}_{\mathbf{S}}}^2 + \norm{\bm{\Delta}_{\mathbf{H}}}^2}, \label{eq:matrix_noise_rate}
\end{align}
where $\mathbf{\Delta}_{\mathbf{S}}$ and $\mathbf{\Delta}_{\mathbf{H}}$ are the matrix perturbation defined in Eq.~\eqref{eq:noise_H}, and $\norm{\cdots}$ denotes the spectral norm. Let 
\begin{align}
    \ket{\phi_0} = \sum_{j=0}^{N-1}\gamma_j\ket{\psi_j}, \label{eq:initial_state_decomposition}
\end{align}
be the eigenstate decomposition of the initial state $\ket{\phi_0}$, with $\ket{\psi_j}$ being the exact eigenstate for the Hamiltonian $\hat{H}$ corresponding to the $j$-th lowest energy (i.e., $E_0\leq E_1\leq \cdots \leq E_{N-1}$ and $\hat{H}\ket{\psi_j}=E_j\ket{\psi_j}$). 
$N (\leq \mathrm{dim}\hat{H})$ is the number of eigenstates required to span the initial state $\ket{\phi_0}$. 
Suppose that we construct a unitary Krylov subspace~\eqref{eq:krylov_space} for this initial state $\ket{\phi_0}$ with a timestep $\tau=\pi/\Delta{E}_{N-1}$, where $\Delta{E}_j:=E_j-E_0$. Then, Ref.~\cite{epperly2022theory} showed that the error of the ground state energy estimate from the noisy regularized GEVP~\eqref{eq:GEVP_threshold} is upper bounded as follows: 
\begin{align}
    &\left| \arctan{(\tilde{E}_0^{(n_\epsilon)})} - \arctan{(E_0)} \right| \leq \order{\frac{1-|\gamma_0|^2}{|\gamma_0|^2}e^{-\order{\frac{n\Delta{E}_1}{\Delta{E}_{N-1}}}} + \left[\frac{\Delta{E}_{N-1}}{|\gamma_0|^2}+d_0^{-1}\right]\eta^{\frac{1}{1+\alpha}}},  \label{eq:kqd_energy_error}
\end{align}
where $0\leq\alpha \leq1/2$ is a constant, and the threshold level is supposed to be $\epsilon=\Theta{(\eta^{\frac{1}{1+\alpha}})}$. $d_0^{-1}$ is the condition number defined as  
\begin{align}
    d_0^{-1} := |\mathbf{x}_0^{\dag}(\tilde{\mathbf{A}}+i\tilde{\mathbf{B}})\mathbf{x}_0|^{-1}, \label{eq:condition_number}
\end{align}
where $(\tilde{\mathbf{A}},\tilde{\mathbf{B}}):=(\tilde{\mathbf{V}}_{>\epsilon}^{\dag}\tilde{\mathbf{H}}\tilde{\mathbf{V}}_{>\epsilon},\tilde{\mathbf{V}}_{>\epsilon}^{\dag}\tilde{\mathbf{S}}\tilde{\mathbf{V}}_{>\epsilon})$ is the noisy regularized matrix pair, and $\mathbf{x}_0$ is its unit-norm eigenvector for the lowest eigenvalue. 
The first term on the right-hand side of Eq.~\eqref{eq:kqd_energy_error} describes the error due to the projection of the many-body Hamiltonian problem into the order-$n$ Krylov subspace. This projection error is suppressed exponentially with increasing the Krylov subspace dimension $n$ under the assumption of the spectral gap $\Delta{E}_1>0$. 
The second term on the right-hand side of Eq.~\eqref{eq:kqd_energy_error} indicates the effect of the matrix perturbation. Since it is sublinearly proportional to the noise strength $\eta$, we can reduce the energy estimation error by suppressing the matrix perturbations.

\section{Sampling variance of the Hadamard test}\label{si:hadamard}
Here, we analyze the statistical behavior of Hadamard test sampling. The analysis presented in this section is based on Refs.~\cite{lee2024sampling,lee2025efficient}. 

In the Hadamard test, the real and imaginary parts of the expectation value of a unitary operator $\hat{U}$ are estimated separately for an input state $\ket{\psi}$ as follows: 
\begin{align}
    \mathbb{E}{[\mathcal{Z}]} &=  \mathbb{E}{[\mathcal{R}]} +  i\mathbb{E}{[\mathcal{I}]} = \mathrm{Re}{\bra{\psi}\hat{U}\ket{\psi}} + i\mathrm{Im}{\bra{\psi}\hat{U}\ket{\psi}},  \label{eq:hadamard_estimator}
\end{align}
where $\mathcal{Z}=\mathcal{R}+i\mathcal{I}$ is an estimator of $\bra{\psi}\hat{U}\ket{\psi}$ with $\mathcal{R}$ and $\mathcal{I}$ being the real and imaginary parts, respectively. Specifically, the sampling is performed using the quantum circuit in Fig.~\ref{fig:measurement_circuit}. In this scenario, the probability of finding $\ket{0}_a$ state in the ancilla register, $p_0$, is related to the expectation value of the estimators$\{\mathcal{R},\mathcal{I}\}=:\mathcal{X}$ as 
\begin{align}
    \mathbb{E}{[\mathcal{X}]} = 2p_0-1 . \label{eq:exp_hadamard}
\end{align}
Here, the random variables $\mathcal{X}=\{\mathcal{R},\mathcal{I}\}$ follow the averaged Bernoulli distributions as  
\begin{align}
    \mathcal{X} \sim \frac{2}{m}\mathrm{Bin}(m,p) - 1,
 \end{align}
 where $\mathrm{Bin}(m,p)$ denotes a binomial distribution, describing the number of successes in $m$ identical experiments with success probability $p$. Thus, the sampling variance of $\mathcal{X}=\{\mathcal{R},\mathcal{I}\}$ is given by 
\begin{align}
    \mathrm{Var}{[\mathcal{X}]} &= \frac{4p(1-p)}{m} = \frac{1-(\mathbb{E}[\mathcal{X}])^2}{m} . \label{eq:var_hadamard}
\end{align}
From Eqs.~\eqref{eq:hadamard_estimator}, \eqref{eq:exp_hadamard} and \eqref{eq:var_hadamard}, the sampling variance of the estimator $\mathcal{Z}=\mathcal{R}+i\mathcal{I}$ is described as 
\begin{align}
   \mathrm{Var}{[\mathcal{Z}]} &=  \frac{1-\mathrm{Re}[\bra{\psi}\hat{U}\ket{\psi}]^2}{m^{(\rm r)}} + \frac{1-\mathrm{Im}[\bra{\psi}\hat{U}\ket{\psi}]^2}{m^{(\rm i)}} , \label{eq:var_Z}
\end{align}
where $m^{\rm (r)}$ and $m^{\rm (i)}$ denote the number of samples for the real part $\mathcal{R}$ and imaginary part $\mathcal{I}$, respectively. 

Equation~\eqref{eq:var_Z} indicates that the sampling variance is minimized when $m^{(\rm r)}\propto \sqrt{1-\mathrm{Re}[\bra{\psi}\hat{U}\ket{\psi}]^2}$ and $m^{(\rm i)}\propto \sqrt{1-\mathrm{Im}[\bra{\psi}\hat{U}\ket{\psi}]^2}$. However, it is infeasible to apply such a shot allocation before measuring $\bra{\psi}\hat{U}\ket{\psi}$ by the Hadamard test sampling. 
To eliminate the dependence on the input state $\ket{\psi}$, we take an average of the state over the uniform Haar random distribution $\mathcal{H}(d)$, where $d$ is the dimension of the Hilbert space. 
To do this, we apply the following relation from Schur's lemma~\cite{lee2025efficient} 
\begin{align}
    \mathbb{E}_{\hat{V}\sim\mathcal{U}(d)}[\hat{V}^{\dag}\hat{A}\hat{V}] = \frac{\hat{\1}}{d}\mathrm{Tr}[\hat{A}], \label{eq:schur_theorem}
\end{align}
where $\hat{A}\in\mathbb{C}^{d\times d}$ is any operator and $\mathcal{U}(d)$ denotes the uniform Haar random distribution over the $d$-dimensional unitary group. Equation~\eqref{eq:schur_theorem} leads to the following equality: 
\begin{align}
    \mathbb{E}_{\psi\sim\mathcal{H}(d)}\left[\bra{\psi}\hat{U}\ket{\psi}\right] = \frac{1}{d}\mathrm{Tr}[\hat{U}]. \label{eq:haar_ave_trace}
\end{align}
Hereafter, we use the notation $\mathbb{E}_{\psi}[\cdots]$ to denote the Haar random average $\mathbb{E}_{\psi\sim\mathcal{H}(d)}[\cdots]$ unless otherwise mentioned. From Eq.~\eqref{eq:haar_ave_trace}, we obtain
\begin{align}
    \mathbb{E}_{\psi,\hat{U}}{\left[\abs{\bra{\psi}\hat{U}\ket{\psi}}^2\right]} 
    &= \mathbb{E}_{\psi,\hat{U}}{\left[\bra{\psi}\hat{U}\ket{\psi}\bra{\psi}\hat{U}^{\dag}\ket{\psi}\right]} \nonumber\\ 
    &= \frac{1}{d}\mathbb{E}_{\psi,\hat{U}}\left[ \mathrm{Tr}[\hat{U}\ket{\psi}\bra{\psi}\hat{U}^{\dag} \right] \nonumber\\
    &= \frac{1}{d}\mathbb{E}_{\psi,\hat{U}}\left[ \bra{\psi}\hat{U}^{\dag} \hat{U}\ket{\psi}\right] \nonumber\\
    &= \frac{1}{d^2}\mathrm{Tr}[\hat{U}^\dag\hat{U}] = \frac{1}{d}. \label{eq:haar_average_1}
\end{align}
Furthermore, the unitary invariant property of the Haar random measure leads to 
\begin{align}
    &\mathbb{E}_{\psi,\hat{U}}{\left[\mathrm{Re}[\bra{\psi}\hat{U}\ket{\psi}]^2\right]} \nonumber\\ 
    &= \frac{1}{4}\mathbb{E}_{\psi,\hat{U}}\left[
    \bra{\psi}\hat{U}\ket{\psi}^2 + \bra{\psi}\hat{U}^{\dag}\ket{\psi}^2 + 2\abs{\bra{\psi}\hat{U}\ket{\psi}}^2 \right] \nonumber\\
    &= \frac{1}{4}\mathbb{E}_{\psi,\hat{U}}\left[
    \bra{\psi}i\hat{U}\ket{\psi}^2 + \bra{\psi}(-i)\hat{U}^{\dag}\ket{\psi}^2 + 2\abs{\bra{\psi}\hat{U}\ket{\psi}}^2 \right] \nonumber\\
    &=  \mathbb{E}_{\psi,\hat{U}}{\left[\mathrm{Im}[\bra{\psi}\hat{U}\ket{\psi}]^2\right]}. \label{eq:haar_average_2}
\end{align}
Applying Eqs.~\eqref{eq:haar_average_1} and \eqref{eq:haar_average_2} to Eq.~\eqref{eq:var_Z}, we obtain
\begin{align}
    \mathbb{E}_{\psi,\hat{U}}{\left[\mathrm{Var}{[\mathcal{Z}]}\right]} 
    &= \frac{2}{m}\left(2- \mathbb{E}_{\psi,\hat{U}}\left[|\bra{\psi}\hat{U}\ket{\psi}|^2\right] \right) = \frac{2}{m}\left( 2 - \frac{1}{d} \right), \label{eq:haar_average_variance}
\end{align}
where $m^{(\rm r)}=m^{(\rm i)}=m/2$. Equation~\eqref{eq:haar_average_variance} leads to the expressions of sampling variance (Eqs.~\eqref{eq:variance_S_opt} and \eqref{eq:variance_H}) in the main text. 

\section{Sampling error analysis}\label{si:sampling_error}

This section details the derivation of the sampling perturbation for KQD and MSD described in the main text. 

We first summarize the Hadamard-test sampling variance and the optimal shot allocation for KQD and MSD.
To estimate the sampling cost of KQD, the authors of Ref.~\cite {lee2024sampling} analyzed the sampling variance of the Hadamard test. They found that the sampling variance for the overlap matrix element $\tilde{\mathbf{S}}_{0k}$ is described as 
\begin{align}
    \mathrm{Var}[\tilde{\mathbf{S}}_{0k},m_k] \simeq \frac{2(1-\delta_{0k})}{m_k}\left(2-\frac{1}{d}\right), \label{eq:variance_S_opt}
\end{align}
where $\simeq$ indicates the Haar random averaging, $d=2^{N_q}$ is the dimension of the Hilbert space, and $m_k$ represents the number of shots allocated to estimate $\tilde{\mathbf{S}}_{0k}$ using the Hadamard test ($m_k/2$ shots for both real and imaginary parts). The variance is defined to be zero for $k=0$ since $\mathbf{S}_{00}=\mathbf{S}_{kk}=\braket{\phi_k|\phi_k}=1$ is precisely known by definition and thus requires no sampling. 
Considering the sampling variance of $\bra{\phi_0}\hat{P}_l\ket{\phi_k}$ based on the LCU decomposition in Eq.~\eqref{eq:LCU}, the sampling variance for the noisy Hamiltonian matrix element $\tilde{\mathbf{H}}_{0k}$ is similarly expressed as  
\begin{align}
    \mathrm{Var}[\tilde{\mathbf{H}}_{0k},\vec{m}_k] &\simeq\sum_{l=1}^{L}\frac{(2-\delta_{0k})|c_l|^2}{m_{kl}}\left(2-\frac{1}{d}\right), \label{eq:variance_H}
\end{align}
where $\vec{m}_k=\{m_{kl}^{(\rm r)},m_{kl}^{(\rm i)}\}_{l=1}^{L}$ denotes the shot allocation for each component, with $m_{kl}^{(\rm r)}$ ($m_{kl}^{(\rm i)}$) being the number of shots to estimate the real (imaginary) part of $\bra{\phi_0}\hat{P}_l\ket{\phi_k}$ by the Hadamard test. Equation~\eqref{eq:variance_H} is obtained by setting $m_{kl}^{(\rm r)}=m_{kl}^{(\rm i)}=m_{kl}/2$ for $k>0$ and $m_{0l}^{(\rm r)}=m_{0l}$ to minimize the variance of each $\bra{\phi_0}\hat{P}_l\ket{\phi_k}$. Note that the diagonal element $\bra{\phi_0}\hat{P}_l\ket{\phi_0}$ has only a real part due to Hermiticity, and hence its imaginary part is determined to be zero without sampling (i.e., $m_{0l}^{(\rm i)}=0$). 
The variance~\eqref{eq:variance_H} can be further reduced by optimizing the shot allocation $\vec{m}_k$ under the constraint $\sum_{l=1}^{L}m_{kl}=m_k$ based on the Lagrange multiplier method. 
The resulting optimized shot allocation and sampling variance are given by
\begin{align}
    m_{kl}^{(\rm opt)}&=\frac{|c_l|}{\lambda}m_k , \label{eq:shot_LCU} \\
    \mathrm{Var}[\tilde{\mathbf{H}}_{0k},\vec{m}_k^{(\rm opt)}] &= \frac{(2-\delta_{0k})\lambda^2}{m_k}\left(2-\frac{1}{d}\right), \label{eq:variance_H_opt}
\end{align}
where $\lambda$ is the 1-norm of the Pauli-LCU Hamiltonian~\eqref{eq:LCU}. 

The above provides the sampling variance of each matrix element. To translate these element-wise variances into a bound on the perturbation of the entire matrix, $\norm{\mathbf{\Delta}_{\mathbf{Z}}}$, we invoke the following lemma from random matrix theory.
\begin{lem}[Norm behavior of Gaussian matrix series~\cite{tropp2015introduction,lee2024sampling}]\label{lemma:random}
    Let $\mathbf{\Delta}_{\mathbf{Z}}$ be a random matrix series with Gaussian coefficients, defined as
    \begin{align}
        \mathbf{\Delta}_{\mathbf{Z}} := \sum_k X_k\mathbf{A}_k,
    \end{align}
    where $\{\mathbf{A}_k\}$ is a finite set of $n\times n$ Hermitian matrices, $\{X_k\}$ is an independent finite set of normal random variables. Each $X_k$ follows the normal distribution $\mathcal{N}(0,\sigma_k^2)$ with zero mean and the variance $\sigma_k^2$. Let $v(\mathbf{\Delta}_{\mathbf{Z}})$ be the matrix variance statistics for $\mathbf{\Delta}_{\mathbf{Z}}$ defined as
    \begin{align}
        v(\mathbf{\Delta}_{\mathbf{Z}}) := \norm{\mathbb{E}[\mathbf{\Delta}_{\mathbf{Z}}^2]} = \left\|\sum_k\sigma_k^2\mathbf{A}_k^2\right\| . \label{eq:matrix_variance_statistics}
    \end{align}
    Then, the expectation value of the norm $\norm{\mathbf{\Delta}_{\mathbf{Z}}}$ is bounded as
    \begin{align}
        \mathbb{E}\left[ \norm{\mathbf{\Delta}_{\mathbf{Z}}} \right] \leq \sqrt{2v(\mathbf{\Delta}_{\mathbf{Z}})\log{(2n)}}. \label{eq:random_matrix_theory}
    \end{align}
\end{lem}

Lemma~\ref{lemma:random} was originally derived in Ref.~\cite{tropp2015introduction}, and subsequently adapted for the sampling error analysis of KQD in Ref.~\cite{lee2024sampling}. Specifically, the latter reference derived the following result for the matrix variance statistics $v(\mathbf{\Delta}_{\mathbf{Z}})$. 

\begin{lem}[Matrix variance statistics for Toeplitz-Hermitian matrix~\cite{lee2024sampling}]\label{lemma:toeplitz}
    Let $\mathbf{\Delta}_\mathbf{Z}$ be an $n\times n$ Gaussian random matrix with Toeplitz-Hermitian structure, defined as
    \begin{align}
        \mathbf{\Delta}_\mathbf{Z} &= X_0\mathbf{I}_n + \sum_{k=1}^{n-1}\left[ (X_k+iY_k)\mathbf{U}_n^k + (X_k-iY_k)\mathbf{L}_n^k \right]
    \end{align}
    where $\{X_k\}_{i=0}^{n-1}$ and $\{Y_k\}_{i=1}^{n-1}$ are Gaussian random variables for the real and imaginary parts of the matrix element, respectively. The diagonal elements are assumed to be real to maintain Hermiticity. Each $X_k$ and $Y_k$ follow the normal distribution $\mathcal{N}(0,\sigma_k^2)$. $\mathbf{I}_n$ is the $n\times n$ identity matrix, and $\mathbf{U}_n$ and $\mathbf{L}_n$ are the $n\times n$ upper and lower shift matrices defined as
    \begin{align}
        \mathbf{U}_{ij} = \delta_{i+1,j}, \quad
        \mathbf{L}_{ij} = \delta_{i,j+1}.
    \end{align}
    Consequently, $\mathbf{U}_n^k$ and $\mathbf{L}_n^k$ are the $n\times n$ matrices with only the $k$-th superdiagonal and subdiagonal elements, respectively, being 1 and all other elements being 0. 
    Then, the matrix variance statistics $v(\mathbf{\Delta}_{\mathbf{Z}})$ in Eq.~\eqref{eq:matrix_variance_statistics} can be rewritten as
    \begin{align}
        v(\mathbf{\Delta}_{\mathbf{Z}}) &= \norm{\mathrm{diag}{(v_1, \cdots, v_n)}} = \max_{l\in\{1,\cdots,\lfloor n/2\rfloor\}}v_l, \label{eq:v_general}
    \end{align}
    where 
    \begin{align}
        v_l &:= 
        \begin{dcases}
            \sigma_0^2 + \sum_{k=1}^{l-1}2\sigma_k^2 + \sum_{k=l}^{n-l}\sigma_k^2, & l\leq \lfloor n/2\rfloor \\
            v_{n-l}, & \mathrm{otherwise}
        \end{dcases}.
    \end{align}
\end{lem}

The proof of Lemma~\ref{lemma:toeplitz} is presented in Ref.~\cite{lee2024sampling}. 
Applying Lemma~\ref{lemma:toeplitz} to the matrices $\mathbf{Z}=\mathbf{H},\mathbf{S}$ measured in KQD and MSD, with sampling variances given by Eqs.~\eqref{eq:variance_S_opt}, \eqref{eq:variance_H_opt}, and~\eqref{eq:variance_H_opt_diff}, we obtain 
\begin{align}
    v_l(\vec{m}) &= 2V_{\mathbf{Z}}^2\left(\frac{\delta_{\mathbf{Z}\mathbf{H}}}{m_0}+\sum_{k=1}^{l-1}\frac{4}{m_k}+\sum_{k=l}^{n-l}\frac{2}{m_k}\right), \label{eq:vl_sampling}
\end{align}
where $V_{\mathbf{S}} := 1$ and
\begin{align}
    V_\mathbf{H} &:= 
    \begin{dcases}
        \lambda, & \mathrm{KQD} \\
        \norm{\vec{a}^{(J;1)}}_1/\delta_t, & \mathrm{MSD}
    \end{dcases}
\end{align}
$\vec{m}=(m_0,m_1,\cdots,m_{n-1})$ denotes the shot allocation for each matrix elements. 
$\delta_{\mathbf{Z}\mathbf{H}}$ is defined as 
\begin{align}
    \delta_{\mathbf{Z}\mathbf{H}} := 
    \begin{cases}
        1, & \mathbf{Z} = \mathbf{H} \\
        0, & \mathbf{Z} = \mathbf{S}
    \end{cases}. 
\end{align}
The shot allocation $\vec{m}$ is optimized to minimize the matrix variance statistics $v(\mathbf{\Delta}_\mathbf{Z})$ under the constraint of $\sum_km_k=M$~\cite{lee2024sampling}. Specifically, we formulate the optimization problem using the following Lagrangian: 
\begin{align}
    \mathcal{L}_l(\vec{m},\xi) = v_l(\vec{m}) + \xi\left(\sum_{k=0}^{n-1}m_k-M\right),
\end{align}
where $\xi$ is a constant. The Lagrangian $\mathcal{L}_l(\vec{m},\xi)$ is a convex function and minimized at a shot allocation $\vec{m}^{(l)}$ given by
\begin{align}
    m_k^{(l)} &= 
    \begin{dcases}
        m_0^{(l)}, & k=0 \\
        2m_0^{(l)}, & 1 \leq k \leq l-1 \\
        \sqrt{2}m_0^{(l)}, & l \leq k \leq n-l \\
        0, & \mathrm{otherwise}
    \end{dcases} \label{eq:mk_opt}
\end{align}
where 
\begin{align}
    m_0^{(l)} := \frac{M}{\sqrt{2}n-2(\sqrt{2}-1)l-2+\sqrt{2}+\delta_{\mathbf{Z}\mathbf{H}}}. \label{eq:m0_opt}
\end{align}
Inserting Eq.~\eqref{eq:mk_opt} into Eq.~\eqref{eq:vl_sampling}, we obtain
\begin{align}
    v_l(\vec{m}^{(l)}) &= \frac{2V_{\mathbf{Z}}^2}{M}\left(\sqrt{2}n-2(\sqrt{2}-1)l-2+\sqrt{2}+\delta_{\mathbf{Z}\mathbf{H}}\right)^{2}, \label{eq:vl_sampling_opt}
\end{align}
for $l\leq \lfloor n/2 \rfloor$. 
Equations.~\eqref{eq:vl_sampling} and~\eqref{eq:mk_opt} indicate that 
\begin{align}
    \max_{j\in\{1,\cdots,\lfloor n/2\rfloor\}}v_j(\vec{m}^{(l)}) = v_1(\vec{m}^{(l)}). 
\end{align}
Furthermore, Eqs.~\eqref{eq:mk_opt} and~\eqref{eq:vl_sampling_opt} lead to
\begin{align}
    v_1(\vec{m}^{(l)})  = v_1(\vec{m}^{(1)}). 
\end{align}
Therefore, from Eq.~\eqref{eq:v_general}, the optimal value of $v(\mathbf{\Delta}_\mathbf{Z})$ is given by
\begin{align}
    v_{\rm opt}(\mathbf{\Delta}_\mathbf{Z}) &= \min_{\vec{m}} \max_{l\in\{1,\cdots,\lfloor n/2\rfloor\}}v_l(\vec{m}) =v_1(\vec{m}^{(1)}) = \frac{2V_{\mathbf{Z}}^2}{M}\left(\sqrt{2}(n-1)+\delta_{\mathbf{Z}\mathbf{H}}\right)^{2},
\end{align}
and the corresponding optimal shot allocation is $\vec{m}^{(1)}$, which is described as Eq.~\eqref{eq:kqd_shot_allocation} in the main text.

Combining this result with Lemma~\ref{lemma:random}, the sampling perturbation for KQD and MSD under the optimal shot allocation is obtained as 
\begin{align}
    \mathbb{E}\left[ \norm{\mathbf{\Delta}_{\mathbf{Z}}} \right] &\leq \sqrt{2v_{\rm opt}(\mathbf{\Delta}_{\mathbf{Z}})\log{(2n)}} \nonumber\\
    &= \frac{2V_{\mathbf{Z}}\sqrt{\log{(2n)}}}{\sqrt{M}}\left(\sqrt{2}(n-1)+\delta_{\mathbf{Z}\mathbf{H}}\right) \nonumber\\
    & \simeq \frac{2nV_{\mathbf{Z}}\sqrt{2\log{(2n)}}}{\sqrt{M}}. 
\end{align}
In Ref.~\cite{lee2024sampling}, it was numerically clarified that the probability of finding the norm $\norm{\mathbf{\Delta}_{\mathbf{Z}}}$ larger than the upper bound of its expectation value $\mathbb{E}\left[ \norm{\mathbf{\Delta}_{\mathbf{Z}}}\right]$ is significantly small in KQD. Such behavior is also confirmed for MSD in our numerical simulation (Fig.~\ref{fig:matrix_error} in the main text). Consequently, the sampling perturbation $\norm{\mathbf{\Delta}_{\mathbf{Z}}}$ can be estimated by its expectation value, which serves as a practical upper bound: 
\begin{align}
    \norm{\mathbf{\Delta}_{\mathbf{Z}}} \lesssim \mathbb{E}\left[ \norm{\mathbf{\Delta}_{\mathbf{Z}}} \right] \simeq \frac{2nV_{\mathbf{Z}}\sqrt{2\log{(2n)}}}{\sqrt{M}}. \label{eq:sampling_error_general}
\end{align}
Equation~\eqref{eq:sampling_error_general} is used to obtain the sampling error bound in Eqs.~\eqref{eq:sampling_perturbation_S}, \eqref{eq:sampling_perturbation_H}, and \eqref{eq:sampling_perturbation_H_diff}. 

\section{Finite-difference error analysis}\label{si:fd_error}
In this section, we analyze the finite difference error for MSD by utilizing the following lemma concerning the degree-$J$ central finite difference formula. 
\begin{lem}\label{lemma:fd_error}
    Let $\delta\in\mathbb{R}_+$ and $J,q\in\mathbb{N}$. Suppose $f:\mathbb{R}\to\mathbb{R}$ is an $(s+1)$-times differentiable function, where $s:=q-1+2(J+1-\lfloor\frac{q+1}{2}\rfloor)$. Then, 
    \begin{align}
        &\abs{\left.\frac{d^q}{dx^q}f(x)\right|_{x=0} - \frac{1}{\delta^q}\sum_{j=-J}^{J}a_j^{(J;q)}f(j\delta)} \leq \sum_{j=-J}^{J}\abs{ a_j^{(J;q)}j^{s+1} }\frac{\norm{f^{(s+1)}}_{\infty}}{(s+1)!}\delta^{s-q+1}
    \end{align}
    where $\norm{f^{(s+1)}}_{\infty}:=\sup_{\zeta\in[-J\delta,J\delta]}\abs{f^{(s+1)}(\zeta)}$ and $a_j^{(J;q)}$ is defined as Eq.~\eqref{eq:cfd_coeff_general}. 
\end{lem}
\begin{proof}
    The proof is a straightforward extension of the analysis in Ref.~\cite{gilyen2019optimizing}, which derives the upper bound for the first-order derivative. Taylor's theorem with Lagrange remainder term leads to the following expression of $f$:
    \begin{align}
        f(x) = \sum_{j=0}^{s}\frac{f^{(j)}(0)}{j!}x^j + \frac{f^{(s+1)}(\zeta)}{(s+1)!}x^{s+1}, \label{eq:taylor_expansion}
    \end{align}
    where $x\in\mathbb{R}$ and $\zeta\in[0,x]$. For a rescaled variable $z:=x/\delta$, Eq.~\eqref{eq:taylor_expansion} can be rewritten as
    \begin{align}
        P_s(z) &:= \sum_{j=0}^{s}\frac{f^{(j)}(0)}{j!}(z\delta)^j = f(z\delta) - \frac{f^{(s+1)}(\zeta)}{(s+1)!}(z\delta)^{s+1}. \label{eq:taylor_expansion_scaled}
    \end{align}
    Here, the Lagrange interpolation formula allows us to rewrite the degree-$s$ polynomial $P_s(z)$ as 
    \begin{align}
        P_s(z) &= \sum_{j=-J}^{J}P_s(j)L_j(z) , \label{eq:lagrange_interpolation} \\
        L_{j}(z) &:= \prod_{\substack{r=-J\\ r\neq j}}^{J}\frac{z-r}{j-r} \label{eq:lagrange_interpolation_function} . 
    \end{align}
    From Eqs.~\eqref{eq:taylor_expansion_scaled} and~\eqref{eq:lagrange_interpolation},  we obtain
    \begin{align}
        &P_s^{(q)}(0) = f^{(q)}(0)\delta^q = \sum_{j=-J}^{J}\left[ f(j\delta) - \frac{f^{(s+1)}(\zeta)}{(s+1)!}(j\delta)^{s+1} \right]L_j^{(q)}(0) = \sum_{j=-J}^{J}a_j^{(J;q)}\left[ f(j\delta) - \frac{f^{(s+1)}(\zeta)}{(s+1)!}(j\delta)^{s+1} \right], \label{eq:deriv_legendre}
    \end{align}
    where Eq.~\eqref{eq:cfd_coeff_general} was used in the last line. 
    Applying the triangle inequality to Eq.~\eqref{eq:deriv_legendre}, the upper bound of the finite difference error is estimated as follows:
    \begin{align}
        &\abs{f^{(q)}(0) -\frac{1}{\delta^q}\sum_{j=-J}^{J}a_j^{(J;q)}f(j\delta)} = \abs{\frac{1}{\delta^q}\sum_{j=-J}^{J}a_j^{(J;q)}\frac{f^{(s+1)}(\zeta)}{(s+1)!}(j\delta)^{s+1}} \leq \abs{\sum_{j=-J}^{J}a_j^{(J;q)}j^{s+1}}\frac{\norm{f^{(s+1)}}_{\infty}}{(s+1)!}\delta^{s-q+1}.
    \end{align}
\end{proof}
Applying Lemma~\ref{lemma:fd_error} to a matrix function $f(\hat{H})=e^{-i\hat{H}t}$, we obtain the upper bound of the finite difference matrix perturbation as shown in Eqs.~\eqref{eq:H_mat_fd_error} and~\eqref{eq:msd_power_perturbation}. 

\section{Matrix perturbation of the projected Hamiltonian power matrix}
\label{si:moment_perturbation}
Following the same steps as Sec. III and \ref{si:fd_error}, we bound the matrix perturbation of the projected Hamiltonian power matrix $\mathbf{M}^{(q)}$. 
Considering finite sampling and finite difference approximation as a source of the matrix perturbation, the corresponding matrix perturbation, $\mathbf{\Delta}_{\mathbf{M}}^{(q)}:=\tilde{\mathbf{M}}^{(q)}-\mathbf{M}^{(q)}$ with $\tilde{\mathbf{M}}^{(q)}$ being the perturbed matrix, is upper bounded as 
\begin{align}
    \norm{\mathbf{\Delta}_{\mathbf{M}}^{(q)}} &\lesssim \frac{2n\sqrt{2v(\mathbf{\Delta}_{\mathbf{M}}^{(q)})\log{(2n)}}}{\delta_t^q \sqrt{M}} + n\sum_{j=-J}^{J}\abs{a_j^{(J;q)}j^{s+1}}\frac{\norm{\hat{H}}_{\mathcal{K}}^{s+1}}{(s+1)!} \delta_t^{s-q+1}, \label{eq:msd_power_perturbation}
\end{align}
where
\begin{align}
    v(\mathbf{\Delta}_{\mathbf{M}}^{(q)}) &:= |a_0^{(J;q)}|^2 + 2\norm{\vec{a}_j^{(J;1)}}_1\sum_{j=1}^{J}\frac{|a_j^{(J;q)}|^2}{|a_j^{(J;1)}|}. 
\end{align} 
Note that Eq.~\eqref{eq:msd_power_perturbation} is derived under the assumption that the matrices $\tilde{\mathbf{H}}$ and $\tilde{\mathbf{S}}$ are estimated with the same number of shots, $M$. 
Specifically, choosing the optimal shift parameter $\delta_t^{(\rm opt)}$ as in Eq.~\eqref{eq:msd_optimal_delta}, the total matrix perturbation is bounded as
\begin{align}
    &\norm{\mathbf{\Delta}_{\mathbf{M}}^{(q)}} \lesssim \order{\frac{n\norm{\hat{H}}_{\mathcal{K}}^{q}}{\sqrt{M}} \left[ \sqrt{v(\mathbf{\Delta}_{\mathbf{M}}^{(q)})} +\frac{\sum_{j=1}^{J}\abs{a_j^{(J;q)}j^{s+1}}}{(s+1)!}\right] }.\label{eq:msd_power_perturbation_optimal}
\end{align}
The right-hand side of Eq.~\eqref{eq:msd_power_perturbation_optimal} suggests that the matrix perturbation increases with the order $q$ mainly due to the factor $\norm{\hat{H}}_{\mathcal{K}}^{q}$. Therefore, the matrix perturbation increases with the order $q$. 

\section{Moment-Lanczos construction from Hamiltonian moments}
\label{si:lanczos}
Here, we detail the construction of the tridiagonal Lanczos matrix used in the moment-Lanczos error mitigation of the main text.
Using the Hamiltonian moments, we can perform error mitigation on the ground state energy estimate based on the Lanczos scheme~\cite{witte1994plaquette,haxton2005piecewise}. Specifically, we construct the following tridiagonal matrix~\cite{lanczos1950iteration}
\begin{align}
    \mathbf{T}_j = 
    \begin{pmatrix}
        \alpha_1 & \beta_1 & 0 \\
        \beta_1 & \alpha_2 & \beta_2 & \\
          & \beta_2 & \alpha_3 & \ddots \\
         &  & \ddots & \ddots & \beta_{j-2} \\
         & & & \beta_{j-2} & \alpha_{j-1} & \beta_{j-1} \\
         & & & &  \beta_{j-1} & \alpha_j
    \end{pmatrix}. \label{eq:lanczos_tridiagonal}
\end{align}
The matrix elements $\alpha_j$ and $\beta_j$ are calculated recursively using the Hamiltonian moments $\{\mu_q\}_{q=1}^{2j}$ as
\begin{align}
    \alpha_j &= \mathcal{M}_{j-1}/\mathcal{L}_{j-1} - \mathcal{M}_{j-2}/\mathcal{L}_{j-2}, \\
    \beta_j^2 &= \mathcal{L}_j\mathcal{L}_{j-2}/\mathcal{L}_{j-1}^2, 
\end{align}
where the determinants $\mathcal{L}_j$ and $\mathcal{M}_j$ are defined as 
\begin{align}
    \mathcal{L}_j &= 
    \begin{vmatrix}
        1 & \mu_1 & \mu_2 & \cdots & \mu_j \\
        \mu_1 & \mu_2 & \mu_3 & \cdots & \mu_{j+1} \\
        & \vdots & & \ddots & \vdots \\
        \mu_j & \mu_{j+1} & \mu_{j+2} & \cdots & \mu_{2j}
    \end{vmatrix},
\end{align}
and 
\begin{align}
    \mathcal{M}_j &= 
    \begin{vmatrix}
        1 & \mu_1 & \mu_2 & \cdots & \mu_{j-1} & \mu_{j+1} \\
        \mu_1 & \mu_2 & \mu_3 & \cdots & \mu_{j} & \mu_{j+2} \\
        & \vdots & & \ddots & \vdots & \vdots \\
        \mu_j & \mu_{j+1} & \mu_{j+2} & \cdots & \mu_{2j-1} & \mu_{2j+1}
    \end{vmatrix},
\end{align}
with $\mathcal{L}_0:=1$, $\mathcal{M}_{-1}:=0$, and $\mathcal{M}_0:=\mu_1$. We perform iterative diagonalization of the tridiagonal matrix $\mathbf{T}_j$ and obtain the lowest eigenvalue for $j=1,2,\cdots,J$.

\section{1-norm reduction of the Hamiltonian}\label{si:1norm}
 This section describes the 1-norm reduction methods for the electronic structure Hamiltonian. The main text uses these methods for numerical simulation and sampling cost estimation of KQD. 
Specifically, we consider the electronic structure Hamiltonian in the second quantized representation:
\begin{align}
    \hat{H} &= \sum_{pq}^{N_{\rm orb}}\sum_{\sigma=\uparrow,\downarrow}h_{pq}\hat{a}_{p\sigma}^{\dag}\hat{a}_{q\sigma} + \frac{1}{2}\sum_{pqrs}^{N_{\rm orb}}\sum_{\sigma\tau}g_{pqrs}\hat{a}_{p\sigma}^{\dag}\hat{a}_{r\tau}^{\dag}\hat{a}_{s\tau}\hat{a}_{q\sigma}, \label{eq:electronic_hamiltonian}
\end{align}
where $N_{\rm orb}$ is the number of spatial orbitals, and $h_{pq}$ and $g_{pqrs}$ are the one- and two-electron integrals for spatial orbital indices $\{p,q,r,s\}$.  $\hat{a}_{p\sigma}$ and $\hat{a}_{p\sigma}^{\dag}$ denote the annihilation and creation operators for a $p$-orbital electron with spin $\sigma\in\{\uparrow,\downarrow\}$, respectively. 
Applying a fermion-to-qubit mapping to the fermionic operators, the electronic structure Hamiltonian~\eqref{eq:electronic_hamiltonian} can be translated into the Pauli-LCU form as in Eq.~\eqref{eq:LCU}. The corresponding 1-norm, defined as Eq.~\eqref{eq:1_norm}, can be generally expressed using $h_{pq}$ and $g_{pqrs}$ as follows~\cite{koridon2021orbital}: 
\begin{align}
    \lambda &= \sum_{pq}\abs{h_{pq}+\sum_{r}g_{pqrr}-\frac{1}{2}\sum_{r}g_{prrq}} + \frac{1}{2}\sum_{p>r,s>q}\abs{g_{pqrs}-g_{psrq}} + \frac{1}{4}\sum_{pqrs}\abs{g_{pqrs}}. \label{eq:1norm_ele}
\end{align}
The 1-norm reduction methods introduced below aim to minimize the 1-norm $\lambda$ given by Eq.~\eqref{eq:1norm_ele} by optimizing the one- and two-electron integrals $h_{pq}$ and $g_{pqrs}$ under the constraint of preserving the spectrum of the Hamiltonian.

\subsection{Orbital optimization}\label{si:oo}
One of the most well-established 1-norm reduction approaches is the orbital optimization~\cite{koridon2021orbital}. The orbital optimization is achieved through a unitary transformation of the spatial molecular orbitals: 
\begin{align}
    \mathbf{U}_{\rm OO} = e^{-\mathbf{K}}, \label{eq:U_oo}
\end{align}
where $\mathbf{K}$ is a skew-symmetric matrix (i.e., $\mathbf{K}^{\rm T}=-\mathbf{K}$) described as 
\begin{align}
    \mathbf{K} = 
    \begin{pmatrix}
        0 & K_{12} & K_{13} & \cdots & K_{1N_{\rm orb}} \\
        -K_{12} & 0 & K_{23} & \cdots & K_{2N_{\rm orb}} \\
        -K_{13} & -K_{23} & 0 & \cdots & K_{3N_{\rm orb}} \\
        \vdots & \vdots & \vdots & \ddots & \vdots \\
        -K_{1N_{\rm orb}} & -K_{2N_{\rm orb}} & -K_{3N_{\rm orb}} & \cdots & 0
    \end{pmatrix}.
\end{align}
For real molecular orbitals, the parameters ${K_{ij}}$ are assumed to be real. Therefore, the unitary matrix $\mathbf{U}_{\rm OO}$ contains $N_{\rm orb}(N_{\rm orb}-1)/2$ parameters in total. Under this unitary transformation, the one- and two-electron integrals are parameterized as
\begin{align}
    \tilde{h}_{pq}(\mathbf{K}) &= \sum_{ij}^{N_{\rm orb}}[e^{-\mathbf{K}}]_{ip}h_{ij}[e^{-\mathbf{K}}]_{jq}, \\
    \tilde{g}_{pqrs}(\mathbf{K}) &= \sum_{ijkl}^{N_{\rm orb}}[e^{-\mathbf{K}}]_{ip}[e^{-\mathbf{K}}]_{jq}g_{ijkl}[e^{-\mathbf{K}}]_{kr}[e^{-\mathbf{K}}]_{ls}.
\end{align}
Then, the parameters $\{K_{ij}\}$ are optimized to minimize the Pauli 1-norm in Eq.~\eqref{eq:1norm_ele} for the modified one- and two-electron integrals $\tilde{h}_{pq}(\mathbf{K})$ and $\tilde{g}_{pqrs}(\mathbf{K})$. 
Since the orbital optimization is a unitary transformation of the Hamiltonian, the spectrum remains invariant. 

We also mention that orbital optimization can be used to estimate the spectral range~\cite{cortes2024spectral}, which dictates the lower bound of the 1-norm~\cite{loaiza2023lcu}. In this case, we perform orbital optimization to minimize the following cost function~\cite{cortes2024spectral}:
\begin{align}
    E_{\rm HF}(\mathbf{K}) &= \sum_{p}^{N_\alpha}\tilde{h}_{pp}(\mathbf{K}) +  \sum_{p}^{N_\beta}\tilde{h}_{pp}(\mathbf{K}) + \frac{1}{2} \sum_{p,q}^{N_\alpha,N_\alpha}\tilde{g}_{ppqq}(\mathbf{K}) + \frac{1}{2}\sum_{p,q}^{N_\beta,N_\beta}\tilde{g}_{ppqq}(\mathbf{K}) \nonumber\\
    &- \frac{1}{2}\sum_{p,q}^{N_\alpha,N_\alpha}\tilde{g}_{pqqp}(\mathbf{K}) - \frac{1}{2}\sum_{p,q}^{N_\beta,N_\beta}\tilde{g}_{pqqp}(\mathbf{K}) + \sum_{p,q}^{N_\alpha,N_\beta}\tilde{g}_{ppqq}(\mathbf{K}) , \label{eq:hf_level_energy}
\end{align}
where $N_\alpha$ and $N_\beta$ denote the number of spin-up and spin-down electrons, respectively. Then, the Hartree-Fock level approximation of the minimum and maximum eigenvalues in the $\bm{Q}$-symmetry sector is obtained as 
\begin{align}
    E_{\rm min}^{(\mathrm {HF}, \bm{Q})} &:= \min_{\mathbf{K}}E_{\rm HF}(\mathbf{K}), \\
     E_{\rm max}^{(\mathrm {HF}, \bm{Q})} &:= \max_{\mathbf{K}}E_{\rm HF}(\mathbf{K}) = -\min_{\mathbf{K}}\bar{E}_{\rm HF}(\mathbf{K}), 
\end{align}
where $\bar{E}_{\rm HF}(\mathbf{K})$ is defined by taking the negative of the one- and two-electron integrals in Eq.~\eqref{eq:hf_level_energy}. Using $E_{\rm min}^{(\mathrm {HF}, \bm{Q})}$ and $E_{\rm max}^{(\mathrm {HF}, \bm{Q})}$, we obtain the Hartree-Fock level approximation of the spectral range as $\Delta{E}_{\bm{Q}}\simeq E_{\rm max}^{(\mathrm {HF}, \bm{Q})} - E_{\rm min}^{(\mathrm {HF}, \bm{Q})}$, which is generally slightly smaller than the exact spectral range, as discussed in the main text.

\subsection{Block-invariant symmetry shift}
Another promising approach to reduce the 1-norm of the electronic structure Hamiltonian is the block-invariant symmetry shift (BLISS)~\cite{loaiza2023bliss,patel2024blisslp}, which constructs a shifted Hamiltonian $\hat{H}_T:=\hat{H}-\hat{T}$ while the spectrum of subspaces of interest remains invariant. BLISS is implemented by a shift operator $\hat{T}$ satisfying $\hat{T}\ket{\psi_i}=0$ for any Hamiltonian eigenstates of interest $\{\ket{\psi_i}\}$. Specifically, the electronic structure Hamiltonian~\eqref{eq:electronic_hamiltonian} preserves the particle number symmetry, and thus it is sufficient to maintain the spectrum of a fixed electron number sector with $N_e$ electrons. Considering this particle number symmetry, the shift operator is described as 
\begin{align}
    \hat{T}(\vec{\mu},\vec{\xi}) &= \mu_1(\hat{N}-N_e)+\mu_2(\hat{N}^2-N_e^2) +\sum_{pq}^{N_{\rm orb}}\xi_{pq}\hat{F}_{pq}(\hat{N}-N_e), \label{eq:bliss_operator}
\end{align}
where $\hat{N}:=\sum_{p\sigma}\hat{a}_{p\sigma}^{\dag}\hat{a}_{p\sigma}$ is the total number operator, $\hat{F}_{pq}:=\sum_{pq,\sigma}\hat{a}_{p\sigma}^{\dag}\hat{a}_{q\sigma}$ are spatial excitation operators, and $\xi_{pq}=\xi_{qp}$ are real symmetric parameters. The shift operator $\hat{T}(\vec{\mu},\vec{\xi})$ is parameterized by $2+N_{\rm orb}(N_{\rm orb}+1)/2$ real parameters, namely $\vec{\mu}:=(\mu_1,\mu_2)$ and $\vec{\xi}:=\{\xi_{pq}\}$. This shift operator satisfies $\hat{T}(\vec{\mu},\vec{\xi})\ket{\psi}=0$ for any state $\ket{\psi}$ with $N_e$ electrons. 
The Hamiltonian is modified by the above BLISS operation as follows: 
\begin{align}
    \hat{H}_T(\vec{\mu},\vec{\xi}) &:= \hat{H}-\hat{T}(\vec{\mu},\vec{\xi}), \nonumber\\
    &= \sum_{pq}^{N_{\rm orb}}\sum_{\sigma=\uparrow,\downarrow}\tilde{h}_{pq}(\vec{\mu},\vec{\xi})\hat{a}_{p\sigma}^{\dag}\hat{a}_{q\sigma} + \frac{1}{2}\sum_{pqrs}^{N_{\rm orb}}\sum_{\sigma\tau}\tilde{g}_{pqrs}(\vec{\mu},\vec{\xi})\hat{a}_{p\sigma}^{\dag}\hat{a}_{r\tau}^{\dag}\hat{a}_{s\tau}\hat{a}_{q\sigma} + \mu_1N_e + \mu_2N_e^2,
\end{align}
where the shifted one- and two-electron integrals are defined as
\begin{align}
   \tilde{h}_{pq}(\vec{\mu},\vec{\xi}) &= h_{pq} - (\mu_1+\mu_2)\delta_{pq} + (N_e-1)\xi_{pq}, \\
    \tilde{g}_{pqrs}(\vec{\mu},\vec{\xi}) &= g_{pqrs} - 2\mu_2\delta_{pq}\delta_{rs} - (\xi_{pq}\delta_{rs}+\delta_{pq}\xi_{rs}).  
\end{align}
The parameters $\vec{\mu}$ and $\vec{\xi}$ are optimized to minimize the 1-norm~\eqref{eq:1norm_ele} of the shifted one- and two-electron integrals, which is formally given by
\begin{align}
    &\lambda_{\rm BLISS}(\vec{\mu},\vec{\xi}) \nonumber\\
    &= \sum_{pq}\abs{\tilde{h}_{pq}(\vec{\mu},\vec{\xi})+\sum_{r}\tilde{g}_{pqrr}(\vec{\mu},\vec{\xi})-\frac{1}{2}\sum_{r}\tilde{g}_{prrq}(\vec{\mu},\vec{\xi})} \nonumber\\
    &+ \frac{1}{2}\sum_{p>r,s>q}\abs{\tilde{g}_{pqrs}(\vec{\mu},\vec{\xi})-\tilde{g}_{psrq}(\vec{\mu},\vec{\xi})} + \frac{1}{4}\sum_{pqrs}\abs{\tilde{g}_{pqrs}(\vec{\mu},\vec{\xi})}.
\end{align}
To accelerate the numerical optimization, we can utilize the analytic gradients for the cost function $\lambda_{\rm BLISS}(\vec{\mu},\vec{\xi})$ obtained as follows: 
\begin{align}
    \frac{\partial\lambda_{\rm BLISS}}{\partial\mu_1} &= -\sum_{p}\sgn{\tilde{h}_{pp} + \sum_{r}\tilde{g}_{pprr} - \frac{1}{2}\sum_{r}\tilde{g}_{prrp} }, \label{eq:bliss_gradient1} \\
    \frac{\partial\lambda_{\rm BLISS}}{\partial\mu_2} 
    &= -\left(N_{\rm orb}+\frac{1}{2}\right)\sum_{p}\sgn{\tilde{h}_{pp} + \sum_{r}\tilde{g}_{pprr} - \frac{1}{2}\sum_{r}\tilde{g}_{prrp}} \nonumber\\
    &+ \frac{1}{2}\sum_{p>r}\sgn{\tilde{g}_{prrp} - \tilde{g}_{pprr} } - \frac{1}{4}\sum_{pr}\sgn{\tilde{g}_{pprr} }, \label{eq:bliss_gradient2} \\
    \frac{\partial\lambda_{\rm BLISS}}{\partial\xi_{pq}} 
    &= (N_e-N_{\rm orb})\sgn{\tilde{h}_{pq} + \sum_{r}\tilde{g}_{pqrr} - \frac{1}{2}\sum_{r}\tilde{g}_{prrq}} \nonumber \\
    &+ \frac{1}{2}\sum_{r\neq p,q}\sgn{\tilde{g}_{prrq} - \tilde{g}_{pqrr}} - \frac{1}{2}\sum_{r}\sgn{ \tilde{g}_{rrpq}}, \label{eq:bliss_gradient3}
\end{align}
where $p<q$ in Eq.~\eqref{eq:bliss_gradient3}.

\section{Details of the molecular Hamiltonians}
\label{si:mol_details}
This section provides detailed information regarding the molecular Hamiltonians employed in the numerical results in the main text. 

Table~\ref{tab:molecule_info} summarizes the information for each molecular Hamiltonian used in the main-text numerical results and sampling-cost estimations. The 1-norm $\lambda$ was obtained for the qubit representation of the molecular Hamiltonians, which were derived using the Jordan-Wigner transformation~\cite{Jordan1928}. Values of $\lambda$ in parentheses represent the 1-norm for the optimized Hamiltonians obtained by applying orbital optimization and BLISS, as described in Sec. VII. We empirically found that applying BLISS after orbital optimization further reduces the 1-norm. More specifically, we numerically performed orbital optimization and BLISS using the sequential least squares programming (SLSQP) optimizer implemented in the \texttt{SciPy} package~\cite{2020SciPy-NMeth}. We chose the SLSQP optimizer due to the presence of an absolute value function in the 1-norm cost function (Eq.~\eqref{eq:1norm_ele}), which introduces discontinuities in its gradient. Indeed, we empirically found that other optimizers, such as the L-BFGS-B method, struggle to minimize the 1-norm, consistent with previous work~\cite{koridon2021orbital}. Furthermore, we leveraged the gradients evaluated using \texttt{JAX}~\cite{jax2018github} or analytical gradients in Eqs.~\eqref{eq:bliss_gradient1}, \eqref{eq:bliss_gradient2}, and \eqref{eq:bliss_gradient3} for the orbital optimization and BLISS methods, respectively, to accelerate the numerical optimization. 

The spectral range $\Delta{E}_{\bm{Q}}$ is obtained for a symmetry sector $\bm{Q}=(N_e,S)$ with the electron number $N_e$ and total spin $S=0$. For small molecular models, $\Delta{E}_{\bm{Q}}$ is calculated exactly by full configuration interaction. For large molecular models, exact diagonalization of the Hamiltonian is difficult due to the exponentially large number of configurations. Therefore, we adopted the orbital optimization method proposed in Ref.~\cite{cortes2024spectral} (see Sec.~\ref{si:oo}) for large molecular models, and estimated the spectral range $\Delta{E}_{\bm{Q}}$ using a Hartree-Fock level approximation. 

\begin{table*}[htbp]
\caption{\label{tab:molecule_info}%
Data of molecular Hamiltonians. $(N_{\rm orb}, N_e)$ represents the system size, with $N_{\rm orb}$ being the number of spatial orbitals for a given basis set and $N_e$ being the number of electrons under the frozen-core approximation. In columns for the 1-norm $\lambda$, values without parentheses are obtained for molecular Hamiltonians in the canonical molecular orbital basis, while values within parentheses are obtained for 1-norm reduced Hamiltonians obtained by applying the orbital optimization and BLISS described in Sec. VII. In columns for the spectral range $\Delta{E}_{(N_e,S)}$, values without parentheses are obtained by exact diagonalization of the Hamiltonian within the $(N_e,S)$ sector, while values within parentheses are obtained by the Hartree-Fock level approximation in Ref.~\cite{cortes2024spectral}. The last column shows the values of the energy shift error defined by Eq.~\eqref{eq:energy_shift_error}. 
}
\begin{ruledtabular}
\begin{tabular}{ l c c c c}
    System & $(N_{\rm orb}, N_e)$ & $\lambda$ & $\Delta{E}_{(N_e,S)}$ & $\delta_{\rm shift}$ (Eq.~\eqref{eq:energy_shift_error}) \\
    \colrule
    \ce{H2} [STO-3G] & (2,2) & 1.86 (1.07) & 1.62 (1.58) & $4.4\times10^{-7}$\% \\
    \ce{H2} [6-31G] & (4,2) & 11.5 (4.48) & 3.08 (3.01) & 1.7\%  \\
    \ce{H2} [cc-pVDZ] & (10,2) & 101.3 (31.1) & 7.32 (7.27) & 0.48\% \\
    \ce{H2} [cc-pVTZ] & (28,2) & 1117.1 (409.7) & 13.9 (13.8) &  0.24\% \\
    \ce{LiH} [STO-3G] & (6,4) & 12.3 (5.69) &  6.62 (6.52) & 0.69\% \\
    \ce{LiH} [6-31G] & (11,4) & 41.9 (15.2) & 9.48 (9.46) & 0.27\% \\
    \ce{LiH} [cc-pVDZ] & (19,4) & 178.4 (60.9) & 13.9 (13.8) & 3.0\% \\
    \ce{LiH} [cc-pVTZ] & (44,4) & 1657.9 (679.3) & (25.1) \\
    \ce{BeH2} [STO-3G] & (7,6) & 21.5 (11.1) & 13.0 (12.8) & 0.35\% \\
    \ce{BeH2} [6-31G] & (13,6) & 85.5 (27.9) & 18.3 (18.2) & 0.080\% \\
    \ce{BeH2} [cc-pVDZ] & (24,6) & 392.4 (131.0) & (27.2)  \\
    \ce{BeH2} [cc-pVTZ] & (58,6) & 3745.8 (1293.1) & (45.5) \\
    \ce{H2O} [STO-3G] & (6,8) & 27.7 (8.43) & 4.70 (4.54) & 0.045\%  \\
    \ce{H2O} [6-31G] & (12,8) & 113.1 (35.3) & 16.20 (16.0) & 0.059\% \\
    \ce{H2O} [cc-pVDZ] & (23,8) & 674.2 (206.2) & (30.9) \\
    \ce{H2O} [cc-pVTZ] & (57,8) & 7652.8 (2201.3) & (73.6) \\
    \ce{NH3} [STO-3G] & (7,8) & 28.5 (11.1) & 5.69 (5.48) & 0.059\% \\
    \ce{NH3} [6-31G] & (14,8) & 160.5 (46.8) & (14.0) \\
    \ce{NH3} [cc-pVDZ] & (28,8) & 1113.0 (316.5) & (27.5) \\
    \ce{NH3} [cc-pVTZ] & (71,8) & 15325.6 (3618.1) & (68.2)  \\
    \ce{N2} [STO-3G] & (8,10) & 35.6 (13.4) & 6.87 (6.59) & 0.017\% \\
    \ce{N2} [6-31G] & (16,10) & 184.9 (56.4) & (18.3) \\
    \ce{N2} [cc-pVDZ] & (26,10) & 684.1 (242.6) & (33.4) \\
    \ce{N2} [cc-pVTZ] &  (58,10) & 5602.1 (2312.8) & (89.9) \\
    \ce{Cr2} [STO-3G] &  (26,28) & 516.6 (261.5) & (73.8) \\
    \ce{Cr2} [6-31G] & (44,28) & 1442.4 (576.7) & (141.5) \\
    \ce{Cr2} [cc-pVDZ] &  (76,28) & 6106.3 (2423.4) & (174.0) \\
    \ce{C4H4N2} [STO-3G] &  (28,30) & 500.9 (177.6) & (44.3) \\
    \ce{C4H4N2} [6-31G] &  (56,30) & 2946.7 (717.5) & (91.7) \\
    \ce{C4H4N2} [cc-pVDZ] &  (98,30) & 14141.1 & (152.4) \\
\end{tabular}
\end{ruledtabular}
\end{table*}

Figure~\ref{fig:mol_scaling} depicts the system size scaling of the 1-norm $\lambda$ and spectral range $\Delta{E}_{(N_e,S)}$ for various molecules. It is illustrated that $\lambda$ and $\Delta{E}_{(N_e,S)}$ polynomially increase with increasing the number of spatial orbitals $N_{\rm orb}$ by enlarging the basis set. Indeed, the data shown in Table~\ref{tab:molecule_scaling} indicates that the dependence of $\lambda$ and $\Delta{E}_{(N_e,S)}$ on $N_{\rm orb}$ can be accurately reproduced by polynomial fitting. Specifically, the 1-norm $\lambda$ scales as $\order{N_{\rm orb}^{2.52}}$ by averaging the fitting coefficients for all benchmark molecules, which decreases to $\order{N_{\rm orb}^{2.34}}$ for the 1-norm reduced Hamiltonians. On the other hand, the spectral range exhibits a sublinear scaling of $\order{N_{\rm orb}^{0.93}}$ on average. The sublinear scaling of the spectral range was also reported in Ref.~\cite{cortes2024spectral}. 

\begin{figure*}[tbp]
  \centering
  \includegraphics[width=0.99\textwidth]{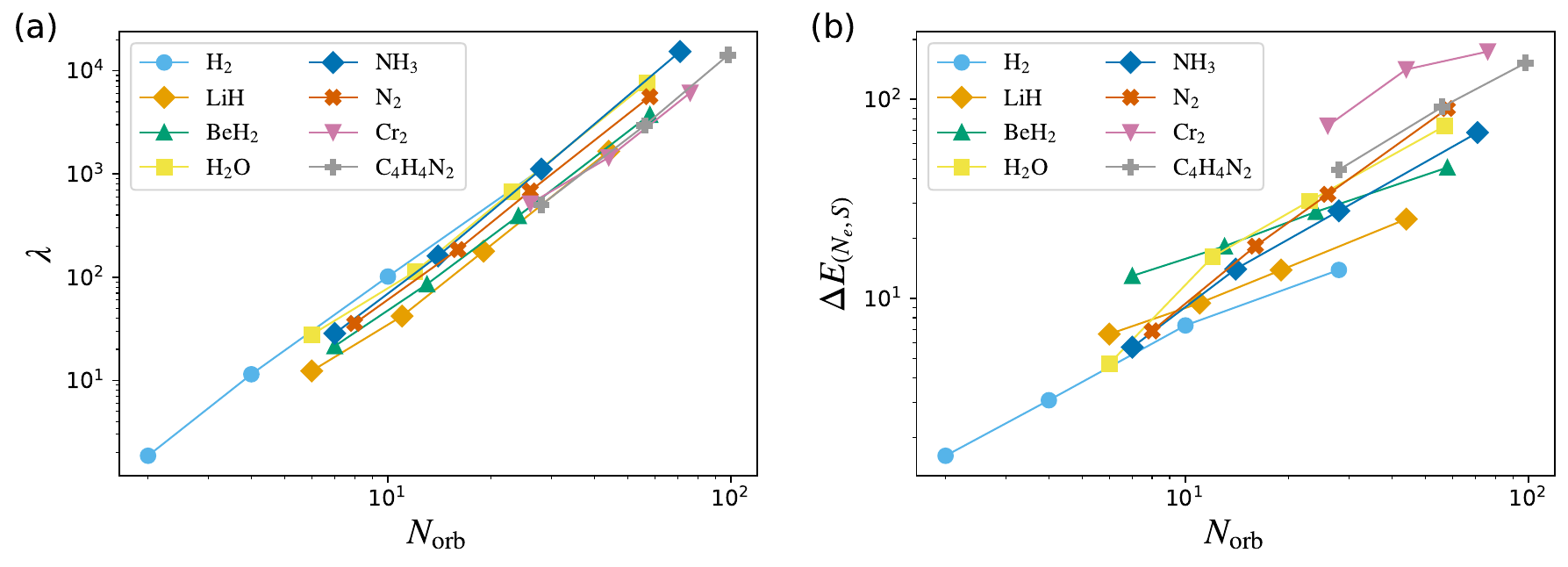}
  \caption{Scaling of the 1-norm $\lambda$ and the spectral range $\Delta{E}_{(N_e,S)}$ with system size for various molecules. Both quantities are shown as functions of the number of spatial orbitals $N_{\rm orb}$ on log--log axes, where $N_{\rm orb}$ is increased for each molecule by enlarging the basis set (STO-3G, 6-31G, cc-pVDZ, and cc-pVTZ). (a) The 1-norm $\lambda$ as a function of $N_{\rm orb}$, evaluated for the canonical molecular orbital basis. (b) The spectral range $\Delta{E}_{(N_e,S)}$ as a function of $N_{\rm orb}$. In both panels, the markers distinguish the eight benchmark molecules (\ce{H2}, \ce{LiH}, \ce{BeH2}, \ce{H2O}, \ce{NH3}, \ce{N2}, \ce{Cr2}, and \ce{C4H4N2}).}
  \label{fig:mol_scaling}
\end{figure*}

\begin{table*}[htbp]
\caption{\label{tab:molecule_scaling}%
Fitting coefficients of the 1-norm $\lambda$ and spectral range $\Delta{E}_{(N_e,S)}$ with respect to the number of spatial orbitals $N_{\rm orb}$ for different basis sets. Values are obtained by conducting a polynomial fitting assuming $\log_{10}\lambda=a_\lambda\log_{10}N_{\rm orb}+b_\lambda$ and $\log_{10}\Delta{E}_{\bm{Q}}=a_\Delta\log_{10}N_{\rm orb}+b_\Delta$ with the associated regression coefficients $R_\lambda^2$ and $R_\Delta^2$, respectively. Values in parentheses represent fitting data obtained for the 1-norm reduced Hamiltonians.  }
\begin{ruledtabular}
\begin{tabular}{ l c c c c c c }
    Molecule & $a_\lambda$ & $b_\lambda$ & $R_\lambda^2$ & $a_\Delta$ & $b_\Delta$ & $R_\Delta^2$ \\
    \colrule
    \ce{H2} & 0.37 (0.21) & 2.41 (2.25) & 0.9994 (0.9977) & 0.97 & 0.82 & 0.9904 \\
    \ce{LiH} & 0.13 (0.06) & 2.48 (2.43) & 0.9962 (0.9854) & 1.94 & 0.67 & 0.9985 \\
    \ce{BeH2} & 0.17 (0.10) & 2.45 (2.29) & 0.9991 (0.9878) & 4.02 & 0.60 & 0.9993 \\
    \ce{H2O} & 0.26 (0.08) & 2.53 (2.50) & 0.9960 (0.9971) & 0.67 & 1.19 & 0.9731 \\
    \ce{NH3} & 0.13 (0.07) & 2.73 (2.53) & 0.9992 (0.9939) & 0.78 & 1.06 & 0.9947 \\
    \ce{N2} & 0.16 (0.05) & 2.56 (2.63) & 0.9993 (0.9939) & 0.48 & 1.29 & 0.9987 \\
    \ce{Cr2} & 0.27 (0.27) & 2.30 (2.08) & 0.9926 (0.9763) & 5.95 & 0.80 & 0.9118 \\
    \ce{C4H4N2} & 0.07 (0.22) & 2.66 (2.01) & 0.9993 (1.0000) & 1.66 & 0.99 & 0.9983 \\
\end{tabular}
\end{ruledtabular}
\end{table*}

\section{Numerical results for other molecules}\label{si:other_mol}
As a further illustration of MSD and potential benefits across various molecular systems, we present numerical results for \ce{H2O} and \ce{NH3} molecules using the STO-3G basis set in Fig.~\ref{fig:msd_mols}. The figures correspond to Figs.~\ref{fig:matrix_error}, \ref{fig:thresholding}, and \ref{fig:energy_error} in the main text. 
In contrast to the \ce{H2} (cc-pVDZ) model studied in the main text, the \ce{H2O} (STO-3G) and \ce{NH3} (STO-3G) models exhibit smaller values of the ratio $\lambda/\Delta{E}_{(N_e,S)}$ (see Table~\ref{tab:molecule_info}). Consequently, the observed advantages of MSD over KQD are less pronounced in these systems compared to what is shown for the \ce{H2} (cc-pVDZ) model. Nevertheless, even with these smaller basis set models, MSD consistently demonstrates its superiority over KQD, as evidenced by the approximately 10-fold reduction in sampling costs needed to achieve chemical accuracy (Figs.~\ref{fig:msd_mols}(e) and~\ref{fig:msd_mols}(f)).

Based on the theoretical insights on the sampling-cost reduction over KQD presented in the main text, we anticipate that the benefits of MSD will become even more significant when applied to larger molecular models. In such scenarios, where classical simulation becomes prohibitively expensive or impossible, MSD is expected to offer a more substantial sampling cost reduction to achieve the chemical accuracy. 

\begin{figure*}[tbp]
  \centering
  \includegraphics[width=1.0\textwidth]{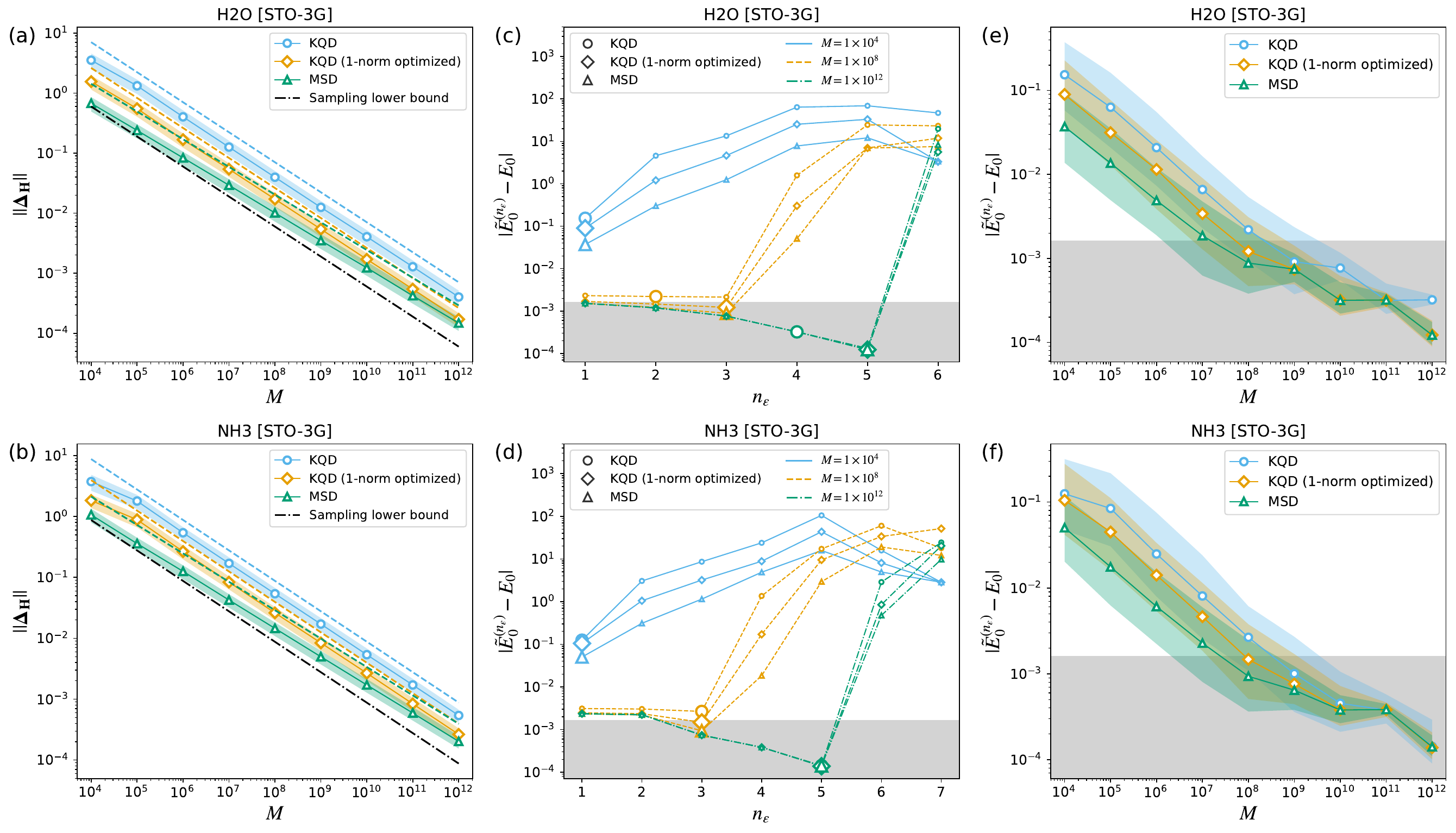}
  \caption{Numerical comparison of KQD (conventional and 1-norm optimized) and MSD for the \ce{H2O} (STO-3G) and \ce{NH3} (STO-3G) models. Similar to the numerical results in the main text, finite sampling simulations were conducted for 1,000 different random seeds. (a), (b) Dependence of matrix perturbation $\norm{\mathbf{\Delta}_\mathbf{H}}$ on total shot count $M$. The markers and shaded regions represent the ensemble average and standard deviation, respectively. The dashed and dash-dot lines indicate the theoretical upper bound (Eq.\eqref{eq:msd_matrix_perturbation}) and sampling lower bound (Eq.\eqref{eq:sampling_lower_bound}), respectively. (c), (d) Absolute difference between the exact ground state energy, $E_0$, and perturbed Krylov solutions, $\tilde{E}_0^{(n_\epsilon)}$, as a function of the remaining Krylov dimension after the thresholding, $n_\epsilon$. The optimal threshold points, determined as $\epsilon = \max{(\norm{\mathbf{\Delta}_{\mathbf{S}}}, \norm{\mathbf{\Delta}_{\mathbf{H}}}/\norm{\hat{{H}}})}$~\cite{kirby2024analysis}, are emphasized with large markers. (e), (f) Dependence of energy error $|\tilde{E}_0^{(n_\epsilon)}-E_0|$ at on total shot count $M$. The thresholding is performed for the optimal thresholding level. The markers and shaded regions represent the ensemble average and full width at half maximum (FWHM), respectively. The chemical accuracy (i.e., an absolute energy error below $1.6\times10^{-3}$ Ha) is highlighted with gray shading. }
  \label{fig:msd_mols}
\end{figure*}

\section{Performance on strongly multireference systems}\label{si:mrf}
A central question for the practical applicability of MSD is whether its sampling-cost advantage over KQD persists for strongly multireference systems, where the HF reference is a poor approximation to the ground state. In this section, we address this question using two complementary strongly correlated systems: the \ce{N2} molecule along its dissociation coordinate, which continuously tunes the degree of static correlation, and the \ce{Cr2} molecule at its equilibrium geometry, which is strongly multireference already at equilibrium owing to its complex electronic correlation~\cite{larsson2022chromium}. For both systems, we construct the active-space electronic-structure Hamiltonian using the complete-active-space self-consistent-field (CASSCF) method, and we specify the active space by $(N_{\rm orb}, N_{e})$ following the convention of Table~\ref{tab:molecule_info}, with $N_{\rm orb}$ spatial orbitals and $N_e$ active electrons.
The CASSCF orbitals are optimized in the presence of strong static correlation and therefore provide a more appropriate active space than fixed Hartree--Fock orbitals for multireference systems such as \ce{Cr2}. These active spaces are small enough to be solved exactly by classical CASSCF, which supplies the reference ground-state energy for benchmarking the sampling cost of MSD; the aim is to demonstrate the sampling-cost reduction of MSD for models that faithfully capture strong electron correlation, rather than to solve a classically intractable problem.

Before presenting the numerical results, we clarify why multireference character is not expected to degrade the relative performance of MSD. KQD and MSD construct the identical Krylov subspace shown in Eq.~\eqref{eq:krylov_space}, and differ only in the measurement protocol for the projected matrix elements $\mathbf{H}_{k'k}$. The expressibility of the Krylov subspace, and hence the Krylov dimension $n$ required for convergence, is therefore common to both methods. As shown in the main text, the sampling-cost ratio scales as $R[M] = \order{(\Delta{E}_{\bm{Q}}/\lambda)^2}$, which mainly depends on the Hamiltonian quantities $\Delta{E}_{\bm{Q}}$ and $\lambda$, not on the reference-state overlap $|\gamma_0|^2 =
|\braket{\Phi_{\rm HF}|\psi_0}|^2$. Here $\ket{\Phi_{\rm HF}}$ denotes the single closed-shell (Aufbau) reference determinant built from the CASSCF molecular orbitals, which we adopt as the initial state $\ket{\phi_0}$ throughout, and $\ket{\psi_0}$ is the exact ground state within the active space; thus $|\gamma_0|^2$ coincides with the squared overlap defined via Eq.~\eqref{eq:initial_state_decomposition} with $\ket{\phi_0} = \ket{\Phi_{\rm HF}}$. 
Strong correlation does increase the absolute sampling cost of both methods—through a reduced spectral gap $\Delta{E}_1$ and a smaller $|\gamma_0|^2$, which enlarge the required Krylov dimension $n$ via Eq.~\eqref{eq:kqd_energy_error}—but it does so equally for KQD and MSD, leaving $R[M]$ governed by $\Delta{E}_{\bm{Q}}/\lambda$. We now confirm this numerically.

We note that, whereas the main-text analysis adopts the central-difference degree $J = n$, here we set $J = N_{\rm orb}$ rather than $J = n$. Because the strongly correlated systems studied here require a large Krylov order $n$, retaining $J = n$ would unnecessarily increase the maximum evolution time [Eq.~\eqref{eq:t_max}]; fixing $J = N_{\rm orb}$ keeps the finite-difference error sufficiently small while limiting the Hamiltonian-simulation cost.

We emphasize that the large Krylov dimension required in the strongly correlated regime originates from the poor overlap between the ground state and the single HF reference used here for simplicity. This dimension could likely be reduced by employing a more suitable initial state or by constructing the Krylov subspace from multiple reference states, as explored in previous works on quantum Krylov methods for strongly correlated systems~\cite{stair2020multireference,dcunha2024fragment}. Since such improvements act on the shared Krylov subspace, they would benefit KQD and MSD equally and are orthogonal to the sampling-cost reduction that is the focus of the present work. A detailed investigation of multireference-adapted initial states within MSD is an interesting direction for future study.

\subsection{N$_2$ dissociation}
We compute the electronic structure of \ce{N2} in a $(N_{\rm orb}, N_e) = (8, 10)$ active space using CASSCF orbitals with the cc-pVDZ basis set, scanning the bond length $d_{\text{N-N}}$ from near-equilibrium ($d_{\text{N-N}}=1.10$  \AA) to the strongly correlated dissociation regime $(d_{\text{N-N}}=2.00$ \AA). Figure~\ref{fig:N2_MR} characterizes the growth of multireference character along the scan: the HF overlap $|\gamma_0|^2$ decreases monotonically from 0.93 to 0.33, the spectral gap $\Delta{E}_1$ shrinks from 0.46 to 0.16 Ha, and the natural-orbital (NO) occupation numbers progressively depart from their closed-shell values of 0 and 2 as the frontier orbitals become half-filled. 
Crucially, Fig.~\ref{fig:N2_MR}(c) compares the $1$-norm $\lambda$ with the exact and HF-estimated spectral ranges, $\Delta{E}_{\bm{Q}}^{\rm (exact)}$ and $\Delta{E}_{\bm{Q}}^{\rm (HF)}$, along the scan. 
The HF estimate closely tracks the exact value, with a relative deviation of at most several percent ($\sim$3.8\% at $d_{\text{N-N}}=1.10$~\AA{} and $\sim$7.9\% even at $d_{\text{N-N}}=2.00$~\AA{} in the strongly correlated dissociation regime).
More importantly, $\lambda \gg \Delta{E}_{\bm{Q}}$ throughout (with $\lambda / \Delta{E}_{\bm{Q}} \approx 5$--$7$), and since $\lambda$ is exact while only $\Delta{E}_{\bm{Q}}$ carries the HF-approximation error, this small deviation has a correspondingly small effect on the ratio $\Delta{E}_{\bm{Q}}/\lambda$ that governs the sampling cost. The HF estimate of the spectral range is therefore sufficiently accurate for MSD across the entire scan.
Separately, the reliability of the energy-level-shift step itself can be quantified by the HF-level spectrum-shift error $\delta_{\rm shift}$ [Eq.~\eqref{eq:energy_shift_error}], which measures how far the HF-shifted spectrum departs from being centered at the origin. This quantity remains below 0.13\% across the entire scan, confirming that, in addition to the accurate estimate of the spectral range discussed above, the energy-level shift remains valid throughout the dissociation.

\begin{figure*}[htbp]
  \centering
  \includegraphics[width=1.0\textwidth]{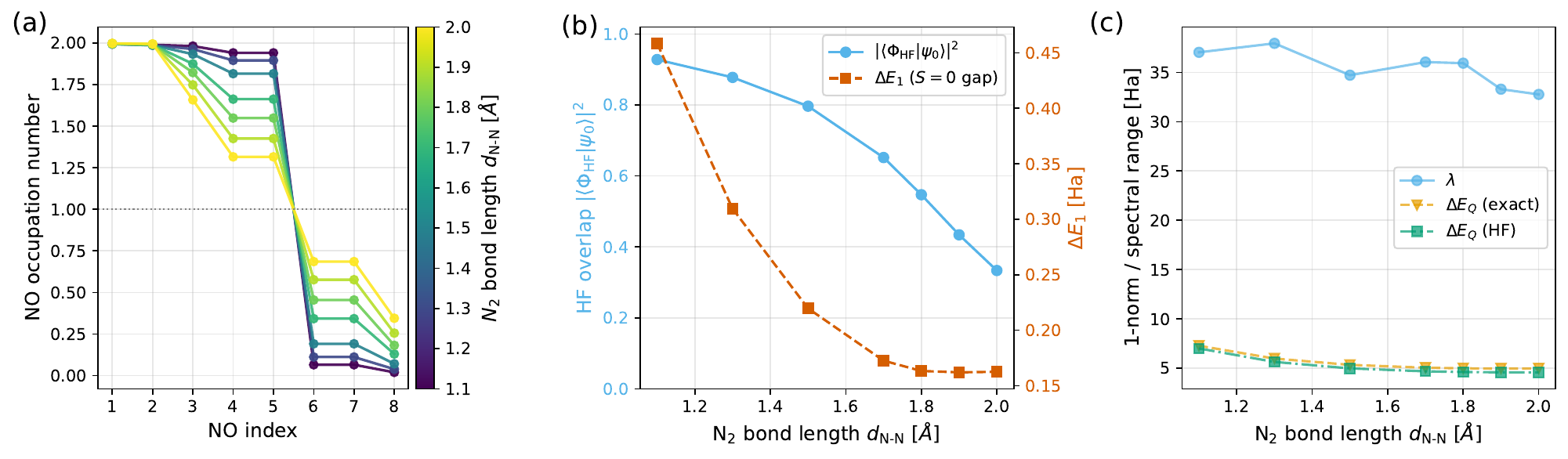}
  \caption{Growth of multireference character along the \ce{N2} dissociation coordinate [$(N_{\rm orb}, N_e) = (8, 10)$, cc-pVDZ basis, CASSCF orbitals]. (a) NO occupation numbers at each bond length $d_{\text{N-N}}$; frontier orbitals depart from the closed-shell values 0 and 2 as $d_{\text{N-N}}$ increases. (b) HF overlap $|\gamma_0|^2 = |\braket{\Phi_{\rm HF}|\psi_0}|^2$ (left axis) and $S=0$ spectral gap $\Delta{E}_1$ (right axis) versus $d_{\text{N-N}}$. 
  (c) The $1$-norm $\lambda$ together with the exact and HF-estimated spectral ranges, $\Delta{E}_{\bm{Q}}^{\rm (exact)}$ and $\Delta{E}_{\bm{Q}}^{\rm (HF)}$, versus $d_{\text{N-N}}$. The HF estimate closely tracks the exact spectral range, while $\lambda \gg \Delta{E}_{\bm{Q}}$ throughout, so that the small HF-approximation error in $\Delta{E}_{\bm{Q}}$ has only a minor effect on the cost-determining ratio $\Delta{E}_{\bm{Q}}/\lambda$. }
  \label{fig:N2_MR}
\end{figure*}
 
Figure~\ref{fig:N2_PERF} shows the corresponding MSD performance. Throughout, the MSD energy-level shift is performed using the HF-level spectral range $\Delta{E}_{\bm{Q}}^{\rm (HF)}$ rather than the exact value, so that these results reflect the performance under the practical HF estimate.
To quantify the practical benefit, we define $M^{\rm (chem)}$ as the total shot count required to reach chemical accuracy (1.6 mHa), and $R[M^{\rm (chem)}] = M^{\rm (chem)}_{\rm msd} / M^{\rm (chem)}_{\rm kqd}$ as the corresponding ratio; these are the chemical-accuracy analogues of the matrix-error quantities $M$ and $R[M]$ used in the main text. 
At a fixed Krylov dimension $n = 36$—chosen large enough that all three methods (KQD, 1-norm-optimized KQD, and MSD) converge across the full scan—MSD attains chemical accuracy with the smallest $M^{\rm (chem)}$ at every bond length. As anticipated, the absolute $M^{\rm (chem)}$ required by all three methods increases with $d_{\text{N-N}}$, reflecting the shrinking $\Delta{E}_1$ and $|\gamma_0|^2$; nevertheless, the ratio $R[M^{\rm (chem)}]$ remains below unity throughout the dissociation, i.e., the MSD advantage is preserved. 
The overall increase of $R[M^{\rm (chem)}]$ toward the dissociation limit follows the same trend as the ratio $\Delta{E}_{\bm{Q}}/\lambda$ [Fig.~\ref{fig:N2_MR}(c)], in broad agreement with the scaling of Eq.~\eqref{eq:msd_sampling_cost_optimized}. We caution, however, that a precise correspondence is not expected: the chemical-accuracy shot count $M^{\rm (chem)}$ is also influenced by factors outside the matrix-error scaling of Eq.~\eqref{eq:msd_sampling_cost_optimized}, such as the regularization (thresholding) and the conditioning of the generalized eigenvalue problem, which may account for the mild non-monotonic behavior seen at shorter bond lengths. The key point is that $R[M^{\rm (chem)}]$ remains well below unity across the entire scan, so the MSD advantage is preserved regardless of these details.

\begin{figure*}[htbp]
  \centering
  \includegraphics[width=1.0\textwidth]{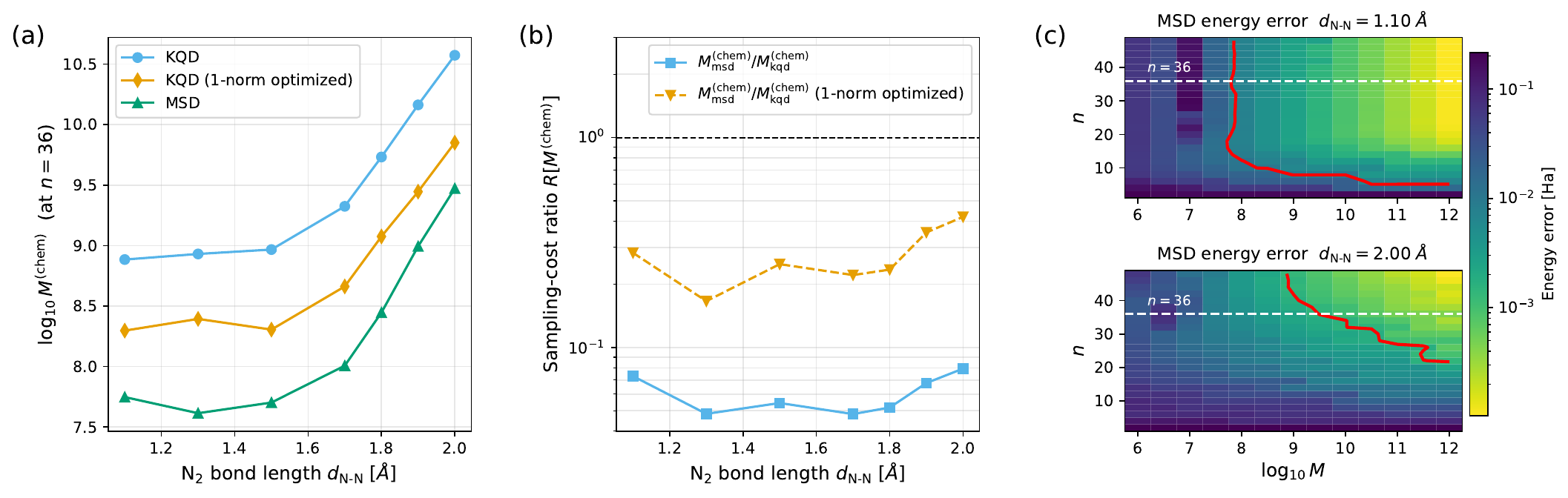}
  \caption{Sampling-cost performance of MSD along the \ce{N2} dissociation coordinate [$(N_{\rm orb}, N_e) = (8, 10)$, cc-pVDZ basis], at a fixed Krylov order $n = 36$ and central-difference degree $J = N_{\rm orb}=8$. 
  (a) Total shot count $M^{\rm (chem)}$ required to reach chemical accuracy (1.6 mHa), for KQD, 1-norm-optimized KQD, and MSD, versus bond length $d_{\text{N-N}}$. 
  (b) Ratio $R[M^{\rm (chem)}] = M^{\rm (chem)}_{\rm msd} / M^{\rm (chem)}_{\rm kqd}$ versus $d_{\text{N-N}}$; the ratio remains below unity throughout the dissociation, showing that the MSD advantage is preserved. Because all three methods share the same Krylov subspace, these ratios reflect only the sampling-cost reduction.
  (c) Energy error of MSD in the $(n, M)$ plane for two representative bond lengths ($d_{\text{N-N}} = 1.10$ and $2.00$~\AA), with the $n = 36$ operating point indicated; the chemical-accuracy (1.6 mHa) boundaries are overlaid as red solid lines.
  All values of the energy error and $M^{\rm (chem)}$ are averaged over $5{,}000$ random trials.
  }
  \label{fig:N2_PERF}
\end{figure*}

\subsection{Cr$_2$ at equilibrium} 
As a more stringent test, we study \ce{Cr2} in a $(N_{\rm orb}, N_e) = (12, 12)$ active space at its equilibrium bond length $1.68$ \AA~\cite{larsson2022chromium}, using CASSCF orbitals with the cc-pVDZ-DK basis set and the exact-two-component (X2C) scalar-relativistic Hamiltonian to account for the relativistic effects of the chromium atoms. 
This system is strongly multireference already at equilibrium: the CASSCF ground state has an HF-determinant weight of only 0.47, an effective number of determinants of 4.4 [defined as the inverse participation ratio $1/\sum_I |c_I|^4$ of the configuration-interaction coefficients $c_I$], and all twelve active natural orbitals are fractionally occupied [Fig.~\ref{fig:Cr2}(a)]. Its $S=0$ first-excitation gap $\Delta{E}_1 = 0.128$ Ha, where $S$ is the total-spin quantum number [the eigenvalue of $\hat{S}^2$ is $S(S+1)$], obtained by singlet-projected ($\braket{\hat{S}^2} = 0$) complete-active-space configuration interaction (CASCI) on the CASSCF orbitals, is even smaller than that of fully dissociated \ce{N2} (0.16 Ha), indicating a more strongly correlated regime. 
Between the singlet ground and first-excited states lie triplet and quintet states; these do not contribute to the Krylov subspace built from a pure-singlet reference, since $[\hat{H}, \hat{S}^2] = 0$, so the convergence is governed by the $S=0$ gap $\Delta{E}_1$.
We also verify that the HF estimate of the spectral range remains accurate for this strongly correlated system: $\Delta{E}_{\bm{Q}}^{\rm (HF)}$ deviates from $\Delta{E}_{\bm{Q}}^{\rm (exact)}$ by only about 4.2\%, comparable to the \ce{N2} scan, and remains far smaller than $\lambda$, so that the energy-level-shift step remains reliable in this regime.
Separately, the HF-level spectrum-shift error is $\delta_{\rm shift} = 0.012\%$, even smaller than the maximum value along the \ce{N2} scan, confirming that the energy-level-shift step itself remains reliable in this regime.

\begin{figure*}[tbp]
  \centering
  \includegraphics[width=1.0\textwidth]{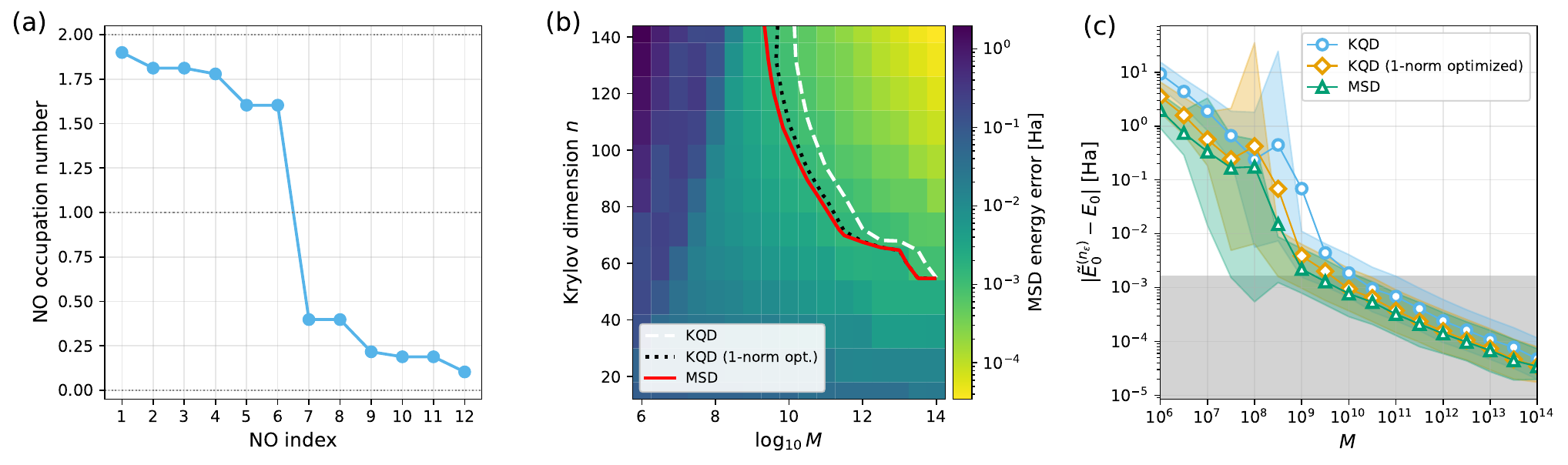}
  \caption{MSD performance for the strongly multireference \ce{Cr2} molecule [$(N_{\rm orb}, N_e) = (12, 12)$, cc-pVDZ-DK basis] at equilibrium bond length of 1.68 \AA, with central-difference degree $J = N_{\rm orb} = 12$. (a) NO occupation numbers of the CASSCF ground state; every active orbital deviates from the closed-shell values of 0 and 2, and several frontier orbitals are strongly fractional, reflecting the multireference character. (b) MSD energy-error landscape in the $(n, M)$ plane, with the chemical-accuracy (1.6 mHa) boundaries of KQD (white dashed), 1-norm-optimized KQD (black dotted), and MSD (red solid) overlaid; the MSD boundary lies at the smallest $M$ for every Krylov dimension. The energy error values are averaged over $5{,}000$ random trials. (c) Energy error versus total shot count $M$ at a fixed Krylov dimension $n = 144$; markers denote ensemble averages and shaded bands the FWHM over 5{,}000 random trials. The gray region indicates chemical accuracy.}
  \label{fig:Cr2}
\end{figure*}

Figure~\ref{fig:Cr2} presents the MSD performance for \ce{Cr2}, again using the HF-level spectral range $\Delta{E}_{\bm{Q}}^{\rm (HF)}$ for the energy-level shift.
Panel (a) shows the NO occupation numbers of the CASSCF ground state, none of which retains an integer occupation, illustrating the strong static correlation of the system. Panel (b) shows the MSD energy-error landscape in the $(n, M)$ plane, with the chemical-accuracy boundaries of the three methods overlaid; the MSD boundary lies at the smallest $M$ for every Krylov dimension, confirming its reduced sampling cost. Panel (c) shows the energy error versus $M$ at a fixed Krylov order $n = 144$, with the shaded bands denoting the full width at half maximum (FWHM) of the error distribution over 5{,}000 random trials. The large Krylov order required here ($n = 144$, far exceeding $N_{\rm orb} = 12$) directly reflects the strong correlation of the system, as anticipated above. Even in this demanding regime, MSD reaches chemical accuracy with fewer shots than both KQD and 1-norm-optimized KQD. These results demonstrate that the sampling-cost advantage of MSD is retained for genuinely strongly correlated molecular systems, consistent with the reference-state-independent scaling discussed above.


\end{document}